\documentclass[reqno,12pt,letterpaper]{amsart}
\usepackage{amsmath,amssymb,amsthm,graphicx,mathrsfs,url}
\usepackage[usenames,dvipsnames]{color}
\usepackage[colorlinks=true,linkcolor=Red,citecolor=Green]{hyperref}
\usepackage{amsxtra}
\usepackage{subcaption}

\def\?[#1]{\textbf{[#1]}\marginpar{\Large{\textbf{??}}}}

\def\smallsection#1{\smallskip\noindent\textbf{#1}.}
\let\epsilon=\varepsilon 

\newcommand{\RR}{{\mathbb R}}
\newcommand{\NN}{{\mathbb N}}
\newcommand{\CC}{{\mathbb C}}
\newcommand{\TT}{{\mathbb T}}
\newcommand{\ZZ}{{\mathbb Z}}

\newtheorem{theo}{Theorem}
\newtheorem{prop}{Proposition}[section]	
\newtheorem{defi}[prop]{Definition}

\newtheorem{lemm}[prop]{Lemma}

\newtheorem{rem}{Remark}
\numberwithin{equation}{section}

\DeclareMathOperator{\Spec}{Spec}

\let\Im=\Imag

\DeclareMathOperator{\Op}{Op}

\let\Re=\Real
\DeclareMathOperator{\sgn}{sgn}

\DeclareMathOperator{\supp}{supp}
\DeclareMathOperator{\vol}{vol}
\DeclareMathOperator{\WF}{WF}
\DeclareMathOperator{\tr}{tr}

\def\WFh{\WF_h}
\def\indic{\operatorname{1\hskip-2.75pt\relax l}}
\usepackage{scalerel}

\newcommand\reallywidehat[1]{\arraycolsep=0pt\relax%
\begin{array}{c}
\stretchto{
  \scaleto{
    \scalerel*[\widthof{\ensuremath{#1}}]{\kern-.5pt\bigwedge\kern-.5pt}
    {\rule[-\textheight/2]{1ex}{\textheight}} 
  }{\textheight} %
}{0.5ex}\\           
#1\\                 
\rule{-1ex}{0ex}
\end{array}
}

\title[Magnetic oscillations in graphene]{Magnetic oscillations in a model of graphene}

\author{Simon Becker}
\email{simon.becker@damtp.cam.ac.uk}
\address{DAMTP, University of Cambridge, Wilberforce Rd, Cambridge CB3 0WA, UK}
\author{Maciej Zworski}
\email{zworski@math.berkeley.edu}
\address{Department of Mathematics, University of California,
Berkeley, CA 94720, USA}

\begin{document}

\begin{abstract}
We consider a quantum graph as a model of graphene in constant magnetic field and describe the density of states in terms of relativistic Landau levels satisfying a Bohr--Sommerfeld quantization condition. That provides semiclassical corrections (with the magnetic flux as the semiclassical parameter) in the study of magnetic oscillations. \end{abstract}

\maketitle


\section{Introduction and statement of results}

\label{s:intr}

The purpose of this paper is to describe the 
density of states for a model of graphene in constant magnetic field
and to relate it to the {\em Shubnikov-de Haas} and {\em de Haas--van Alphen} effects.

We use a quantum graph model introduced by Kuchment--Post \cite{KP} 
with the magnetic field formalism coming from Br\"uning--Geyler--Pankrashin 
\cite{BGP}. Quantum graphs help to investigate spectral properties
of complex systems:
the complexity is captured by the graph but analytic aspects remain
one dimensional and hence relatively simple. In particular, existence of
{\em Dirac points} in the Bloch--Floquet dispersion relation -- see \S \ref{s:dir} -- follows from a straightforward computation. This should be compared with the subtle study by Fefferman--Weinstein \cite{FW1} which starts with 
periodic Sch\"odinger operators on $ \RR^2$. 

\begin{figure}
\includegraphics[height=5.65cm]{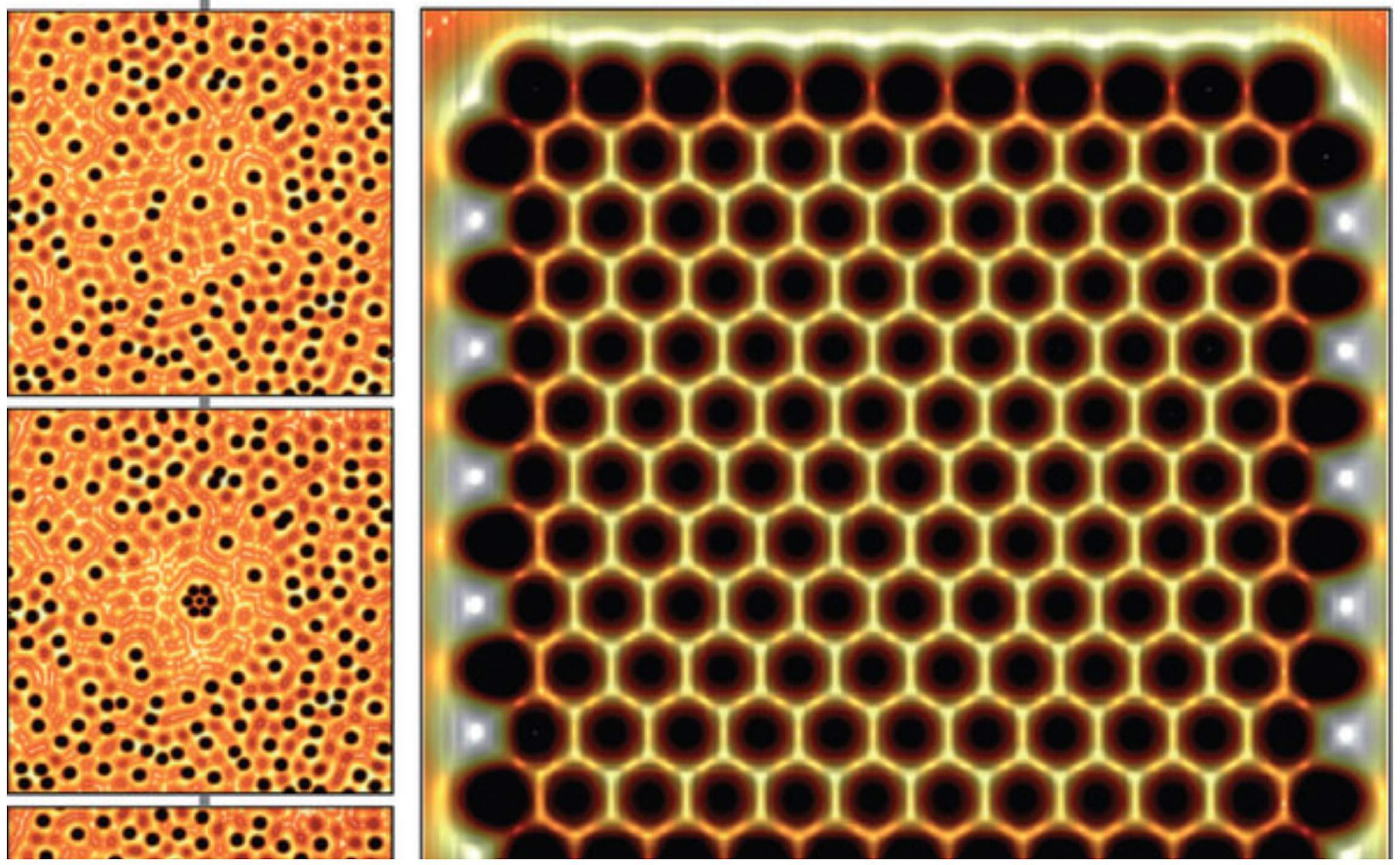}
\caption{A molecular graphene \cite{hari} 
in which the CO molecules confine the Bloch electron to a one dimensional hexagonal structure. \label{fig:hari}}
\end{figure}

One experimental setting for which quantum graphs could be a reasonable model is {\em molecular graphene} studied by the Manoharan group \cite{hari} -- see also \cite{PoMa} for a general discussion. In that case CO molecules placed on a copper plate confine 
the electrons to a one dimensional hexagonal structure -- see Figure
 \ref{fig:hari}.

The ideas behind rigorous study of the density of states and of magnetic oscillations come from the works of Helffer--Sj\"ostrand 
\cite{HS0},\cite{HS1},\cite{HS20},\cite{HS2},\cite{BGHKS} (to which we refer for background and additional references). However, the simplicity of our model allows us to give an essentially self-contained presentation. 

\begin{figure}
\includegraphics[height=6.5cm]{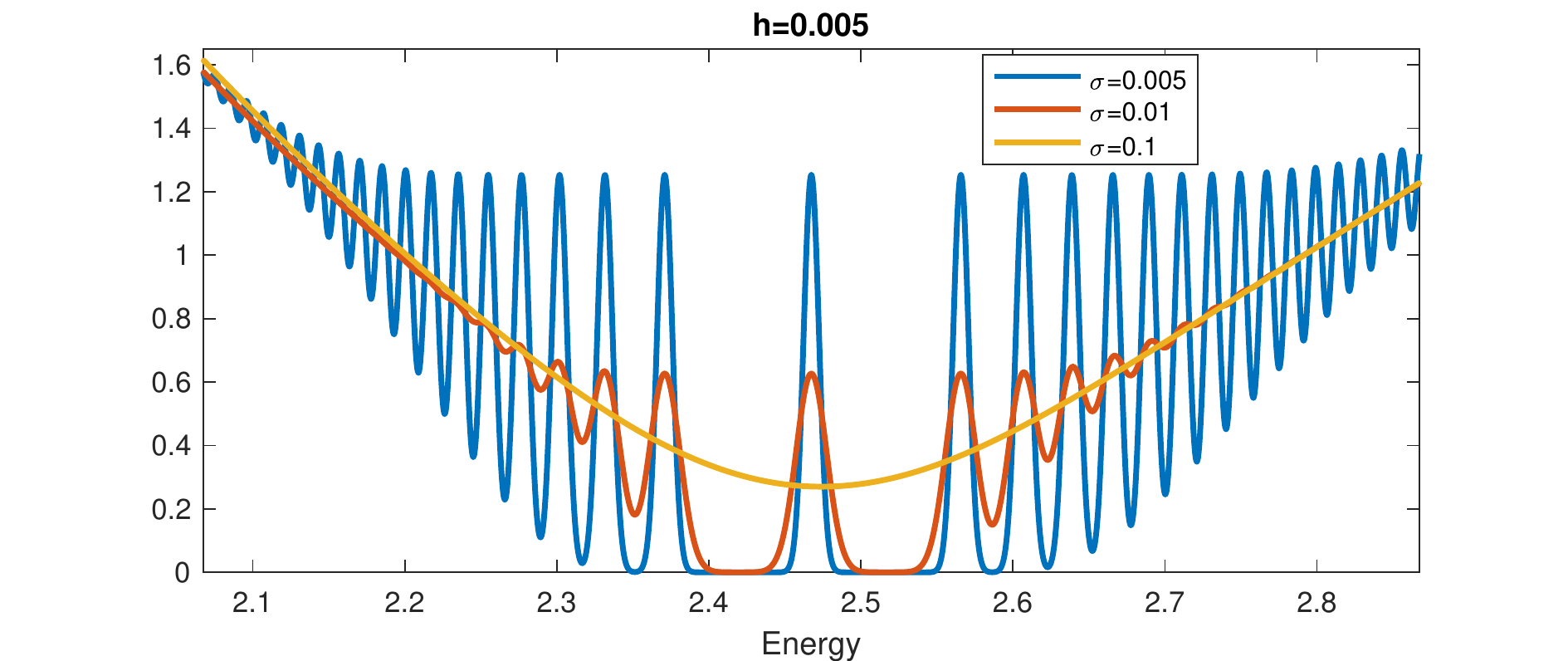}
\caption{\label{fig:DOS1} Graphs of smoothed-out density of states
 $ \mu \mapsto \widetilde \rho_B ( \exp ( ( \bullet - \mu)^2 /2 \sigma^2 ) / \sqrt{ 2 \pi} \sigma )  $ for $ h = 0.005$ and for  different values of $ \sigma $ (using an approximation \eqref{eq:Diracmeasures}). When $ \sigma $ is large 
 the function $ x \mapsto \exp ( ( x - \mu)^2 /2 \sigma^2 ) / \sqrt{ 2 \pi} \sigma ) $ is uniformly smooth and we see no oscillations as predicted by 
 \eqref{eq:smooth}. Figure \ref{fig:DOS4} compares these graphs with 
 the ones based on the ``perfect cone" approximation \eqref{eq:physde}; see
 also Figure \ref{Fig:DOS} for the density of states of the zero magnetic field.}
\end{figure}

The main object of our study is the density of state (DOS) for a magnetic Hamiltonian, $ H_B $, on a hexagonal quantum graph -- see \eqref{magop}.
The DOS is defined as a non-negative distribution 
$ \rho_B \in {\mathscr D' }( \RR ) $ (that is, a measure) produced by a renormalized trace: for $ f \in C_{\rm{c}} ( \RR) $,
\begin{equation}
\label{eq:DOS}
\rho_B ( f ) = \widetilde{\operatorname{tr}}(f(H^B)) := 
\lim_{ R \to \infty } \frac{ \tr \indic_{ B ( R ) } f ( H_B ) }{ \vol(B (R )) } , \ \ B ( R ) := \{ x \in \RR^2 : |x| < R \} ,
\end{equation}
see Definition \ref{de:DOS}. 

Our desription of the density of states comes close to formal expressions for the density of states $\rho_B $
given in the physics literature,
\begin{gather}
\label{eq:physde}
\begin{gathered} \rho_B  ( E ) = \frac{ B } {  \pi } \sum_{ n \in \ZZ }
\delta ( E - E_n ) , \ \ E_n := {\rm{sign}} (n ) v_F \sqrt{ |n|  B } ,\\
v_F = \text{ Fermi velocity } , \ \ B = \text{ strength of the magnetic field}, 
\end{gathered}
\end{gather}
see for instance \cite[(42)]{cau} or \cite[(4.2)]{sgb}. The energies $ E_n $ 
are the (approximate) relativistic Landau levels. 

Theorem \ref{theorem3} gives the following rigorous version of \eqref{eq:physde}. It is convenient to consider a semiclassical parameter given by the magnetic flux through a cell of the hexagonal lattice, see 
Figure \ref{fig:1}:
\[  h:= \frac{3\sqrt{3}} 2 B = | b_1 \wedge b_2 | B .\]
Then, if $ I $ is a neighbourhood of a Dirac energy (see \S 
\ref{s:dir}) and 
$ f \in C_{\rm{c}}^\alpha ( I )$, $ \alpha > 0 $,
\begin{equation}
\label{eq:rhoBh}  \rho_B (f ) = \frac h {\pi|b_1 \wedge b_2| } \sum_{ n\in \ZZ } f ( z_n ( h ) ) + 
\mathcal O  (\| f \|_{ C^\alpha } h^\infty) , \ \ \ h \to 0 , 
\end{equation}
where $ z_n ( h ) $ satisfy natural quantization conditions 
\eqref{eq:g2F} and \eqref{eq:tracef}. They are approximately given by 
$ z_D + E_n$ with $ E_n$'s in \eqref{eq:physde}, where $ z_D $ is the Dirac energy.  This simple asymptotic formula should be contrasted with the complicated structure of the spectrum of $ H^B $ -- see
the analysis by Becker--Han--Jitomirskaya \cite{BHJ17}. 

\begin{figure}
  \centering
   \includegraphics[width=11cm]{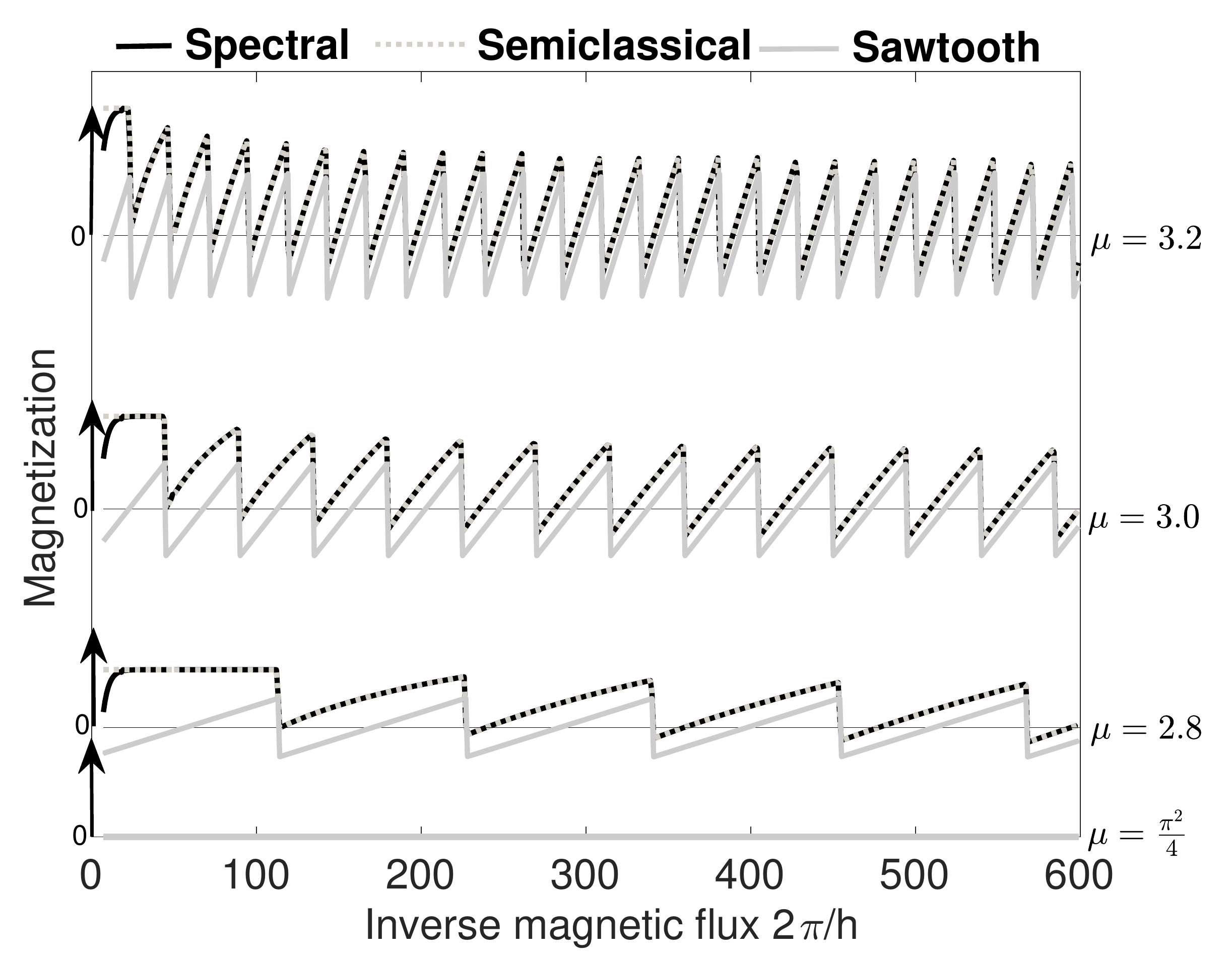} 
\caption{\label{Fig:spec} Plots of different approximations of 
the magnetization localized to the upper cone: the spectral one using
$ M_\infty $ defined in \eqref{eq:grand} (see \S \ref{spectral}), the semiclassical approximation 
$ m_\infty $ given in \eqref{eq:formal} and the sawtooth approximation 
\eqref{eq:sawtooth1}.  The agreement of $ M_\infty $ and $ m_\infty $ is remarkable even for relatively large values $ h $. The sawtooth approximation gives the correct oscillations but with amplitude errors
$ \mathcal O (\sqrt h ) $.}
\end{figure}

The importance of considering functions which are not smooth is their appearance in condensed matter calculations -- see \S 7. Oscillations as functions of $ 1/B $ are not seen for smooth functions in view of 
Theorem \ref{t:smooth}:
\begin{equation}
\label{eq:smooth} \rho_B (f ) \sim \sum_{j=0}^\infty  A_j ( f ) h^j , \ \ A_0 ( f ) = 
\rho_0 ( f ) , \ \ h \to 0 , \ \ 
f \in C^\infty_{\rm{c}} ( I ) .\end{equation}
Roughly speaking, this expansion follows from the expansion of the Riemann
sum given by \eqref{eq:rhoBh} -- see \cite{HS2}. Here the proof follows
\cite[Chapter 8]{D-S}. 

Many physical quantities are computed using DOS, in particular  grand-canonical potentials and magnetizations at temperature $ T = 1/\beta$
localized to an energy interval using a function $ \eta $:
\begin{equation}
\label{eq:grand}
\begin{gathered}
\Omega_\beta ( \mu , h ) := \rho_B ( \eta ( \bullet ) f_\beta ( \mu - \bullet) ) , \ \  f_\beta ( x )  := - \beta^{-1} \log ( e^{ \beta x } +  1) , \\
M_{\beta}(\mu, h): = - \left\lvert b_1 \wedge b_2 \right\rvert \frac{\partial }{\partial h} \Omega_{\beta}(\mu, h).
\end{gathered}
\end{equation}
For $ \beta = \infty $ we take $ f_\infty = x_+ $ which is a Lipschitz
function. Hence \eqref{eq:rhoBh} applies and for $ \beta > h^{-M_0 }$
one could also obtain expansions for $ M_\beta $ -- see \S \ref{dHvA}. 
We take a simpler approach and calculate a semiclassical approximation,
$ m_\infty ( \mu, h ) $, to 
magnetization -- see \eqref{eq:formal} and compare it to (almost) exact 
spectral calculations -- see \S\S \ref{spectral},\ref{compa}. The agreement 
between $ M_\infty $ computed spectrally and the approximation $ m_\infty $
is remarkable already at fairly high values of the magnetic field. In Theorem \ref{t:sawtooth} we derive a simple ``sawtooth" approximation for
$ m_\infty $ confirming approximations seen in the physics literature
\cite{sgb},\cite{CM01}:
\begin{equation}
\label{eq:sawtooth1}
\begin{gathered}
m_\infty ( \mu, h )  
= \frac{1}{\pi} \sigma \left( \frac{g ( \mu)} h\right) 
\frac{ g ( \mu) }{ g' ( \mu ) } + \mathcal O ( h^{\frac12} ), 
\\ \sigma (y ) :=   y - [y] - \tfrac12 .
\end{gathered}
\end{equation}
The function $ g $ comes from the dispersion relation for the quantum graph model of graphene \cite{KP} (see \S \ref{s:dir}):
\[  g ( \mu ) :=  \frac{1}{4 \pi} \int_{ \gamma_{ \Delta (\mu)^2 } } \xi  d x , \ \ 
\gamma_\omega = \left\{ ( x, \xi) \in 
\RR^2/ 2 \pi \ZZ^2 : \frac{| 1 + e^{ix } + e^{i\xi} |^2}9 = \omega \right\}, \]
where $ \Delta ( \mu ) $ is the Floquet discriminant of the potential on the edges (and is equal to $ \cos \sqrt \lambda $ for the zero potential). 
The Dirac energy, $ z_D $, for a given band is determined by $ z_D = \Delta|_{ B_k }^{-1} ( 0 ) $.

\smallsection{Notation}
We write $ f_\alpha =                                                                  
\mathcal O_\alpha ( g )_H $ for $ \| f \|_H \leq C_\alpha g $, 
that is we have a bound with constants depending on $ \alpha $. In particular,
$ f = \mathcal O ( h^\infty )_H  $ means that for any $ N $ there exists
$ C_N $ such that $ \|f \|_H \leq C_N h^N $. 
We denote $\langle x\rangle:=\sqrt{1+|x|^2}$. 


\smallsection{Acknowledgements} 
We gratefully acknowledge support by
the UK Engineering and Physical Sciences Research Council (EPSRC) grant EP/L016516/1 for the University of Cambridge Centre for Doctoral Training, the Cambridge Centre for Analysis (SB),  
by the National Science Foundation under the grant
DMS-1500852 and by the Simons Foundation (MZ). We would also like to thank
Nicolas Burq for useful discussions, Semyon Dyatlov for help with MATLAB coding and insightful comments and Hari Manoharan
for introducing us to molecular graphene and for allowing us to use
Figures \ref{fig:hari} and \ref{Fig:DOS}(B).

\section{Hexagonal quantum graphs}

Quantum graphs provide a simple model for a graphene-like structure 
in which many features can be rigorously derived with minimal technical effort. Hence we consider a hexagonal graph,  $ \Lambda $, with Schr\"odinger operators defined on each edge \cite{KP}.
The graph $\Lambda$ is obtained by translating its fundamental cell $W_{\Lambda}$, consisting of vertices
\begin{equation}
\label{eq:r0r}
r_0:=(0,0) , \ \ \  r_1:=\left(\tfrac{1}{2},\tfrac{\sqrt{3}}{2}\right)
\end{equation}
and edges
\begin{equation}
\label{eq:fgh}
\begin{split}
f & :=\operatorname{conv}\left(\left\{r_0,r_1\right\}\right) \ \backslash \ \left\{r_0,r_1\right\},  \\ 
g & :=\operatorname{conv}\left(\left\{r_0,\left(-1,0\right)\right\}\right) \ \backslash \left\{r_0,\left(-1,0\right)\right\} , \\
h& :=\operatorname{conv}\left(\left\{r_0,-r_1\right\}\right) \ \backslash \ \left\{r_0,-r_1\right\},
\end{split}
\end{equation}
along the basis vectors of the lattice.
The basis vectors are
\begin{align}
b_1:= \left(\tfrac{3}{2}, \tfrac{\sqrt{3}}{2} \right) \  \text{ and } \ 
b_2:= \left(0,\sqrt{3}\right)
\end{align}
and so the hexagonal graph $\Lambda \subset \mathbb{R}^2$ is given by the range of a $\mathbb{Z}^2$-action on the fundamental domain $W_{\Lambda}$
\begin{equation}
\label{lattice}
\Lambda:= \left\{ x \in \mathbb{R}^2; x = \gamma_1b_1+\gamma_2b_2+[x] \text{ for } \gamma \in \mathbb{Z}^2 \text{ and } [x] \in W_{\Lambda} \right\}.
\end{equation}
\begin{figure}
\includegraphics[height=7cm]{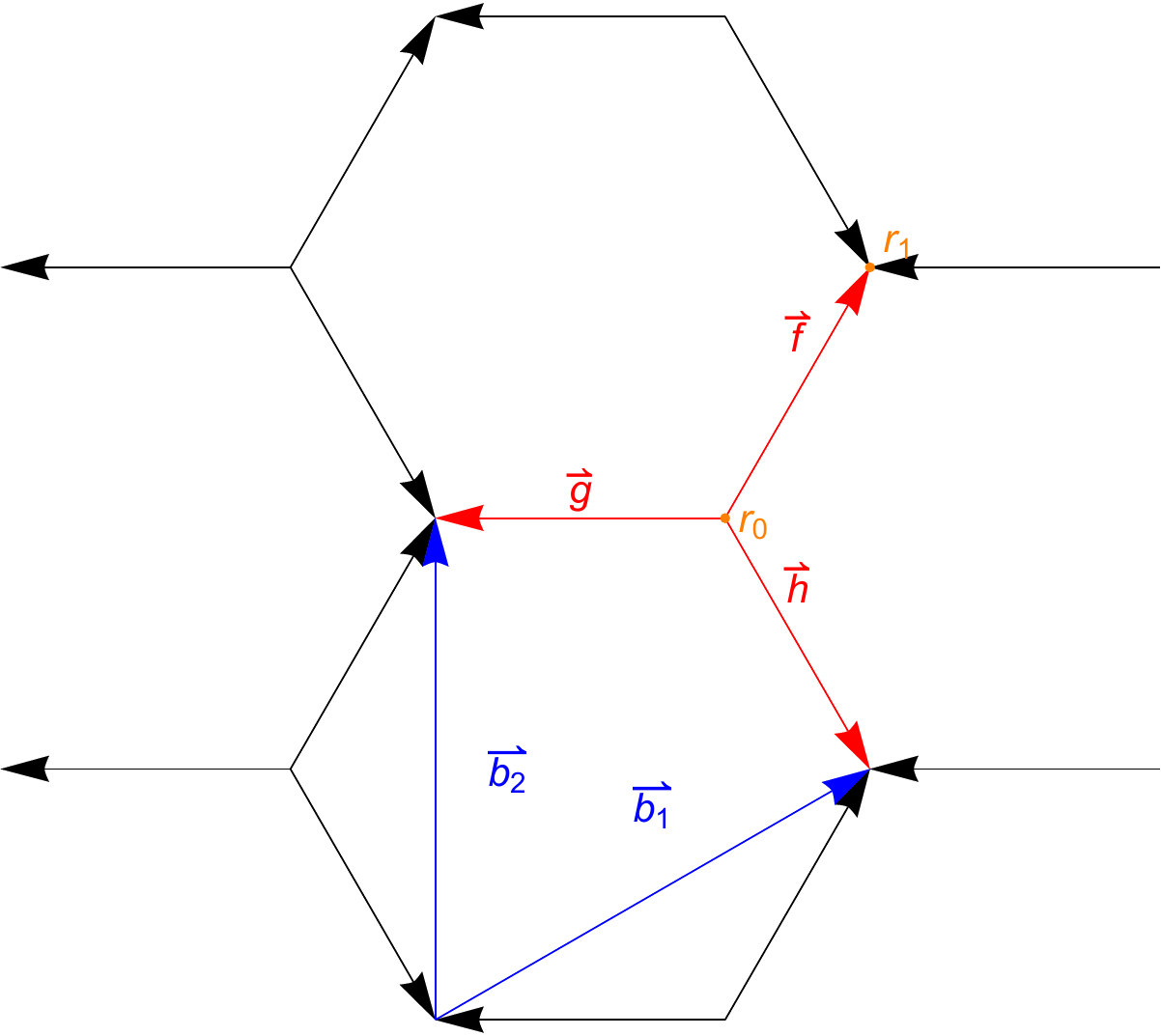} 
\caption{The fundamental cell and lattice basis vectors of $\Lambda$. \label{fig:1}}
\end{figure}
The set of edges of $\Lambda$ is denoted by $
\mathcal E = \mathcal{E}(\Lambda),$ the set of vertices by $ \mathcal V = \mathcal{V}(\Lambda)$ and the set of edges adjacent to a given vertex $v \in \mathcal{V}$ by $\mathcal{E}_v.$  We drop $ \Lambda $ in the notation if no 
confusion is likely to arise. 

When we say that $ u \in C ( \Lambda ) $ 
we mean that $ u $ is a continuous function on $ \Lambda $, a closed subset of  $\RR^2 $ -- see \eqref{lattice}.


For any edge $e \in \mathcal{E} $ we denote by $[e] \in \mathcal{E}(W_{\Lambda})$ the unique edge (thought of as a vector in $ \RR^2 $) for which there is $\gamma \in \mathbb{Z}^2$ such that
$ e=\gamma_1b_1+\gamma_2b_2+[e]$.
We impose a global orientation on the graph by orienting the edges in terms of {\em initial} and {\em terminal} maps
\begin{equation*}
i :\mathcal{E} \rightarrow \mathcal{V},  \ \ \ 
t :\mathcal{E} \rightarrow \mathcal{V}
\end{equation*}
where $i$ and $t$ map edges to their initial and terminal ends.
It suffices to specify the orientation on the fundamental domain
\[ 
i(f) =i(g)=i(h)=r_0 , \ \
t(f) =r_1, \ t(g)=r_1-b_1, \ \ t(h)=r_1-b_2.
\] 
For arbitrary $e \in \mathcal{E} $, we then extend those maps by
\begin{align}
i(e):=\gamma_1b_1+\gamma_2b_2+i([e]) \text{ and }
t(e):=\gamma_1b_1+\gamma_2b_2+t([e]).
\end{align}

In the case of our special graph with orientations showed in Figure~\ref{fig:1} a given vertex is either initial or terminal and hence we wrote
\begin{equation}
\label{eq:defV} 
\mathcal V = \mathcal V^i \sqcup \mathcal V^t, \ \ 
\mathcal V^\bullet := \{ v : \text{ $ v = \bullet ( e) $ for some $ e \in \mathcal E$} \}, \ \ \bullet = i,t .
\end{equation}

The fundamental domain of the dual lattice can be identified, because the lattice is spanned by a $\mathbb{Z}^2$-action, with the dual $2$-torus
\begin{equation}
\label{eq:dualT}
\mathbb{T}^2_*:= \mathbb{R}^2 / (2 \pi \mathbb{Z})^2.
\end{equation}
We assume every edge $e \in \mathcal{E} $ is of length one and has a standard chart
\begin{equation}
\label{chart}
\kappa_e: e \rightarrow (0,1), \ \ \ \kappa_e (ti(e)+(1-t)t(e)) = t.
\end{equation}  
Thus, for $n \in \mathbb{N}_0$, the Sobolev space $H^n(\mathcal{E})$ on $\Lambda$ is given by the Hilbert space direct sum
\begin{equation}
H^n(\mathcal{E}):= \bigoplus_{e \in \mathcal{E}} H^n(e).
\end{equation}
On edges $e \in \mathcal{E}$ we define the maximal Schr\"odinger operator 
\begin{equation}
\label{maximaloperator}
H_e: H^2(e) \subset L^2(e) \rightarrow L^2(e) , \ \ \
H_e \psi_e := -\psi''_e+V \psi_e
\end{equation}
with potential $V  \in L^2 ((0,1)) \simeq L^2 ( e ) $ which is the same on every edge and even with respect to the edge's centre. 

\subsection{Relation to Hill operators}
Using the potential introduced in the previous section, we define the $\mathbb{Z}$-periodic Hill potential $V_{\operatorname{per}} \in L^2_{\text{loc}}(\mathbb{R})$
\begin{equation}
\label{Hillpotential}
V_{\text{per}}(x + n ):=V (x ), \ \ n \in \ZZ, \ x \in ( 0 , 1 ) . 
\end{equation}
Next, we study the associated self-adjoint Hill operator on the real line
\begin{equation*}
H_{\text{per}} : H^2(\mathbb{R}) \subset L^2(\mathbb{R}) \rightarrow L^2(\mathbb{R})  \ \ 
H_{\text{per}} \psi:=-\psi''+ V \psi.
\end{equation*}
There are always two linearly independent solutions $c_{\lambda},s_{\lambda} \in H^2_{\text{loc}}(\mathbb{R})$ to $
H_{\text{per}}\psi= \lambda \psi
$
satisfying 
\begin{equation}
\label{eq:csla}
\text{ $ c_{\lambda}(0) =1, \  c_{\lambda}'(0) =0 \ $ and $ \ s_{\lambda}(0)=0, \  s_{\lambda}'(0)=1$.}
\end{equation}
Consider the Dirichlet operator on $(0,1)$
\begin{equation*}
\Lambda^D_{(0,1)}: H_0^1(0,1) \cap H^2(0,1) \subset L^2(0,1) \rightarrow L^2(0,1) \ \ \Lambda^D_{(0,1)}\psi=-\psi''+V_{\text{per}}\psi.
\end{equation*}
Any function $\psi_{\lambda} \in H^2(0,1)$ satisfying $-\psi_{\lambda}'' +V_{\text{per}} \psi_{\lambda} = \lambda \psi_{\lambda}$ with  $\lambda \notin \operatorname{Spec}(\Lambda^D_{(0,1)} )$ (that is with
 $s_{\lambda}(1) \neq 0$) can be written as a linear combination of $s_{\lambda}, c_{\lambda}$:
\begin{equation}
\label{representsol}
\psi_{\lambda} (t) = \frac{\psi_{\lambda}(1)-\psi_{\lambda}(0)c_{\lambda}(1)}{s_{\lambda}(1)}s_{\lambda}(t)+\psi_{\lambda}(0) c_{\lambda}(t).
\end{equation}

For $\lambda \notin \operatorname{Spec}(\Lambda^D_{(0,1)} )$, we define the Dirichlet-to-Neumann map
\begin{equation}
\label{DtN}
m(\lambda) := \frac{1}{s_{\lambda}(1)} \left( \begin{matrix} -c_{\lambda}(1) & 1 \\ 1 & -s'_{\lambda}(1)  \end{matrix} \right), \ \ 
\left( \begin{matrix} \ \ \psi_{\lambda}'(0) \\ -\psi_{\lambda}'(1) \end{matrix} \right) = m(\lambda) \left( \begin{matrix} \psi_{\lambda}(0) \\ \psi_{\lambda}(1) \end{matrix} \right).
\end{equation}
\begin{rem}
Since $V_{\rm{per}}$ is assumed to be symmetric with respect to $ \frac{1}{2}$ on the interval $(0,1)$,  $c_{\lambda}(1)=s'_{\lambda}(1)$. The Dirichlet eigenfunctions are consequently either even or odd with respect to $\frac{1}{2}$.
\end{rem}

The monodromy matrix associated with $H_{\text{per}}$ is the matrix valued entire function of $\lambda$:
\begin{equation*}
\mathcal M (\lambda):= \left( 
\begin{matrix}
  c_{\lambda}(1) & s_{\lambda}(1)  \\
  c_{\lambda}'(1) & s_{\lambda}'(1)  
 \end{matrix}\right).
 \end{equation*}
Its normalized trace 
\begin{equation}
\label{eq:Floq}
\Delta({\lambda}):= \frac{\operatorname{tr}(\mathcal M(\lambda))}{2}
\end{equation}
is called the Floquet discriminant. 
In the case when $ V \equiv 0 $ we have 
\begin{equation}
\label{eq:Floqf}
\Delta ( \lambda ) = \cos \sqrt \lambda ,
\end{equation}
and this will serve as an example throughout the paper.

The spectrum of the Hill operator, $ H_{\rm{per}} $ is purely absolutely continuous spectrum and is given by 
\begin{equation}
\label{Hillbands}
\begin{gathered}
\operatorname{Spec}(H_{\text{per}})= \left\{\lambda \in \mathbb{R} \, : \, \left\lvert \Delta(\lambda) \right\rvert \le 1 \right\}=\bigcup_{n=1}^{\infty} B_n \\
B_n:=[\alpha_n,\beta_n], \ \ \beta_n \le \alpha_{n+1}, \ \ 
\Delta'\vert_{\operatorname{int}(B_n)}\neq 0 , 
\end{gathered}
\end{equation}
see \cite[\S XIII]{RS4}.

\section{Magnetic Hamiltonians on quantum graphs}

The vector potential $\textbf{A} $ is a one form on $ \RR^2 $ and the magnetic field is given by $ \textbf{B} = d \textbf{A} $. For a homogeneous magnetic field
\begin{equation}
\textbf{B} : =B \ dx_1 \wedge dx_2
\end{equation}
we can choose a symmetric gauge, that is $ \mathbf A $ given as follows:
\begin{equation}
\textbf{B} = d\textbf{A} , \ \ \ \textbf{A}=\tfrac12 {B} \left(-x_2 \, dx_1 +x_1 \, dx_2 \right).
\end{equation}

The scalar vector potential $A_e \in C^{\infty}(e)$ along edges $e \in \mathcal{E}$ is obtained by evaluating the form on the graph along the vector field generated by the respective edge $[e]$:
\begin{equation}
\begin{split}
A_{e}(t)&:= \textbf{A}\left(i(e)+t[e] \right)\left([e]_1 \partial_1 + [e]_2 \partial_2 \right) \\
&=\textbf{A}\left(i(e)\right)\left([e]_1 \partial_1 + [e]_2 \partial_2 \right)+ \underbrace{tA([e])\left([e]_1 \partial_1 + [e]_2 \partial_2 \right)}_{=0} \\
&=\textbf{A}\left(i(e)\right)\left([e]_1 \partial_1 + [e]_2 \partial_2 \right)
\end{split}
\label{eq:Aet}
\end{equation}
which is constant along any single edge.

In terms of the magnetic differential operator $(D^B \psi)_e:= - i \psi_e' - A_{e}\psi_e $, the \newline
Schr\"odinger operator modeling graphene in a magnetic field becomes
\begin{equation}
\label{magop}
H^B :D(H^B) \subset L^2(\mathcal{E}) \rightarrow L^2(\mathcal{E}) , \ \ (H^B\psi)_{{e}}:=(D^B D^B \psi)_{{e}}+V\psi_{e},
\end{equation}
where $ D(H^B) $ is defined as the set of $ \psi \in {H}^2(\mathcal{E}) $ satisfying
\begin{equation*}
\begin{split}
&  
\psi_{e_1} ( v ) = \psi_{e_2} ( v) , \  {e}_1,{e}_2 \in \mathcal{E}_{v} ,  \ \ \ 
\sum_{{e} \in \mathcal{E}_{v}}  \left(D^B\psi\right)_{{e}}(v) 
 = 0 .
\end{split}
\end{equation*}
 \begin{rem}
The Hamiltonian $H^B$ for any magnetic field with constant flux per hexagon is unitarily equivalent to the setting of a constant magnetic field with the same flux per hexagon. 
\end{rem}

The unitary Peierls' substitution is the multiplication operator
\begin{equation}
\label{Peierl}
P:  L^2(\mathcal{E}) \rightarrow L^2(\mathcal{E}) , \ \
\psi_e ( t )  \mapsto e^{i A_{e} t }  \psi_e ( t) , \ \ 
t \in ( 0,1 ) . 
\end{equation} 
The operator $ P $
transforms $H^B$ into
\begin{equation}
\label{LambdaB}
\Lambda^B:=P^{-1} H^B P ,  \ \ 
(\Lambda^B \psi )_ e = - \psi_e'' +V \psi_{e} .
\end{equation}
The domain of $\Lambda^B$ consists of $ \psi \in H^2(\mathcal{E})$
such that, in the notation of \eqref{eq:defV},
\begin{equation*}
\begin{split} 
&  v \in 
 \mathcal V^i \ \Longrightarrow \  \psi_{e_1}(v)=\psi_{e_2}(v) , 
 \  \ e_1, e_2 \in \mathcal E_v, \ \ \ \sum_{e \in \mathcal{E}_{v}} \psi'_e(v) =0, \\
&  v \in 
 \mathcal V^t \ \Longrightarrow \  e^{  iA_{e_1}}\psi_{e_1}(v)=e^{ iA_{e_2} }\psi_{e_2}(v), \  \ e_1, e_2 \in \mathcal E_v, 
 \ \ \ \sum_{e \in \mathcal{E}_{v}} e^{ iA_{e}}\psi'_e(v) =0.  
\end{split}
\end{equation*}

Thus, the problem reduces to the study of non-magnetic Schr\"odinger operators with the magnetic field moved into the boundary conditions.
We note that the magnetic Dirichlet operator, 
\begin{equation}
\label{magopd}
H^{D} :\bigoplus_{e \in \mathcal{E}(\Lambda)}\left(H_0^1(e)\cap H^2(e) \right)  \rightarrow L^2(\mathcal{E}) , \ \ 
(H^{D} \psi)_e:=(D^B D^B \psi)_e+V_e\psi_e, 
\end{equation}
is (using Peierls' substitution \eqref{Peierl}) unitarily equivalent to the Dirichlet operator without magnetic field
\begin{equation}
\label{LambdaD}
\Lambda^D:= \bigoplus_{e \in \mathcal{E}(\Lambda)} \Lambda^{D}_e = P^{-1}H^{D}P, 
\end{equation}
where $ \Lambda^D_e $ is the Dirichlet realization of $ - \partial_t^2 + V_e $ on $ e $. 
Thus, the spectrum of the Dirichlet operator does not change under magnetic perturbations.

\subsection{Effective Hamiltonian}

We now follow  Pankrashin \cite{P} and Br\"uning--Geyler--Pankrashin\cite{BGP} and use the Krein resolvent formula to reduce the operator $\Lambda^B$ into a term containing only parts of the Dirichlet spectrum and an effective operator that will be further investigated afterwards. We will find that the contribution of Dirichlet eigenvalues to the spectrum of $H^B$ is fully explicit and thus we will be left with an effective operator which will be used to describe the density of states.

We define 
\begin{equation*}
\textbf{H}:D(\textbf{H}) \subset L^2(\mathcal{E}) \rightarrow L^2(\mathcal{E}) , \ \ \ 
(\textbf{H}\psi)_e:=(D^B  D^B \psi)_e+V_e\psi_e
\end{equation*}
where $ D ( \textbf{H} ) $ consists of $ \psi \in H^2(\mathcal{E}) $ satifying  (using notation of \eqref{eq:defV})
\begin{equation}
\label{eq:domH}
\begin{split}
&v \in \mathcal V^i  \ \Longrightarrow \ 
\psi_{e_1}(v)=\psi_{e_2}(v) , \  \ e_1, e_2 \in \mathcal E_v,  \\ 
& v \in \mathcal V^t \ \Longrightarrow \   e^{  iA_{e_1}}\psi_{e_1}(t(e_1))=e^{ iA_{e_2} }\psi_{e_2}(t(e_2)) , \ \ \ e_1, e_2 \in \mathcal E_v.
\end{split}
\end{equation}
With this domain $\textbf{H} $ is a closed operator.

Then, we consider the map $\pi:D(\textbf{H})\rightarrow \ell^2(\mathcal{V})$ defined by
\begin{equation}
\label{eq:defpi}
\pi(\psi)(v):= \left\{ \begin{array}{ll}
\ \  \psi_{e}(v), &  v \in \mathcal V^i , \  \  e \in \mathcal E_v, \\
 e^{ iA_{e}} \psi_{e}(v), & v \in \mathcal V^t,  \  \  e \in \mathcal E_v . 
\end{array} \right.
\end{equation}
The operator $ \pi $ is well defined because of \eqref{eq:domH} and 
is an isomorphism from $\operatorname{ker}(\textbf{H} - \lambda)$ onto $\ell^2(\mathcal{V})$ for any $\lambda \notin \Spec (\Lambda^D)$.
This leads to the definition of the gamma-field 
\begin{equation}
\label{gfield}
\gamma: \complement \Spec(\Lambda^D) \rightarrow \mathcal L\left(\ell^2(\mathcal{V}),D(\textbf{H})\right) , \ \ \ 
\gamma(\lambda):= \left(\pi \vert_{\operatorname{ker}(\textbf{H} - \lambda )} \right)^{-1}.
\end{equation} 
In the notation of \eqref{eq:csla},  the gamma-field is given by 
\begin{equation}
\label{eq:defga}
(\gamma(\lambda)z)_{e}(t) = \frac{\left(s_{\lambda}(1)c_{\lambda}(t)-s_{\lambda}(t)c_{\lambda}(1)\right)z(i(e))+e^{-iA_e}z(t(e))s_{\lambda}(t)}{s_{\lambda}(1)}.
\end{equation}
Using this we can then state Krein's formula from \cite{P} and \cite{BGP}. For that 
we define
\begin{equation}
\label{eq:defM}
M(\lambda) :=  {s_{\lambda}(1)}^{-1} ( {K_{\Lambda}-\Delta(\lambda)}) 
\end{equation}
where 
\begin{equation}
\label{discreter}
(K_\Lambda z)(v):=\tfrac{1}{3} \left(\sum_{e \in \mathcal{E}, i(e)=v } e^{-i  A_e} z(t(e)) +  \sum_{e \in \mathcal{E}, t(e)=v } e^{iA_e} z(i(e)) \right)
\end{equation} 
defines an operator on $\ell^2(\mathcal{V})$ with $\left\lVert K_{\Lambda} \right\rVert \le 1.$

\begin{prop}[Krein's resolvent formula]
\label{KRF}
Let $ \Lambda^B $ and $ \Lambda^D $ be given by \eqref{LambdaB} and 
\eqref{LambdaD} respectively. 
For $\lambda \notin \Spec(\Lambda^D)\cup \Spec(\Lambda^B)$ the operator $ M(\lambda)$ is invertible and satisfies
\begin{equation}
\label{Kreinresolvform}
(\Lambda^B-\lambda)^{-1} = (\Lambda^D-\lambda)^{-1}  - \gamma(\lambda) M ( \lambda)^{-1} \gamma(\overline{\lambda})^*,
\end{equation}
where $ M (\lambda ) $ is given by \eqref{eq:defM}.
\end{prop}

As a consequence of \eqref{Kreinresolvform} we see that 
\[ \operatorname{Spec}(\Lambda^B)\backslash \operatorname{Spec}(\Lambda^D) = \{\lambda \in \complement \Spec(\Lambda^D); 0 \in \operatorname{Spec}(M(\lambda))\}.\] 
If $\lambda \notin \Spec(\Lambda^D)$  it follows that $\gamma(\lambda)\operatorname{ker}(M(\lambda)) = \operatorname{ker}(\Lambda^B-\lambda).$ This implies that both null-spaces are of equal dimension.

\begin{rem}
The general theory of spectral triples gives the following formula for the 
derivative for $ M $, 
\begin{equation}
\label{eq:derM}  \partial_\lambda M ( \lambda ) =  \gamma( \bar \lambda)^* \gamma ( \lambda )  ,
\end{equation}
see  \cite[Proposition 14.5]{Sch}. This will be important later.
\end{rem}

\subsection{Magnetic translations}

The magnetic Schr\"odinger operator $H^B$ does not  commute with standard lattice translation operators 
\begin{equation}
\label{stl}
T_{\gamma} \psi ( x ) :=\psi(x -\gamma_1 b_1- \gamma_2 b_2).
\end{equation}
It does however commute with modified translations which do not commute with each other in general.
Those magnetic translations $T^{B}_{\gamma} : L^2(\mathcal{E}) \rightarrow L^2(\mathcal{E})$ are unitary operators defined by
\begin{equation}
\label{MagTra}
T^{B}_{\gamma} \psi := u^B(\gamma)T _{\gamma}\psi, 
\ \ \ \psi=(\psi_e)_{e \in \mathcal{E}} \in 
L^2(\mathcal{E}), \ \ \gamma \in \mathbb{Z}^2,
\end{equation}
To define $u^B$ we first consider it as $ 
u^B :\ZZ^2\to C ( W_{\Lambda}) $  where 
$ W_\Lambda $ is the fundamental domain defined in \eqref{eq:r0r} and 
\eqref{eq:fgh}:
\begin{equation}
\begin{gathered}
u^B(\gamma)_e(s\, i(e)+(1-s)\, t(e)):=e^{i \alpha_e(\gamma) s}, \ \ 
e \in W_\Lambda, \\ 
\alpha_e(\gamma):= A(\gamma_1 b_1+\gamma_2 b_2)([e]_1\partial_1+[e]_2 \partial_2) , \\
u^B(\gamma)(r_0) :=1 , \ \  \ 
u^B(\gamma)\left(r_1\right) :=e^{i \alpha_f(\gamma)}.
\end{gathered}
\end{equation}
We then extend $ u^B $ to the graph using translations. 
Using \eqref{eq:Aet} we see that 
\begin{equation}
\label{Vertauscher}
\begin{split}
\alpha_f(\gamma) &= \frac{B}{2} \frac{\sqrt{3}}{2}\left(\gamma_1 -\gamma_2\right) =  \frac{h}{6}(\gamma_1-\gamma_2) , \\
\alpha_g(\gamma) & = \frac{B}{2} \frac{\sqrt{3}}{2} \left(\gamma_1 +2\gamma_2\right) =  \frac{h}{6}(\gamma_1+2\gamma_2) , \\
\alpha_h(\gamma ) & = -\frac{B}{2} \frac{\sqrt{3}}{2} \left(2\gamma_1 +\gamma_2\right) =  -\frac{h}{6}(2\gamma_1+\gamma_2)
\end{split}
\end{equation}
where 
\begin{equation}
\label{eq:defh} h:= \tfrac{3\sqrt{3}} 2 B = {B}{ | b_1 \wedge b_2 |} 
\end{equation} 
is the magnetic flux through one hexagon of the graph.
For any $\gamma,\delta \in \mathbb{Z}^2$
\begin{equation}
\label{eq:transl}
u^B(\gamma)_{[e] - \delta_{1}b_1-\delta_2b_2}:=
e^{i\frac{h \omega(\delta,\gamma)}{2}}  u^B(\gamma)_{[e]}
\end{equation}
where $\omega(\delta,\gamma):=\delta_{1}\gamma_2-\delta_2\gamma_1$ is the standard symplectic form on $\mathbb{R}^2.$
A computation shows that $T^B_\bullet $ satisfies the commutation relation
\begin{equation}
\label{commrelkl}
T_{\gamma}^BT_{\delta}^B = e^{i h \omega \left( \gamma,\delta \right)  }  T_{\delta}^BT_{\gamma}^B.
\end{equation}
It also follows that $T^B_{\gamma}\left(D(H^B)\right) = D(H^B)$, 
and that 
$ T_\gamma^B $ are unitary operators.

Since
\begin{equation}
\label{eq:noncom}
T_{\gamma}^B H^B = H^B T_{\gamma}^B.
\end{equation}
it follows that for every bounded measurable function $f: \mathbb{R} \rightarrow \mathbb{C}$
\begin{equation}
\label{functionalcomm}
T_{\gamma}^B f(H^B) = f(H^B) T_{\gamma}^B.
\end{equation}

  \subsection{Dirac points and band velocities}
  \label{s:dir}
It is well-known that the energy as a function of quasimomenta for graphene has two conical cusps at energies {\em Dirac energies}:  
\[ z_D:=\Delta\vert_{B_n}^{-1}(0) \ \ \text{ (we drop the index $ n$)}
\]
 Those cones (see Figure \ref{Fig:firstbands}) in the energy-quasimomentum representation are referred to as \emph{Dirac cones}. The name is derived from the linear energy-momentum relation for relativistic massless fermions the Dirac equation predicts. 
 \begin{figure}
\centerline{\includegraphics[height=10cm]{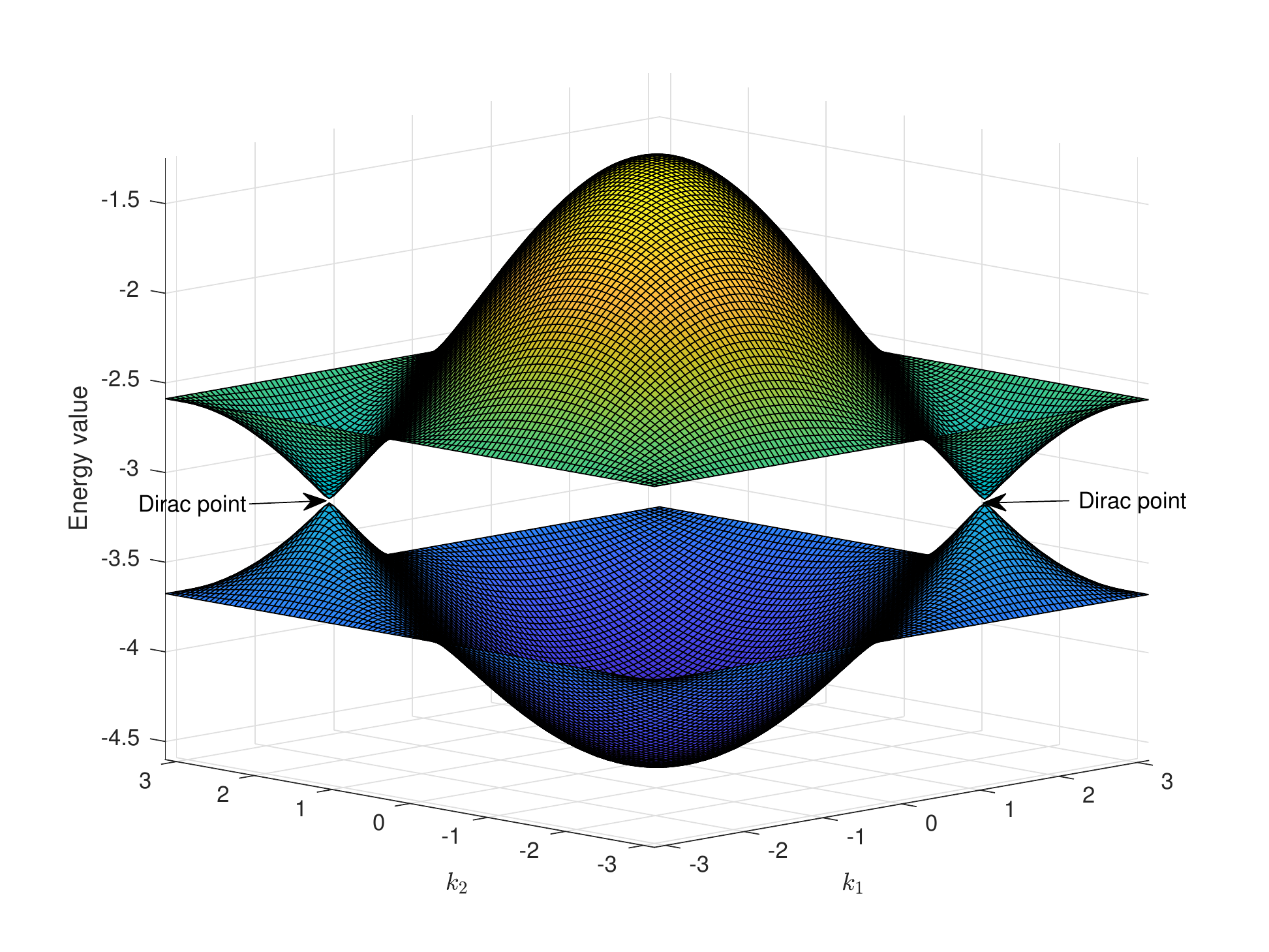}} 
\caption{The first two bands of the Schr\"odinger operator with a Mathieu potential without magnetic perturbation showing the characteristic conical Dirac points at energy level $\approx -\pi$ where the two bands touch.\label{Fig:firstbands}}
\end{figure}

The Hamiltonian $H^B$ with $B=0$ is translational invariant, that is, it commutes with translation operators $T_{\gamma}$ defined in  \eqref{stl}. Using standard Floquet-Bloch theory, one can then diagonalize the operator $H^{B=0}$ as in Kuchment--Post \cite{KP} to write the spectrum for quasimomenta $(k_1,k_2)\in \mathbb{T}^2_*$ (see \eqref{eq:dualT}) in terms of a two-valued function 
\begin{equation}
\label{lambda}
\mathbb{T}^2_* \ni k \mapsto \lambda^{\pm}\vert_{B_n}(k):=\Delta\vert_{B_n}^{-1} \left( \pm \frac{\left\lvert 1+e^{ik_1}+e^{ik_2} \right\rvert}{3}\right)
\end{equation}
 on every Hill band $B_n$ \eqref{Hillbands}.
Expanding $\lambda^{\pm}\vert_{B_n}$ in polar coordinates at the Dirac points $k=\pm \left(\frac{2\pi}{3},-\frac{2\pi}{3}\right)$ yields the linearized energy level sets above $(+)$ and below $(-)$ the conical point
\begin{equation}
\label{TaylorexDir}
\lambda^{\pm}\vert_{B_n}(r,\varphi):=z_D \pm \frac{\Delta\vert_{B_n}^{-1'}(0)}{3} \sqrt{1-\frac{\sin(2 \varphi)}{2} }\ r+ o(r)
\end{equation}
where $r$ is the distance from $k=\pm \left(\frac{2\pi}{3},-\frac{2\pi}{3}\right).$
\begin{defi}[Band velocities]
The Bloch state velocity associated with quasimomenta $(k_1,k_2) \neq \pm \left(\frac{2\pi}{3},-\frac{2\pi}{3}\right)$ is just
\begin{equation}
v^{\pm}\vert_{B_n}(k)
 = \nabla \lambda^{\pm}\vert_{B_n}(k)
\end{equation}
and is fully explicit using \eqref{lambda}. 
\end{defi}
 \begin{figure}
\centerline{\includegraphics[width=12cm]{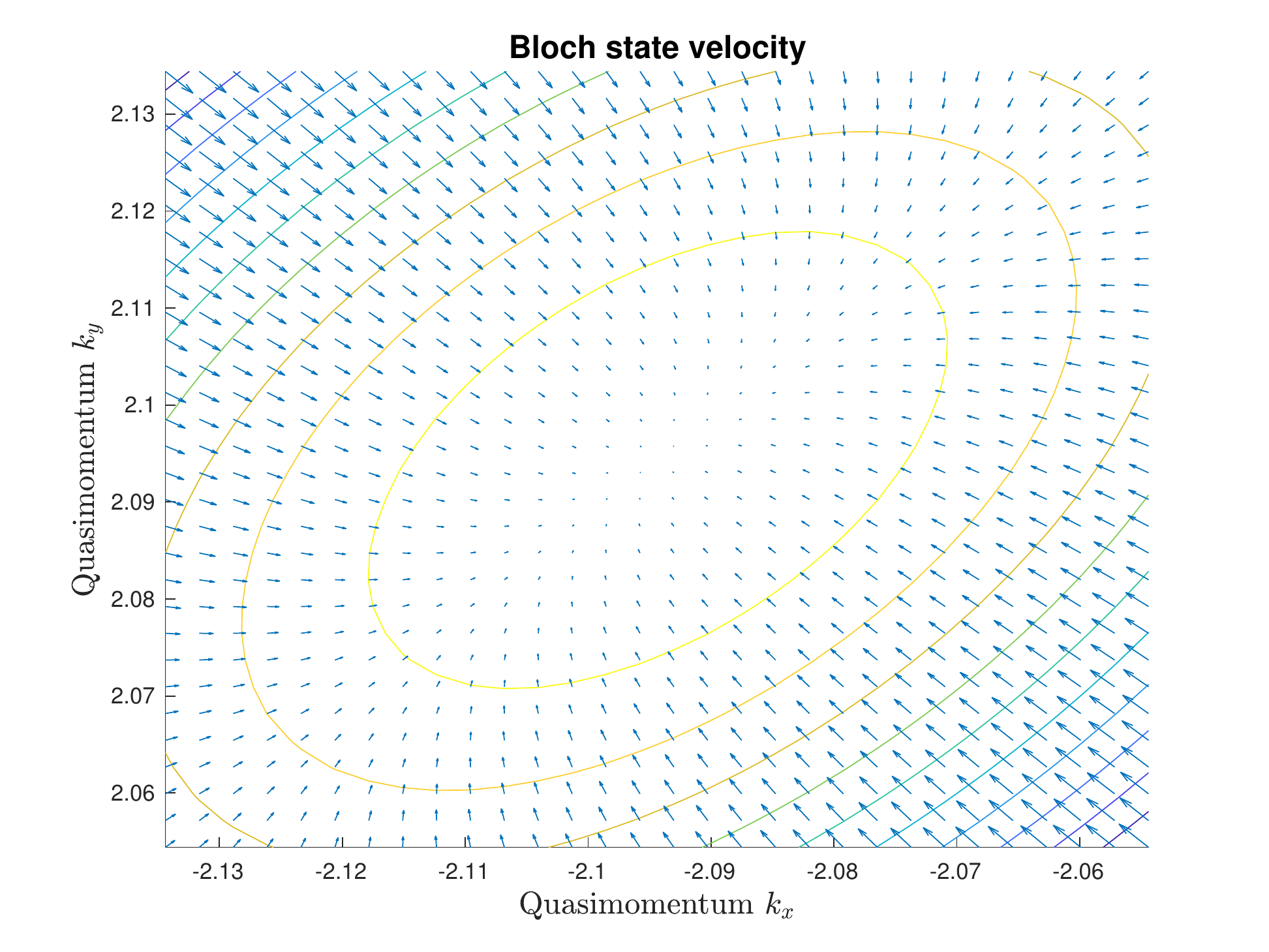}
} 
\caption{The Bloch state velocity of the upper cone near the Dirac point located at $(k_x,k_y)=\frac{2\pi}{3}(-1,1)$ for zero potential $V_e=0.$ In particular, the Bloch state velocity is not rotationally invariant. }
\label{BSV}
\end{figure}

\begin{rem}
The notion of a Fermi velocity in 
the physics literature corresponds to a Bloch state velocity at the conical points. From \eqref{TaylorexDir} and also Figure \ref{BSV} we see that such a limit (if taken in norm) would depend on the angle from which we approach the conical points. Thus, this quantity is not well-defined in this model. Likewise, there has been some controversy about the nature of this quantity in graphene \cite{S17}. See \eqref{eq:gc} for an approximation in our setting.
\end{rem}

\subsection{Different representations of the effective Hamiltonian}
\label{Representations2}
Since any vertex is an integer translate of either of the two vertices $r_0,r_1 \in W_{\Lambda}$ by basis vectors $b_1,b_2$, we indentify $\ell^2(\mathcal{V}) \simeq \ell^2(\mathbb{Z}^2;\mathbb{C}^2)$. Our next Lemma provides the equivalent form of $K_{\Lambda}$ \eqref{discreter} under this identification. 

\begin{lemm} 
\label{eq:unit}
The operator $K_{\Lambda}$ given by \eqref{discreter} 
is unitarily equivalent to an operator $Q_{\Lambda} \in \mathcal L (  \ell^2(\mathbb{Z}^2;\mathbb{C}^2))$ 
\begin{equation}
\label{Q-op}
Q_{\Lambda}  := \tfrac{1}{3} \left(\begin{matrix} 0 && 1+\tau^0+\tau^1 \\ \left(1 +\tau^0+\tau^1 \right)^* && 0 \end{matrix} \right)
\end{equation}
where $ \tau^0, \tau^1 \in \mathcal L ( \ell^2(\mathbb{Z}^2; \mathbb{C}))$
are defined by  
\begin{equation}
\label{translatop}
\begin{gathered}
\tau^0(r)(\gamma):= r(\gamma_1-1,\gamma_2) \nonumber \ \ \ 
\tau^1(r)(\gamma):= e^{i h \gamma_1}r(\gamma_1,\gamma_2-1),   \ \ 
\gamma \in \mathbb{Z}^2, \ \  r \in \ell^2(\mathbb{Z}^2;\mathbb{C}) 
\end{gathered}
\end{equation} 
and satisfy  the Weyl commutation relation
\begin{equation}
\label{CCRDT}
\tau^1 \tau^0  = e^{i h} \tau^0 \tau^1 .\end{equation}
\end{lemm}
\begin{proof}
The unitary operator eliminating the vector potential along two of the three non-equivalent edges is the multiplication operator
\begin{equation}
\label{mul1}
Uz:=\left(\zeta_{v} z ( v) \right)_{v \in \mathcal{V}(\Lambda)}
\end{equation} 
with recursively defined factors 
\begin{gather*}
\zeta_{r_0} :=1 , \ \ 
\zeta_{\gamma_1b_1+\gamma_2b_2+r_1}:=e^{i A_{\gamma_1b_1+\gamma_2b_2+f}} \zeta_{\gamma_1b_1+\gamma_2b_2+r_0}  \\
\zeta_{\gamma_1b_1+(\gamma_2+1)b_2+r_0}:=e^{i \left(-A_{\gamma_1b_1+(\gamma_2+1)b_2 + h}-h\gamma_1+A_{\gamma_1b_1+\gamma_2b_2 + f} \right)}\zeta_{\gamma_1b_1+\gamma_2b_2+r_0} \\
\zeta_{(\gamma_1+1)b_1+\gamma_2 b_2+r_0}:=e^{i \left(-A_{(\gamma_1+1)b_1+\gamma_2b_2 + g}+A_{\gamma_1b_1+\gamma_2b_2 + f} \right)}\zeta_{\gamma_1b_1+\gamma_2b_2+r_0}. 
\end{gather*}
Defining
$ K_{\Lambda}^{\#}:=U^{*}K_{\Lambda}U$ 
we see that 
\begin{equation}
\label{Reduced form}
K_{\Lambda}^{\#}(z)(v)= \frac{1}{3} \begin{cases}  z(v + g) +z(v+f) + e^{i h \gamma_1}z(v+h),\ v \in i\left(\mathcal{V}(\Lambda)\right) \\
z(v - g)+z(v-f) +e^{-ih \gamma_1}z(v-h),\ v \in t \left(\mathcal{V}(\Lambda)\right)
\end{cases}
\end{equation}
where $\gamma_1$ is such that $v= \gamma_1b_1+\gamma_2b_2+r_{0,1}.$
In order to transform $K_{\Lambda}^{\#}$ to $Q_{\Lambda}$ we use the unitary map $
W:\ell^2(\mathcal{V}(\Lambda)) \rightarrow \ell^2(\mathbb{Z}^2, \mathbb{C}^2)$ defined as
\begin{equation}
\label{id1}
Wz\left(\gamma\right):= \left( \begin{matrix} z(r_0 + \gamma_1 b_1+\gamma_2 b_2) \ ,z(\gamma_1 b_1+\gamma_2 b_2+r_1 )\end{matrix} \right)^T.
\end{equation}
We conclude that, $Q_{\Lambda} = (UW^*)^*K_{\Lambda} (U W^*)$, proving the lemma.
\end{proof}

Consider the matrix-valued sequence $a \in \ell^2(\mathbb{Z}^2,\mathbb{C}^2)$ such that
\begin{equation}
\label{1}
\begin{gathered}
a_{(0,0)}:=\frac{1}{3} \left(\begin{matrix} 0 && 1 \\ 1 && 0 \end{matrix}\right), \ \ 
a_{(0,1)}:=\frac{1}{3} \left(\begin{matrix} 0 && 1 \\ 0 && 0 \end{matrix}\right) , \ \ 
a_{(1,0)}:=\frac{1}{3} \left(\begin{matrix} 0 && 1 \\ 0 && 0 \end{matrix}\right), \\
a_{(0,-1)}:=\frac{1}{3} \left(\begin{matrix} 0 && 0 \\ 1 && 0 \end{matrix}\right) \ \ 
a_{(-1,0)}:=\frac{1}{3} \left(\begin{matrix} 0 && 0 \\ 1 && 0 \end{matrix}\right)   \end{gathered}
\end{equation}
and $ a_{\beta}:=0 $ for any other $\beta \in \mathbb{Z}^2$.
Then, we can write \eqref{Q-op} in the compact form
\begin{equation}
Q_{\Lambda}=\sum_{\beta \in \mathbb{Z}^2; \lvert \beta \rvert \le 1}a_{\beta}(\tau^0)^{\beta_1} (\tau^1)^{\beta_2}.
\end{equation} 

We will exhibit two representations of $Q_{\Lambda}$: the first as a {\em magnetic matrix} and then as a {\em  pseudodifferential operator}. 
For that we follow the presentation of Helffer-Sj\"ostrand \cite{HS2}.
We proceed by defining the set of rapidly decaying $\mathbb{C}^{2 \times 2}$-valued functions on $\mathbb{Z}^2$:
\begin{equation*}
\mathscr{S}(\mathbb{Z}^2):=\left\{f : \mathbb{Z}^2 \rightarrow \mathbb{C}^{2 \times 2} \; : \; \forall \, N \ \exists \, C_N \ \ \left\lVert f(\gamma) \right\rVert \le C_N ( 1 + |\gamma|)^{-N}  \right\}.
\end{equation*}
\begin{defi}[Magnetic matrices]
\label{magmat}
A function $  f\in \mathscr S (\mathbb{Z}^2) $ defines a {\em magnetic matrix}\begin{equation}
\label{Magneticmatrix}
A^{h} (f) \in 
\mathcal L\left(\ell^2(\mathbb{Z}^2,\mathbb{C}^{2})\right) ,
\ \ \ 
A^{h}(f) := \left(e^{-i \frac{h}{2} \omega(\gamma,\delta)}f(\gamma-\delta)\right)_{\gamma,\delta \in \mathbb{Z}^2}
\end{equation}
which acts on $\ell^2(\mathbb{Z}^2;\mathbb{C}^2)$ by matrix-like multiplication
\begin{equation}
(A^{h}(f)u)_{\gamma} = \sum_{\delta \in \mathbb{Z}^2} 
\left(A^{h}(f)\right)_{\gamma,\delta} u_{\delta}.
\end{equation}
\end{defi}
We now consider discrete magnetic translations $\tau^B_{\gamma}$ induced by the continuous magnetic translations \eqref{MagTra} on the $\mathbb{Z}^2$-lattice $\tau^B_{\gamma} \in \mathcal L ( \ell^2(\mathbb{Z}^2))$ that are  given by
\begin{equation}
\label{eq:deftau}
\tau^B_{\delta}(f)(\gamma):= e^{-i\frac{h}{2} \omega(\gamma,\delta )} f(\gamma-\delta) , \ \ \ \omega(\gamma,\delta ) := \delta_1 \gamma_2 - \delta_2 \gamma_1 . 
\end{equation}
Just as $H^B$ commutes with the continuous magnetic translations \eqref{MagTra}, the magnetic matrices commute with discrete translations 
\begin{equation}
\left(A^{h}(f)u\right)_{\gamma} = \sum_{\delta \in \mathbb{Z}^2} \left(A^{h}(f)\right)_{\gamma,\delta} u_{\delta}=\sum_{\delta \in \mathbb{Z}^2} \left(\tau_{\delta}^Bf\right)_{\gamma}u_{\delta}, 
\end{equation}
which satisfy
\begin{equation}
\label{commrel}
\tau^B_{\gamma}{\tau}^B_{\delta} = e^{i h \omega(\gamma,\delta)}{\tau}^B_{\delta}\tau^B_{\gamma}.
\end{equation}

\begin{lemm}
\label{uniteq}
$Q_{\Lambda}$ and $A^{h}(a)$, with $a$ given by \eqref{1}, are unitary equivalent.
\end{lemm}
\begin{proof} Let $u \in \ell^2(\mathbb{Z}^2;\mathbb{C}^2)$, then we have
\begin{equation*}
\begin{split} 
(Q_{\Lambda}u)(\gamma) &= \sum_{\delta \in \mathbb{Z}^2; \lvert \delta \rvert \le 1} a_{\delta} e^{i h \gamma_1 \delta_2
} u(\gamma-\delta) 
= \sum_{\delta \in \mathbb{Z}^2; \lvert \delta \rvert \le 1} a_{\delta} e^{-i \frac{h}{2} \delta_1 \delta_2 } e^{i h \gamma_1 \delta_2
} u(\gamma-\delta)  \\
&= \sum_{\delta \in \mathbb{Z}^2; \lvert \gamma-\delta \rvert \le 1} a_{\gamma-\delta} e^{-i \frac{h}{2} (\gamma_1-\delta_1) (\gamma_2-\delta_2)} e^{i h \gamma_1(\gamma_2-\delta_2)}u(\delta)\\
&=\sum_{\delta \in \mathbb{Z}^2; \lvert \gamma-\delta \rvert \le 1}  e^{i \frac{h}{2} (\gamma_1\gamma_2-\delta_1\delta_2)} e^{i \frac{h}{2} (\gamma_2 \delta_1-\delta_2 \gamma_1)} a_{\gamma-\delta}u(\delta) \\
&=\sum_{\delta \in \mathbb{Z}^2; \lvert \gamma-\delta \rvert \le 1}  e^{i \frac{h}{2} \gamma_1\gamma_2} A_{\gamma,\delta}^{h}(a) e^{-i \frac{h}{2} \delta_1\delta_2}u(\delta).
\end{split}
\end{equation*}
Hence, the unitary operator $V \in \mathcal L ( \ell^2(\mathbb{Z}^2; \mathbb{C}^2))$ acting by $\label{mul2}Vu(\gamma):=e^{-i\frac{h}{2} \gamma_1\gamma_2}u(\gamma)$, yields unitary equivalence
$ Q_{\Lambda} = V^*A^{h}(a)V $.
\end{proof}

For $ f,g \in \mathscr S (\mathbb{Z}^2)$ we define a (non-commutative) product
\begin{equation}
\label{eq:defprod}
f\#_{h}g : =A^h(f)(g)=A^{-h}(g)(f)=\sum_{\gamma \in \mathbb{Z}^2} f(\gamma) (\tau_{\gamma}^{-B} g)(\bullet) .
\end{equation}
If $f \in \mathscr{S}(\mathbb{Z}^2)$ then 
\begin{equation}
\label{eq:in1} 
 A^{h}(f)^{-1} \in \mathcal L\left( \ell^2  \left(\mathbb{Z}^2,\mathbb{C}^{2 \times 2} \right)\right) \ \Longrightarrow 
 \ \exists\, g \in \mathscr{S}(\mathbb{Z}^2), \ \ 
 A^h ( f ) ^{-1} = A^{h}(g) ,
 \end{equation}
 and 
\begin{equation*}
f \#_{h} g = g \#_{h} f = \operatorname{id}_{\mathbb{C}^{2 \times 2}}\delta_{0},
\end{equation*}
see the proof at the end of this section and \cite[Proposition $5.1$]{HS1} for a slightly different statement.

For  $f \in \mathscr{S}(\mathbb{Z}^2)$, we  define the Fourier transform as
\begin{equation*}
\widehat{f}(x,\xi):= \sum_{\gamma \in \mathbb{Z}^2} f(\gamma)e^{i\langle \gamma, (x,\xi)^T \rangle}, \text{ such that } \widehat f \in C^\infty ( \mathbb T^2_* ) .
\end{equation*}
In particular, for $a \in \mathscr{S}(\mathbb{Z}^2)$ given in \eqref{1} the Fourier transform is given by  
\begin{equation}
\widehat{a}(x,\xi) := \tfrac{1}{3} \left(\begin{matrix} 0 & 1 +e^{ix}+e^{i\xi} \\
1+e^{-ix}+e^{-i\xi} &0  \end{matrix}\right).
\end{equation}

We observe that for $\gamma:=(1,0)$ and $\delta:=(0,1)$, equation \eqref{commrel} becomes
\begin{equation}
 \tau^{-B}_{\gamma}\tau^{-B}_{\delta}= e^{-i h}\tau^{-B}_{\delta}\tau^{-B}_{\gamma}  .
\end{equation}
In semiclassical Weyl quantization (see \cite[Theorem $4.7$]{ev-zw}) 
 the same commutation relation is satisfied by
\begin{equation}
\operatorname{Op}_{h}^{\rm{w}}\left(e^{i x}\right) \operatorname{Op}_{h}^{\rm{w}} \left(e^{i\xi}\right)=e^{-i h}\operatorname{Op}_{h}^{\rm{w}} \left(e^{i \xi}\right)\operatorname{Op}_{h}^{\rm{w}}\left(e^{i x}\right).
\end{equation}
Looking at the product formula we see that when we replace 
  $\tau^{-B}_{\gamma}$  in \eqref{eq:defprod} by 
 \[ \operatorname{Op}^{\rm{w}} _{h}\left((x,\xi)\mapsto e^{i\langle \gamma, (x,\xi)^T \rangle}\right)\]
 we obtain a homomorphism
\begin{equation}
\label{star}
\begin{gathered} 
\Theta : \mathscr S ( \ZZ^2 ) \to \mathcal L\left(L^2(\mathbb{R})\right) \nonumber  , \ \ 
\Theta ( f ) := \sum_{\gamma \in \mathbb{Z}^2} f(\gamma)\operatorname{Op}^{\rm{w}} _{h}\left((x,\xi)\mapsto e^{i\langle \gamma, (x,\xi)^T \rangle}\right)= \operatorname{Op}_h^{\rm{w}}(\widehat{f}) , \\
\Theta ( f \#_h g ) = \Theta ( f ) \circ \Theta ( g ) , \ \ 
\Theta ( f ( - \bullet )^* ) = \Theta ( f)^* . 
\end{gathered}
\end{equation}

\begin{proof}[Proof of \eqref{eq:in1}] Invertibility of $ A^h ( f ) $
on $ \ell^2 $ is equivalent to invertibility of $ \Op_h^{\rm{w}} ( \widehat f ) $. A semiclassical version of Beals's lemma, due to Helffer--Sj\"ostrand (see \cite[Chapter 8]{D-S} or \cite[Theorem 8.3]{ev-zw}), shows that $ \Op_h^{\rm{w}} (\widehat f)^{-1}
= \Op_h^{\rm{w}} ( G ) $, $ G \in S ( 1 ) $. We also see that $ G $ has to be periodic and that implies that $ G = \widehat g $ for $ g \in \mathscr S $.
\end{proof}

\section{ Regularized traces }

As recalled in \S \ref{s:intr} the density of states is defined using 
regularized traces of functions of the Hamiltonian. We start with a 
general definition:
\begin{defi} Put $ B(R) := \{ x \in \RR^2 : |x| < R \} $ and 
suppose that $T\in \mathcal L ( L^2(\mathcal{E})) $ has the
property for all $R>0$ the operator $\indic_{B(R)} T \indic_{B(R)}$ is of trace-class. Then we define 
\begin{equation}
\label{conttrace}
\widetilde{\operatorname{tr}} T := \lim_{R \rightarrow \infty} \frac{\operatorname{tr}\indic_{B(R)} T \indic_{B(R)}}{\left\lvert B(R)\right\rvert}
\end{equation}
provided this limit exists.

Similarly, for a lattice $ \Gamma \subset \RR^2  $ and $ A \in \mathcal L ( \ell^2( \Gamma , \CC^2  ) ) $ given by 
\begin{equation*}
A (s) (\gamma):= \sum_{\beta \in \ZZ^2} k(\gamma,\beta)s(\beta)
\end{equation*}
with $k(\gamma,\beta) \in \mathbb{C}^{2 \times 2},$ we define \begin{equation}
\label{discretetrace}
\widehat{\operatorname{tr}}_{\Gamma}  A := \lim_{R \rightarrow \infty} \frac{1}{\left\lvert B(R) \right\rvert} \sum_{\gamma \in \Gamma  \cap B(R)} \operatorname{tr}_{\mathbb{C}^{2}} k ( \gamma, \gamma )
\end{equation}
provided the limit exists.
\end{defi}

\begin{rem}
Most of the results of this section hold for both $H^B$ and $H^{{D}}$ and the proofs do not differ for the two operators. In such case we consider $H^B $ only.
\end{rem}

We start with some general comments about $ \widehat \tr $:

\begin{lemm}
\label{trmagex}
Let $ g \in \mathscr S ( \ZZ^2 ) $ and let $ A^h(g)$ be the 
corresponding magnetic matrix (Definition \ref{magmat}). Then the regularized trace $\widehat{\operatorname{tr}}_{\mathbb{Z}^2} (A^h(g)) $ exists and is given by
\begin{equation}
\label{eq:traces}
\widehat{\operatorname{tr}}_{\mathbb{Z}^2} (A^h(g))=\operatorname{tr}_{\mathbb{C}^{2}}(g(0)) = \tfrac{1}{ ( 2 \pi)^2 }
\int_{\mathbb{T}^2_*} \operatorname{tr}_{\mathbb{C}^{2}} 
\widehat g(x,\xi)  \, {dx  d\xi} .
\end{equation}
\end{lemm}
\begin{proof}
Since the kernels of the magnetic matrix satisfy on the diagonal $A^h(g)_{\gamma,\gamma}=g(0)$ the proof of this equality is immediate.
\end{proof}
In view of this lemma we will abuse the notation slightly and introduce
\begin{defi}
\label{pdotrace}
Let $f  \in C^\infty ( \RR^2 ) $ be $(2\pi \mathbb{Z})^2$ periodic. Then we define the regularized trace
\begin{equation}
\widehat{\operatorname{tr}}(\operatorname{Op}_h^{w}(f)) :=\tfrac{1}{ ( 2 \pi)^2 }
 \int_{\mathbb{T}^2_*} \operatorname{tr}_{\mathbb{C}^{2 }} f(x,\xi)\,  {dx  d\xi}.
\end{equation}
\end{defi}

We now show that for $ f \in C_{\rm{c}} ( \RR ) $ the operators 
$ f ( H^B ) $ and $ f ( H^{B, D } ) $ have regularized traces. Because we are essentially in dimension one, we have stronger trace class properties:
\begin{lemm}
\label{existtra}
For $ z \in \CC \setminus \RR $ the regularized traces of $(H^\bullet-z)^{-1}$  exist and 
\begin{equation}
\label{eq:ttes}
\widetilde \tr (H^\bullet-z)^{-1} = 
\tfrac{2}{3 \sqrt 3 } \tr \indic_{ \mathcal E ( W_\Lambda ) } 
( H^\bullet - z )^{-1}, \ \  \bullet = B, D .
\end{equation}
\end{lemm}
\begin{proof}
We consider $H^B$ only. Since $ D ( H^B ) \subset 
H^2 ( \mathcal E ) $, we see that for $ \psi \in C_{\rm{c}}^\infty ( B_{\mathbb R^2}  ( 0, 2 R  )) $, $ \psi ( H^B - z )^{-1} :
L^2 ( \mathcal E ) \to H^2 ( \mathcal E \cap B_{\RR^2}  ( 0 , 2 R ) ) $ is of trace class. (We are in dimension one here and the trace class property
in dimension $ n$ is obtained for maps $ L^2 ( \RR^n ) \to H^{s }( B ( 0 , 
r ) ) $, $s > n $; hence $ H^2 $ is sufficient -- see for instance 
\cite[Proposition B.20]{res}.) In addition, we have the trace norm estimate:
\begin{equation}
\begin{split} 
\label{eq:zest}
\| \psi ( H^B - z )^{-1} \|_{\mathcal L_1}  & \leq 
C_\psi \| ( H^B - z )^{-1} \|_{ L^2 \to D ( H^B ) } 
\leq C_\psi \sup_{ x \in \RR } |x -z|^{-1} ( 1 + |x|) \\
& \leq C_\psi ( 1 + |\Re z | ) | \Im z|^{-1} . 
\end{split}
\end{equation}
If we choose $ \psi \equiv 1 $ on a neighbourhood of $ B_{\RR^2} ( 0 , R ) $ then 
\[  \indic_{B_{\RR^2}  ( 0 , R ) } ( H^B - z )^{-1} = 
 \indic_{B_{\RR^2}  ( 0 , R ) } \psi ( H^B - z )^{-1} \in 
 \mathcal L_1 ( L^2 ( \mathcal E )) . \]
We now choose $ m_R , M_R \subset \ZZ^2 $ such that 
\begin{equation}
\label{eq:latap}
\begin{gathered} \Omega_{m_R} 
\subset 
    B_{\RR^2 } ( 0 , R ) \cap \mathcal E 
    \subset \Omega_{M_R}  , \\
\Omega_{ Q} := \bigcup_{ \gamma \in Q } ( \mathcal E ( W_\Lambda ) + \gamma_1 b_1 + 
\gamma_2 b_2 ) , 
    \ \ \ 
  | M_R \setminus m_R  | \leq C R . 
  \end{gathered}
\end{equation}
In particular, since the area of a hexagonal cell is given by $ \frac{ 3 \sqrt 3 }2 $, we have
\begin{equation}
\label{eq:mR} 
 | m_R | = \tfrac{2}{3 \sqrt 3 } | B_{\RR^2 } ( 0, R ) | + \mathcal {\mathcal O} (R ). \end{equation}
We now write 
\begin{equation}
\label{eq:trde} 
\begin{split} \tr \indic_ { B_{\RR^2 } ( 0 , R ) } ( H^B - z )^{-1} & = \tr \indic_ { \Omega_{m_R}   } ( H^B - z )^{-1}
+  \tr \indic_{ B_{\RR^2 } ( 0 , R ) \setminus \Omega_{m_R} } ( H^B - z )^{-1}.
\end{split}
\end{equation}
Using \eqref{eq:transl} we get
\begin{equation}
\begin{split}
 T_\gamma^B \indic_{ \mathcal E ( W_\Lambda ) } T_{-\gamma}^Bf &=  T_\gamma^B \indic_{ \mathcal E ( W_\Lambda ) } u^B(-\gamma)f(\bullet+ \gamma_1b_1+\gamma_2b_2)  \\
 &= \indic_{ \mathcal E(W_\Lambda) + \gamma_1 b_1 + \gamma_2 b_2 }f 
\end{split}
\end{equation}
so that we can expand the first term on the right hand side of \eqref{eq:trde} as follows
\begin{equation}
\label{eq:ex2}
\begin{split}
\tr \indic { \Omega_{m_R}  } ( H^B - z )^{-1} & = \sum_{\gamma \in m_R } 
\tr \indic_{ \mathcal E(W_\Lambda) + \gamma_1 b_1 + \gamma_2 b_2 } ( H^B - z )^{-1}  \\
& = \sum_{\gamma \in m_R } 
\tr T_\gamma^B \indic_{ \mathcal E(W_\Lambda)} T_{-\gamma}^B ( H^B - z )^{-1} \\ 
& = |m_R | \tr \indic_{ \mathcal E( W_\Lambda ) } ( H^B - z )^{-1} .
\end{split}
\end{equation} 
Here we used \eqref{eq:noncom} and the cyclicity of the trace.

To estimate the second term in \eqref{eq:trde} we write
\begin{equation}
\label{eq:2te} \begin{split}  \| \indic_{ B_{\RR^2 } ( 0 , R ) \setminus \Omega_{m_R} } ( H^B - z )^{-1} \|_{ \mathcal L_1 } &  \leq 
\| \indic_{ \Omega_{M_R}  \setminus \Omega_{m_R} } ( H^B - z )^{-1} \|_{ \mathcal L_1 } \\
& \leq \sum_{ \gamma \in M_R \setminus m_R } 
\|  \indic_{ \mathcal E(W_\Lambda) + \gamma_1 b_1 + \gamma_2 b_2 } ( H^B - z )^{-1} \|_{ \mathcal L_1} 
\\
& \leq  \sum_{ \gamma \in M_R \setminus m_R } 
\|  T_\gamma^B \indic_{ \mathcal E(W_\Lambda)} T_{-\gamma}^B  ( H^B - z )^{-1} \|_{ \mathcal L_1} \\
& = | M_R \setminus m_R | \| \indic_{ \mathcal E( W_\Lambda ) } ( H^B - z )^{-1} \|_{\mathcal L_1 } \\
& \leq C R ( 1 + |\Re z | ) | \Im z|^{-1} , 
\end{split}
\end{equation}
where we used \eqref{eq:zest} and \eqref{eq:latap}. Returning to 
\eqref{eq:trde} we see that \eqref{eq:ttes} follows from \eqref{eq:mR} and
\eqref{eq:ex2}.
\end{proof}

We now consider regularized traces of $ f ( H^B ) $ and
$ f ( H^D ) $ and we will use the functional calculus of Helffer--Sj\"ostrand. For that we recall that for 
any $f \in C_{\rm{c}}^{\infty}(\mathbb{R})$ can be extended to $ \widetilde f \in \mathscr S ( \CC ) $ such that $ \widetilde f|_\RR = f $ and
$ \partial_{\bar z} \widetilde f = \mathcal {\mathcal O} ( |\Im z |^\infty ) $. 
The function $ \tilde f $ is a then called an {\em almost analytic} extension of $ f $. A compact formula for $ \widetilde f $ was 
given by Mather and Jensen--Nakamura:
\begin{equation}
\label{eq:mather}
\begin{gathered} 
\widetilde f ( x+ i y ) = \tfrac{1}{ 2 \pi }  \chi ( y ) \psi ( x ) 
\int_{\RR} \chi( y \xi ) \widehat f ( \xi ) e^{ i (x + i y ) \xi } d \xi , 
\\  \chi, \psi \in C^\infty_{\rm{c}}( \RR ) , \ \  \psi| _{ \supp f + 
( - 1, 1 )} = 1, \   \ \chi |_{ (-1, 1) } = 1 ,
\end{gathered}
\end{equation}
see for instance \cite[Chapter 8]{D-S}.
The relevance of this construction here comes from the Helffer-Sj\"ostrand formula: for any self-adjoint operator $ P $, 
\begin{equation}
\label{eq:HeSj}
f(P) = \tfrac{1}{\pi} \int_{\mathbb{C}} \partial_{\overline{z}} \widetilde{f}(z) (P -z)^{-1} dm(z)
\end{equation}
where $\lambda_{\mathbb{C}}$ is the Lebesgue measure on $\mathbb{C}.$ The integral on the right hand side is  well-defined as $ \partial_{\overline{z}}\widetilde{f}(z) = \mathcal O ( |\Im z |^\infty ) $  and  $\| (P -z)^{-1} \| = \mathcal O ( 1/|\Im z |) $, by self-adjointness.

The proof of Lemma \ref{existtra} and the dominated convergence theorem 
based on \eqref{eq:zest},\eqref{eq:ex2} and \eqref{eq:2te}, immediately give
\begin{lemm}
\label{existencelemma}
Let $f \in C_c(\mathbb{R})$ then $\widetilde{\operatorname{tr}}(f(H^\bullet))$  exist and 
\begin{equation}
\widetilde{\operatorname{tr}}f(H^\bullet)=\tfrac{2}{3 \sqrt{3}} \tr\indic_{\mathcal{E}(W_{\Lambda})} {f(H^\bullet)} . 
\end{equation}
\end{lemm}
The lemma allows a rigorous definition of the density of states measure:
the functional $C_c(\mathbb{R}) \ni f \mapsto \widetilde{\operatorname{tr}}(f(H^B))$ is positive. Thus, by the Riesz-Markov theorem, it defines a Radon measure:
\begin{defi}[Density of states measure]
\label{de:DOS}
The density of states $ \rho_B \in {\mathscr D'}^0 ( \RR ) $ is the Radon measure such that
\begin{equation*}
\widetilde{\operatorname{tr}}(f(H^B))=\int_{\mathbb{R}} f(x) \rho_{B}(x) dx,
\end{equation*}
where we use the informal notation for the action of distributions of order zero on function (see \cite[\S 2.1]{H1}) The distribution function of the
measure $ \rho_B $ is called the integrated density of states.
\end{defi}

In Krein's resolvent formula \eqref{Kreinresolvform} the auxiliary operators $\Lambda^B$ and $\Lambda^D$ appear instead of $H^B$ and $H^{D}$. The following Lemma shows that their regularized traces coincide.
\begin{lemm}
\label{substitute}
For $f \in C_c(\mathbb{R})$, 
\begin{equation}
\widetilde{\operatorname{tr}}(f(\Lambda^\bullet))= \widetilde{\operatorname{tr}}(f(H^\bullet)), \ \ \bullet = B, D.
\end{equation}
\end{lemm}
\begin{proof}
By the functional calculus, the unitary Peierls' substitution $P$ satisfies \eqref{LambdaB}
\begin{equation}
f(\Lambda^B) = P^{-1} f(H^B) P.
\end{equation}
Since $P$ and $P^{-1}$ are just multiplication operators
\begin{align}
\operatorname{tr}(\indic_{B(R)}\ f(\Lambda^B)\ \indic_{B(R)})
&=\operatorname{tr}(P\ \indic_{B(R)}\ f(\Lambda^B)\ \indic_{B(R)}\ P^{-1}) \nonumber\\
&=\operatorname{tr}(P\ \indic_{B(R)}\ P^{-1} f(H^B) P \ \indic_{B(R)}\ P^{-1}) \nonumber\\
&=\operatorname{tr}(\indic_{B(R)}\ f(H^B)\  \indic_{B(R)}). 
\end{align}
Lemma \ref{existencelemma} shows the existence of the regularized trace then.
\end{proof}

We now combine \eqref{eq:HeSj} with Krein's formula \eqref{Kreinresolvform} to see that
\begin{align}
f(\Lambda^B) 
&=\tfrac{1}{\pi} \int_{\mathbb{C}} \partial_{\overline{z}} \widetilde{f}(z)\left((\Lambda^D-z)^{-1} -\gamma(z) M(z)^{-1} \gamma(\overline{z})^*\right) dm (z) \nonumber \\
&= f(\Lambda^D) - \tfrac{1}{\pi} \int_{\mathbb{C}} \partial_{\overline{z}} \widetilde{f}(z)\gamma(z) M(z)^{-1}\gamma(\overline{z})^* dm(z).
\end{align}
Using Lemma \ref{substitute}, we can apply the operator $\widetilde{\operatorname{tr}}$ to the preceding equation and obtain
\begin{equation}
\label{centraleq}
\widetilde{\operatorname{tr}}f(H^B)= \widetilde{\operatorname{tr}}f(\Lambda^D)- \tfrac{1}{\pi}\widetilde{\operatorname{tr}}\int_{\mathbb{C}}\partial_{\overline{z}} \widetilde{f}(z)\gamma(z) M(z)^{-1} \gamma(\overline{z})^* dm (z) .
\end{equation}

In the following, we will systematically analyze the terms on the right side. 
We start with the term containing operator $\Lambda^D.$
\begin{lemm}
\label{Dirichlet-contribution}
The contribution $\widetilde{\operatorname{tr}}(f(\Lambda^D))$ of the Dirichlet operator $\Lambda^D$ is given by
\begin{equation}
\label{eq:Dsum}
\widetilde{\operatorname{tr}}f (\Lambda^D)=\tfrac{2}{\sqrt{3}} \sum_{\lambda \in \operatorname{Spec}(\Lambda^D_{(0,1)})} f(\lambda)
\end{equation}
where $\Lambda^D_{(0,1)}:H_0^1(0,1)\cap H^2(0,1) \subset L^2(0,1) \rightarrow L^2(0,1)$ with $\Lambda^D_{(0,1)}\psi:=-\psi''+V\psi.$
\end{lemm}
\begin{proof}
Let $\Lambda^D = \sum_{\lambda \in \operatorname{Spec}(\Lambda^D_{(0,1)})}^\infty \lambda P_{\ker(\Lambda^D-\lambda)}$ be the spectral decomposition of $\Lambda^D$ where \newline $P_{\ker(\Lambda^D-\lambda)}$ is the orthogonal projection onto the infinite dimensional space $\ker(\Lambda^D-\lambda)$. The spectral theorem implies $f(\Lambda^D) = \sum_{\lambda \in \operatorname{Spec}(\Lambda^D_{(0,1)})}^\infty f(\lambda) P_{\ker(\Lambda^D-\lambda)}$, which is a finite sum, as the eigenvalues of the Dirichlet operator tend to infinity.
Thus, since each edge carries precisely one non-degenerate eigenfunction for every eigenvalue $\lambda \in \operatorname{Spec}(\Lambda^D)$,
\begin{equation*} 
\begin{split} 
\widetilde{\operatorname{tr}}f(\Lambda^D) &= \lim_{R \rightarrow \infty} \frac{\operatorname{tr}\left(\indic_{B(R)}\ f(\Lambda^D)\ \indic_{B(R)} \right)}{\left\lvert B(R) \right\rvert}   \\
&=  \sum_{\lambda \in \operatorname{Spec}(\Lambda^D_{(0,1)})}f(\lambda) \lim_{R \rightarrow \infty} \frac{\operatorname{tr}\left(\indic_{B(R)}\   P_{\ker(\Lambda^D-\lambda)}\ \indic_{B(R)}\right)}{\left\lvert B(R) \right\rvert}  
 = \tfrac{2}{\sqrt{3}} \sum_{\lambda \in \operatorname{Spec}(\Lambda^D_{(0,1)})} f(\lambda),
\end{split}
\end{equation*}
with $\frac{2}{\sqrt{3}}$ being the ratio of edges per unit volume.
\end{proof}

We now move to the second term in \eqref{centraleq}.
In particulare we eliminate the gamma field in our expressions. 
\begin{lemm}
With $ M ( z ) $ defined in \eqref{eq:defM} we have 
\begin{equation*}
\label{Reduzierer}
\widetilde{\operatorname{tr}} \int_{\mathbb{C}} \partial_{\overline{z}} \widetilde{f}(z)
 \, \gamma(z)M(z)^{-1} \gamma(\overline{z})^*
dm(z) 
= \int_{\mathbb{C}} \partial_{\overline{z}}  \widetilde{f}(z)\, 
\widehat{\operatorname{tr}}_{\mathcal{V}(\Lambda)} 
\, \partial_z M ( z ) 
M( z )^{-1}  
dm (z) .
\end{equation*}
\end{lemm}
\begin{proof}
The estimates in the proof of Lemma \ref{existtra} show that we can 
move $ \widetilde \tr $ inside of the integral on the left hand side. Together with 
\eqref{eq:derM} this means that it suffices to prove that 
\begin{equation}
\label{eq:trtr} \widetilde{\operatorname{tr}} \left(  \gamma(z)M(z)^{-1} \gamma(\overline{z})^*\right) 
= \widehat{\operatorname{tr}}_{\mathcal{V}(\Lambda)} \left(\gamma(\overline{z})^*\gamma(z)M(z)^{-1}\right), \ \  z \in \CC \setminus \RR .
\end{equation}
This identity can now be shown by verifying the conditions of the third statement in \cite[Proposition $7.1$]{HS1} with 
\begin{equation}  C:=\gamma(z)M(z)^{-1}, \ \ \ D:=\gamma(\overline{z})^* 
\label{eq:CD}
\end{equation}
but we present a different argument. 

Using the unitary Peierls operator $ P $ \eqref{Peierl} 
magnetic translations \eqref{MagTra}, and operators $ \pi $ and $ \gamma ( z ) $ from \eqref{eq:defpi},\eqref{gfield} we define modified magnetic translations (note that $ z \notin \RR $) as the following 
unitary operators:
\[  S_{\delta}^B:=P^{-1}T_{\delta}^BP \in \mathcal U ( L^2 ( \mathcal E) ) , \ \ \ \sigma_\delta^B := \pi S_\delta^B \gamma ( z ) \in 
\mathcal U ( \ell^2 ( \mathcal V ) ) ,
 \]
where we note that $ \sigma_\delta^B $ does not depend on $ z $.

To see that $ \sigma_\delta^B $ is unitary we first note that 
$ (\sigma_\delta^B)^{-1} = \sigma_{-\delta}^B $ and that it is an isometry
(see \eqref{MagTra} and \eqref{eq:transl} for definitions of 
$ T_\gamma^B $ and $ u^B( \gamma ) $):
\begin{equation*}
\begin{split}
\left\lVert \sigma^B_{\delta}w \right\rVert^2 &=\sum_{v \in \mathcal{V}(\Lambda)} \left\lvert (\pi S_{\delta}^B  \gamma(z)w)(v) \right\rvert^2 
= \sum_{v \in \mathcal{V}(\Lambda)} \left\lvert ( P^{-1}T_{\delta}^BP\gamma(z)w)(v) \right\rvert^2\\
&= \sum_{v \in \mathcal{V}(\Lambda)} \left\lvert (u^B(\delta)P\gamma(z)w)(v-\delta_1b_1-\delta_2b_2) \right\rvert^2 = \sum_{v \in \mathcal{V}(\Lambda)} \left\lvert (\gamma(z)w)(v) \right\rvert^2\\
&= \sum_{v \in \mathcal{V}(\Lambda)} \left\lvert w(v) \right\rvert^2=\left\lVert w \right\rVert^2.
\end{split}
\end{equation*}
We now claim that $ M ( z ) ^{-1} $ commutes with $ \sigma_\gamma^B$. In fact, since $H^B$ \eqref{magop} and $H^{D}$ \eqref{magopd} commute with magnetic translations $T_{\delta}^B$, we see that \eqref{Peierl}, $\Lambda^B$ and $\Lambda^D$ commute then with $S_{\delta}^B$.
The Krein formula \eqref{Kreinresolvform} then implies that 
$S^B_{\delta} \left(\gamma(z)M(z)^{-1} \gamma(\overline{z})^*\right) = \left(\gamma(z)M(z)^{-1} \gamma(\overline{z})^*\right)S^B_{\delta}$.
Multiplying with the inverse of $\gamma(z)$ and $\gamma(\overline{z})^*$ from both sides respectively, it follows that
\begin{equation}
\label{4.93}
\sigma^B_{\delta} M(z)^{-1} = \left(\pi S^B_{\delta} \gamma(z) \right) M(z)^{-1}=M(z)^{-1}\left( \gamma(\overline{z})^* S^B_{\delta} \pi^{*}\right) = M(z)^{-1}\sigma^B_{\delta}.
\end{equation}

In the notation of \eqref{eq:CD} we then see that 
\begin{align}
\label{firstidentity}
S_{\delta}^B C & =S_{\delta}^B \gamma( z ) M(z)^{-1} = \gamma(z)\sigma^B_{\delta} M(z)^{-1} = \gamma(z) M(z)^{-1} \sigma^B_{\delta}=C \sigma^B_{\delta}
\end{align}
and
\begin{equation}
\label{secondidentity}
\sigma^B_{\delta}  D= \left(\gamma(\overline{z})\sigma^B_{-\delta}\right)^*= 
\left( \gamma( \bar z ) \pi S_{-\delta}^B \gamma( \bar z ) \right)^* = 
\left( S_{-\delta}^B \gamma( \bar z ) \right)^* = 
\gamma(\overline{z})^* S_{\delta}^{B}=DS_{\delta}^{B}
\end{equation}

As in the proof of Lemma \ref{existtra}, 
\[ \begin{split} 
\widetilde \tr \, CD & =\tfrac{2}{3 \sqrt{3}}\tr_{L^2 ( \mathcal E)}  \indic_{\mathcal E ( W_\Lambda ) } C D =\tfrac{2}{3 \sqrt{3}}
\tr_{\ell^2 ( \mathcal V ) }  D \indic_{\mathcal E (W_\Lambda ) } C 
\\ & 
=\tfrac{2}{3 \sqrt{3}} \sum_{ \gamma \in \mathbb{Z}^2 } \sum_{v \in \mathcal{V}(W_\Lambda)}
\left[ \sigma_\gamma^B D \indic_{\mathcal E (W_\Lambda ) } C {\sigma_{-\gamma}^{B} } \right] ( v ) , \end{split} \]
where for an operator $ A $ on $ \ell^2 ( \mathcal V ) $ we write
$ A u ( \gamma ) = \sum_{ \alpha \in \mathcal V } [A] ( \gamma, \alpha ) u ( \alpha ) $. Using \eqref{firstidentity}, \eqref{secondidentity} and
\eqref{eq:transl} we then obtain 
\begin{equation*}
\begin{split}
 \widetilde \tr \, CD &=\tfrac{2}{3 \sqrt{3}} \sum_{\gamma \in \mathbb{Z}^2} \sum_{v \in \mathcal{V}(W_\Lambda)}
[ D \indic_{ \mathcal E ( W_\Lambda ) + \gamma_1 b_1 + \gamma_2 b_2 } 
C ] ( v, v ) \nonumber\\
&= \tfrac{2}{3 \sqrt{3}} \sum_{v \in \mathcal{V}(W_\Lambda)}[ D I_{L^2 (\mathcal E ) } C] ( v,v ) =\tfrac{2}{3 \sqrt{3}} \sum_{v \in \mathcal{V}(W_\Lambda)} [ D C ] ( v , v ). 
\end{split}
\end{equation*}
Since $ DC $ commutes with $ \sigma_{\delta}^B $ it is unitarily equivalent to a magnetic matrix which in view of Lemma \ref{trmagex} and a lattice identification means that
\begin{equation*}
 \widehat \tr_{\mathcal{V}(\Lambda)} DC = \tfrac{2}{3 \sqrt{3}} \sum_{v \in \mathcal{V}(W_{\Lambda})}[ D C ] (v,v ) .
 \end{equation*} 
This proves \eqref{eq:trtr} which as explained in the beginning concludes the proof.
\end{proof}

We can now combine Lemmas \ref{Dirichlet-contribution},\ref{Reduzierer}
and the Krein formula to obtain
\begin{lemm}
\label{Dirichletcalc}
Using  \eqref{eq:defM} and \eqref{discreter} define.
\begin{equation}
\label{eq:Wlambda}   W (z) := s_z ( 1 ) M (z ) = K_\Lambda  - \Delta ( z ) .
\end{equation}
Then for $ f \in C^\infty_{\rm{c}}  ( \RR) $ with an almost analytic extension \eqref{eq:mather}, $ \widetilde f \in C^\infty_{\rm{c}} ( \CC) $, 
\begin{equation}
\label{representation1}
\widetilde{\operatorname{tr}}(f(H^B))= - \tfrac{1}{\pi} \int_{\mathbb{C}} \partial_{\overline{z}}  \widetilde{f}(z)   \widehat \tr _{\mathcal{V}} \,  \partial_z W(z) W(z)^{-1} \,d m(z)+ \tfrac{2}{3 \sqrt{3}} \sum_{\lambda \in \operatorname{Spec}(\Lambda^D_{(0,1)})}f(\lambda).
\end{equation}
\end{lemm}

\begin{proof}
Since $ \widehat  \tr \, I_{ \ell^2( \mathcal V ) } = \frac4{ 3 \sqrt 3}  $ (the number of vertices per unit volume) we have 
\begin{equation}
\label{eq:M2W} \widehat \tr \, \partial_z M( z) M ( z)^{-1} = 
- \tfrac4{ 3 \sqrt 3}\, \partial_z s_z ( 1) s_z ( 1 )^{-1} + \widehat \tr \, \partial_z W ( z )  W ( z )^{-1} . 
\end{equation}
Since the zeros of $ z \mapsto s_z ( 1 ) $ are given by 
the eigenvalues of $ \Lambda^D_{(0,1)} $, the Cauchy formula \cite[(3.1.11)]{H1} shows that 
\[ \tfrac{1}{\pi} \int_{\mathbb{C}} \partial_{\overline{z}}  \widetilde{f}(z)  \partial_z s_z ( 1) s_z(1)^{-1}  d m(z) = \sum_{ \lambda \in 
\Spec ( \Lambda^D_{(0,1)} ) } f ( \lambda ) . \]
Combining this with \eqref{centraleq}, \eqref{eq:Dsum}, \eqref{Reduzierer} and \eqref{eq:M2W} proves \eqref{representation1}.
\end{proof}
\begin{rem}
The Dirichlet spectrum contribution has a straightforward interpretation in the absence of magnetic fields. In that case, there is precisely one hexagonal eigenstate per fundamental cell. The ratio of fundamental cells per ball $B(R)$ scales exactly like $\frac{2}{3\sqrt{3}}$ in the $R \rightarrow \infty$ limit which coincides with the pre-factor determined in \eqref{representation1}. 
\end{rem}

We now proceed to the reduction to the effective Hamiltonian,
\begin{equation}
\label{eq:effha}
Q^{\rm{w}} ( x, h D ) - \Delta( z) , \ \ Q^{\rm{w}}( x , h D )  := \frac{1}3 \left(\begin{matrix} 0& 1 + e^{ i x } + e^{ i h D_x } \\
1 + e^{ - i x } + e^{ - i h D_x } & 0 \end{matrix}\right) ,
\end{equation}
which is the semiclassical quantization of
\begin{equation}
\label{eq:symbolQ}
Q(x,\xi):= \tfrac{1}{3} \left( \begin{matrix} 0 && 1+e^{ix}+ e^{i\xi} \\ 1+e^{-ix}+e^{-i\xi} &&  0  \end{matrix} \right) 
\end{equation}

The regularized trace, $\widehat \tr_{\mathcal V} $, in \eqref{representation1} can be expressed in terms of the regularized trace from Definition \ref{pdotrace} of pseudodifferential operators $Q^{\rm{w}}$:
\begin{lemm}
In the notation of Definition \ref{pdotrace}, Lemma \ref{Dirichletcalc} and
\eqref{eq:effha} we have 
\begin{align}
\label{eq:W2Q}
& \widehat{\operatorname{tr}}_{\mathcal{V}} \,  W'(z) W(z)^{-1}  = -\tfrac{2}{3\sqrt{3}}  \Delta'( z) \, \widehat{\operatorname{tr}} \,   (Q^{\rm{w}} ( x, h D ) - \Delta( z))^{-1}, \ \ z \in \CC \setminus \RR .
\end{align}
\end{lemm}
\begin{proof}
The explicit unitary transformation in Lemma \ref{uniteq} shows that
we can identify $ W ( z ) $ with a magnetic matrix $ A^h ( a - \Delta( z )  ) $ where $ a $ is given by \eqref{1}. The
limiting density  of vertices  in the hexagonal lattice is given by 
$\frac{4}{3 \sqrt{3}}$ and
half of this number corresponds to translates of each of $r_0$ and $r_1$.
Hence, 
\begin{equation*}
\widehat{\operatorname{tr}}_{\mathcal{V}} \left(W'(z) W(z)^{-1}\right)
= - \Delta'(z ) \tfrac{2}{3 \sqrt{3}}\widehat{\operatorname{tr}}_{\mathbb{Z}^2} \left( \left(A^h( a - \Delta ( z ) )\right)^{-1}\right)
\end{equation*}
We note that by \eqref{eq:in1} for $ z \notin \RR $,   $ (A^h( a - \Delta ( z ) ))^{-1} $ is
also a magnetic matrix. Formula \eqref{eq:W2Q} then follows from Lemma \ref{trmagex} and Definition \eqref{pdotrace}.
\end{proof}

Putting all this together we obtain the main result of this section:
\begin{prop}
\label{eq:fP2Q}
For $ f \in C^\infty_{\rm{c}}  ( \RR) $ with an almost analytic extension \eqref{eq:mather}, $ \widetilde f \in C^\infty_{\rm{c}} ( \RR) $, we have 
\begin{equation}
\label{tracediff}
\begin{split} 
\widetilde{\operatorname{tr}}(f(H^B)) & = \tfrac{2}{3 \sqrt{3} \pi } \int_{\mathbb{C}} {\partial_{\overline{z}}  \widetilde{f}(z)} \Delta'( z) \, \widehat{\operatorname{tr}} \,   (Q^{\rm{w}} ( x, h D ) - \Delta( z))^{-1}
d m (z) \\
& \ \ \ \ \ \ \ \ \ \ \ + \tfrac{2}{3 \sqrt{3} } 
 \sum_{\lambda \in \operatorname{Spec}(\Lambda^D_{(0,1)})} f(\lambda) ,
\end{split} \end{equation}
where $ Q ( x, \xi) $ is given by \eqref{eq:symbolQ} and $ \Delta ( z )$ by
\eqref{eq:Floq}.
\end{prop}

\section{Analysis of the effective Hamiltonian}
\label{anal}

We now study the effective Hamiltonian \eqref{eq:effha} 
for $ z $ near $ z_0 $ with $\Delta ( z_0 ) = 0 $. 
The goal is to obtain asymptotics of of the renormalized trace of $ ( Q^{\rm{w}} - \Delta(z))^{-1} $ -- see 
Theorem \ref{th:trac1} where for the moment we replace $\Delta(z)$ by $z$.
For that we use the strategy of Helffer--Sj\"ostrand outlined in 
\cite[\S 8]{HS2} but rather than follow \cite[\S 2]{HS0} and other numerous references
cited in \cite[\S 8]{HS2} we present direct arguments.

We start with some elementary analysis of the symbol $ Q $ given in 
\eqref{eq:symbolQ}.
Its determinant is given by 
$- {| 1 + e^{ i x} + e^{ i \xi } |^2}/9 $, 
and it vanishes at 
\[  ( x , \xi ) \in  \mathbb Z^2_* \pm \left( \tfrac{2 \pi } {3} , - \tfrac{ 2 \pi } 3 \right) , \]
that is, at the Dirac points. 

We consider neighbourhoods of $ \pm ( \frac{2 \pi } {3} , - \frac{ 2 \pi } 3 )  $ and make a symplectic change of variables:
\[  y = a ( x + \xi ) , \ \ \eta = b \left( \xi - x \pm \frac{4 \pi } 3 \right) , \ \ 
2 ab = 1, \]
we see that
\begin{equation}
\label{eq:norfo}
\begin{split}  1 + e^{ i x } + e^{ i \xi } & = 
c ( \eta \mp i y )  + \mathcal O ( y^2 + \eta^2 ) , \\
1 + e^{ -i x } + e^{ - i \xi } & = c ( \eta \pm i y )  + \mathcal O ( y^2 + \eta^2 )  , 
\end{split}
\end{equation}
where $ c = 3^{\frac14} 2^{-\frac12} $ and we chose
$ a = \pm 2^{-\frac34} 3^{-\frac14} $ and $ b = \pm 2^{-\frac14} 3^{\frac14} 
$. 

To study regularized traces of the resolvent of $ Q ( x, h D ) $ we introduce a localized operator with discrete spectrum near $ 0 $: Its Weyl symbol is given by 
\begin{equation}
\label{eq:defQ0}
\begin{gathered} 
Q_{0}  ( x, \xi ) := Q ( x, \xi) + 
\left(\begin{matrix} -1  + \chi_{0} ( x, \xi) & 0  \\
 0 & 1 - \chi_{0}   ( x, \xi) \end{matrix}\right), \\ 
\chi_{0} \in C^\infty_{\rm{c}}  ( \RR^2 ; [ 0, 1 ] ) , \ \ \chi_{0} ( \rho ) = \chi_{0} ( -\rho ) ,  \ \ 
\chi_{0} ( \rho ) = \left\{ \begin{array}{ll} 1, &\left\lVert \rho\right\rVert_\infty < \pi + \frac1{10}  , \\
0, & \left\lVert \rho\right\rVert_{\infty} > \pi + \frac2{10} , \end{array} \right.  
\end{gathered}
\end{equation}
where $ \rho = ( x, \xi ) $.

We observe that for any $ \delta > 0 $, there exists $ \epsilon > 0 $ such that 
\begin{equation}
\label{eq:ellQ0}  \det Q_{0} ( x, \xi )    < -  \epsilon  \ \text{ for } \left\lvert x \mp \frac{2 \pi}3\right\rvert + \left\lvert \xi \pm \frac{ 2 \pi } 3 \right\rvert > \delta . \end{equation}
This means that $ \det ( Q_{0}( x, \xi ) - z ) \in S ( 1 ) $ is elliptic (in the sense of \cite[\S 4.7.1]{ev-zw}) away from neighbourhoods
of $ \pm ( \frac{2 \pi } {3} , - \frac{ 2 \pi } 3)  $ and for $ z $ in a neighbourhood of $ 0 $. 

We also use microlocal weights defined as follows (see \cite[\S 8.2]{ev-zw}):
\begin{equation}
\label{eq:defG}
\begin{gathered} 
G ( x, \xi ) = \frac{1}2 \log ( 1 + \xi^2 + x^2 ) , \ \ G^{\rm{w}} = 
G^{\rm{w}} ( x, h D ) , \\ e^{ \pm N G^{\rm{w}} } = s_N ( x, h D , h ) , \ \
s_N \in S ( ( 1 + \xi^2 + x^2 )^{\pm N/2} ) .
\end{gathered}
\end{equation}

\begin{prop}
\label{p:specQ0}
For $ \delta_{0} > 0 $ small enough, the spectrum of $ Q_{0}^{\rm{w}} ( x, h D ) $ in 
$ [ - \delta_{0} , \delta_{0} ] $ is discrete and 
\begin{equation}
\label{eq:specQ0}
\Spec ( Q_{0}^{\rm{w}} ( x, h D ) ) \cap [ - \delta_{0} , \delta_{0} ] 
= \{  \kappa ( n h , h ) + \mathcal O (h^\infty) : n \in \mathbb Z \} \cap [ - \delta_{0} , \delta_{0} ] , 
\end{equation}
with eigenvalues of multiplicity $ 2 $,   $ \kappa ( - \zeta, h ) = - \kappa ( \zeta , h )  $, and
\begin{equation}
\label{eq:g2F}  
\begin{gathered} F ( \kappa ( \zeta, h )^2 , h ) = |\zeta| + \mathcal O 
(h^\infty)  , \ \ 
F ( \omega , h ) \sim  F_0 ( \omega ) + \sum_{j=2}^\infty h^j F_j ( \omega ) , \ \ F_j \in C^\infty ( \RR ) , \\ 
F_0 ( \omega  ) =  \frac{1}{4 \pi} \int_{ \gamma_\omega } \xi dx , \ \ 
\gamma_\omega = \left\{ ( x, \xi) \in \mathbb{T}^2_*: \frac{| 1 + e^{ix } + e^{i\xi} |^2}9 = \omega \right\}, \ \ F_j ( 0 ) = 0 , 
\end{gathered} 
\end{equation}
where $ \gamma_\omega $ is oriented clockwise in the $ ( x, \xi ) $ plane. 

Moreover, the orthonormal set of eigenfunctions,  $ (u_n^+ ( h ))
_{ n \in \mathbb Z } \cup ( u_n^-(h) )_{ n \in \mathbb Z } $, satisfies 
\begin{equation}
\label{eq:QM}
\begin{gathered}
Q_{0}^{\rm{w}} ( x, h D , z_0 ) u_n^\pm (h) = \kappa ( nh, h ) u_n^\pm ( h ),  \ \ \ \WF_h ( u_n^\pm ) \subset {\rm{nbhd} } 
\left( \pm \left( \frac{2 \pi } {3} , - \frac{ 2 \pi } 3 \right) \right),  
\end{gathered}
\end{equation}
and, for all $ N $, 
\begin{equation}
\label{eq:unloc}
\begin{split}
& \|  ( 1 - \chi_{0}^{\rm{w}} ( x, h D ) ) e^{ N G^{\rm{w}} ( x, h D ) }  u^\pm_n ( h ) \|
= \mathcal O_N  ( h^\infty ) , \\
& \|   e^{ N G^{\rm{w}} ( x, h D ) } ( 1 - \chi_{0}^{\rm{w}} ( x, h D ) )  u^\pm_n ( h ) \|
= \mathcal O_N  ( h^\infty ) ,
\end{split}
\end{equation}
where $ \chi_{0} $ is defined in \eqref{eq:defQ0} and $ G$ in \eqref{eq:defG}.
\end{prop}
\begin{proof}
We start by showing that for $ \delta_{0} $ small enough the spectrum 
of $ Q_{0}^{\rm{w}} $ in $ [ - \delta_{0} , \delta_{0} ]$ is discrete and that
the eigenfunctions are localized to neighbourhoods in the sense of 
\eqref{eq:QM} and \eqref{eq:unloc}. 
For that we define $ Q_{1} := Q + {\rm{diag}} ( -1 , 1 ) $. Then 
$ Q_{0} = Q_{1} + {\rm{diag}} ( \chi_{0} , - \chi_{0} ) $ and $ Q_{1} - z $ is 
elliptic in $ S ( 1 ) $ for $ |z| $ small enough. 
That implies that for $ 0 < h < h_0 $, $ ( Q_{1}^{\rm{w}} - z)^{-1} 
= \mathcal O ( 1 )_{L^2 \to L^2}$ in $h$ --  see \cite[\S 4.7.1]{ev-zw}. It follows
that 
\[  Q_{0}^{\rm{w}} - z = ( Q_{1}^{\rm{w}} - z ) ( \operatorname{id} + K ( z ) ) ,  \ \ 
K ( z ) := (Q_{1}^{\rm{w}} - z )^{-1} {\rm{diag}} ( \chi_{0}^{\rm{w}} , - \chi_{0}^{\rm{w}}  ) 
. \] 
Since $ \chi_{0} ^{\rm{w}} $ is a compact operator on $ L ^2 $ (see \cite[Theorem 4.26]{ev-zw}) we can use analytic Fredholm theory (see
\cite[Theorem D.4]{ev-zw}) to show that 
$ ( \operatorname{id} + K ( z ) )^{-1} $ is meromorphic. That shows that $ (Q_{0}^{\rm{w}} - z )^{-1} $ is meromorphic for $ |z| $ small, that it has a discrete set of poles there, which in turn means that the spectrum near $ 0 $ is discrete.

The comment after \eqref{eq:ellQ0} and \cite[\S 8.4]{ev-zw} give the
localization of eigenfunctions in \eqref{eq:QM}. To see \eqref{eq:unloc} 
we consider the conjugated operator
\[ Q_G^{\rm{w}} - z :=  e^{  N G^{\rm{w}} } ( Q^{ \rm{w}}_{0} - z ) e^{ - NG^{\rm{w}} } .\]
From \cite[Theorems 4.18 and 8.6]{ev-zw} we see that $ Q_G \in S ( 1 ) $ 
and that $ Q_G = Q_{0} + \mathcal O_N ( h)_{ S(1)} $. Hence $ Q_G  - z$ is elliptic where $ Q_{0} - z$ is elliptic and in particular near the support of 
$ 1 - \chi_{0} $. Since $ (Q_G^{\rm{w}} - z ) e^{ N G^{\rm{w}} }
u  = 0$, $ z \in \Spec ( Q_{0}^{\rm{w}} ) $, $ u $ an eigenfunction, the first estimate in \eqref{eq:unloc} follows. To see the second estimate we use 
the wave front set estimate \eqref{eq:QM} and the fact that the 
essential support (see \cite[\S 8.4]{ev-zw}) of the commutator of  
$ \chi_{0}^{\rm{w}} $ and $ e^{ s G^{\rm{w}}} $ is supported away from 
$ \WFh ( u_n^\pm ) $.

This means that to {\em approximate} eigenvalues of $ Q_{0}^{\rm{w}} ( x, h D ) $ we need to find {\em all } microlocal solutions $ ( u , z ) $ (that is solutions modulo
$ \mathcal O ( h ^\infty ) $) such that $ u $ satisfies \eqref{eq:unloc} and  
\begin{equation}  
\label{eq:apeq} ( Q^{\rm{w}} - z ) u = \mathcal O (h^\infty ) , \ \ \  \WF_h ( u ) \subset {\rm{nbhd}} \left( \pm \left( \frac{2\pi}3, - \frac{ 2 \pi}3 \right) \right) .
\end{equation}
Here we replaced $ Q_{0} $ by $ Q $ since the corresponding operators
are microlocally the same near $ \pm ( \frac{2\pi}3, - \frac{ 2 \pi}3 )$
(see \cite[\S 8.4.5]{ev-zw} for a discussion of this concept).
Since $ Q_{0}^{\rm{w}} $ is self-adjoint the uniqueness of microlocal solutions gives uniquess of eigenfunctions as they have to be orthogonal.

We have 
\[  Q^{\rm{w}} = \left(\begin{matrix} 0 & \Lambda^{\rm{w}}_+ \\
\Lambda^{\rm{w}}_- & 0 \end{matrix}\right) , \ \ \Lambda_{\pm} ( x , \xi ) := \frac{1 + e^{ \pm i x } + e^{ \pm i \xi }}3 , \ \ 
(\Lambda_\pm^{\rm{w}} )^* = \Lambda_{\mp}^{\rm{w}} . \]
Because of the symmetry 
$ ( x , \xi ) \to ( - x, - \xi ) $
we will work microlocally near  $  ( \frac{2\pi}3, - \frac{ 2 \pi}3 )  $.
At that point
\eqref{eq:norfo} shows that 
the Poisson brackets of $ \Lambda_{\pm } $ satisfy
\begin{equation}
\label{eq:Poison}  \{ \Re \Lambda_+ , \Im \Lambda_+ \} < 0 , \ \  \ \{ \Re \Lambda_- , \Im \Lambda_- \} >  0 , \ \ \
 \frac 1 i \{ \Lambda_+, \Lambda_- \} > 0. 
 \end{equation}
 The last inequality is also known as H\"ormander's hypoellipticity condition.
Using \cite[\S\S 12.4 and 12.5]{ev-zw} we see that the first
two inequalities in \eqref{eq:Poison} show that 
there exist microlocally unique solutions 
\begin{equation}
\label{eq:u0}    \Lambda_+^{\rm{w}} u_0 = \mathcal O ( h^\infty ) , \ \ \
\WFh 
( u ) \subset {\rm{nbhd}} \left( \left( \frac{2\pi}3, - \frac{ 2 \pi}3 \right)\right) .  
\end{equation}
On the other hand the last inequality in \eqref{eq:Poison} shows that
\begin{equation}
\label{eq:laze} 
\WFh ( u ) \subset {\rm{nbhd}} \left(\left( \frac{2\pi}3, - \frac{ 2 \pi}3 \right)\right)
\ \Longrightarrow \ 
\langle \Lambda_+^{\rm{w}} \Lambda_- ^{\rm{w}} u , u 
\rangle \geq c_0 h \| u \|^2 ,  \end{equation}
see for instance the proof of \cite[Theorem 7.5]{ev-zw}. 
This
characterizes the microlocal kernel of $ Q^{\rm{w}} $ near  
$ ( \frac{2\pi}3, - \frac{ 2 \pi}3 ) $. Since
$ (Q^{\rm{w}})^*  Q^{\rm{w}} = {\rm{diag}} ( \Lambda_-^{\rm{w}} \Lambda_+^{\rm{w}}, 
\Lambda_+^{\rm{w}} \Lambda_-^{\rm{w}} ) $, this means that all solutions
to \eqref{eq:apeq} other than the unique solution $ ( 0, u_0 ) $ satisfy
$ |z | \geq c \sqrt h $. That gives the correspondence with microlocal 
solutions $ w $ (satisfying \eqref{eq:unloc}) to
\begin{equation}
\label{eq:Hpl} 
\begin{gathered} 
 H_+ w =  \lambda  w  , \ \ \WFh ( w ) \subset {\rm{nbhd}} \left(\left( \frac{2\pi}3, - \frac{ 2 \pi}3\right)\right), \ \ 
H_+ :=\Lambda_+^{\rm{w}} \Lambda_-^{\rm{w}} , \\
\left(\begin{matrix} 0 & \Lambda^{\rm{w}}_+ \\
\Lambda^{\rm{w}}_- & 0 \end{matrix}\right) \left(\begin{matrix} u_1 \\ u_2 \end{matrix}\right) 
= z \left(\begin{matrix} u_1 \\ u_2 \end{matrix}\right) , \ \ \ \WFh ( u_j ) 
\subset {\rm{nbhd}} \left(\left( \frac{2\pi}3, - \frac{ 2 \pi}3\right)\right) , \\
z = \pm \sqrt \lambda, \ \ u_1= w , \ \  u_2 = z^{-1} \Lambda_-^{\rm{w}} w . 
\end{gathered}
\end{equation}
Recalling \eqref{eq:norfo} we see that $ H_+ $, microlocally near  
$  ( \frac{2\pi}3, - \frac{ 2 \pi}3 )$ has the structure of a potential 
well and the distribution of eigenvalues near $ 0 $ has been extensively studied. Following earlier works of Weinstein \cite{we} and Colin de 
Verdi\`ere \cite{cdv} the semiclassical version was given by Helffer--Robert \cite{hr} and a clear outline can be found in \cite[\S 8, Case II, p.292]{BGHKS}. In particular, there exists a function $ F $ with an expansion 
$ F ( \omega, h ) \sim F_0 ( \omega ) + h F_1  +  h^2 F_2 ( \omega ) \cdots $, where $ F_1 $ is a constant (see \cite[Corollaire (3.15)]{hr})
such that $ \mathcal O ( h^\infty ) $ 
quasimodes of $ H_+ $ are given by the quantization condition 
$ F ( \lambda_n ( h ) , h ) = n h $, $ n = 0, 1 , \cdots $. 
Since we have shown that $ \lambda_0 ( h ) = \mathcal O ( h^\infty ) $
we obtain that $ F_j ( 0 ) = 0$ for all $ j $.
That gives \eqref{eq:g2F}.
\end{proof}

The spectrum and eigenfunctions of $ Q_{0}^{\rm{w}} $ will now be used 
to describe $ ( Q^{\rm{w}} - z )^{-1} $ for $ |\Im z | > h^{M} $ for any fixed
 $M $.

We first show that away from the spectrum of $ Q^{\rm{w}}_{0} $,
$  Q^{\rm{w}} - z $ is invertible. The proof is a simpler 
version of the proof of Proposition \ref{l:Gru} and the estimates are similar.

\begin{lemm}
\label{l:awayspec}
Let $ 0 < \delta_{1} < \delta_{0} $ 
and suppose that $ z \in [ - \delta_{1} , \delta_{1} ] - i [ -1,1] $ satisfies
\[   d ( z , \Spec( Q_{0}^{\rm{w}} ( x, h D ) ) ) > h^{N_0} , \]
for some fixed $ N_0 $. Then for $ 0 < h < h_0 $, 
\[  ( Q^{\rm{w}} ( x, h D ) - z )^{-1} = \mathcal O (  d ( z , \Spec( Q_{0}^{\rm{w}} ( x, h D ) ))^{-1}  )_{ L^2 \to L^2 } . \]
\end{lemm}
\begin{proof}
In addition to $ Q_{0}^{\rm{w}} $ we define another auxiliary operator with the 
symbol
\begin{equation}
\label{eq:defQ1} 
\begin{gathered} Q_{1} ( x, \xi ) := Q_{0} ( x, \xi) + 
\left(\begin{matrix}  - \chi_{1} (x, \xi)  & 0 \\
0  & \chi_{1} ( x, \xi)  \end{matrix}\right), \\
\chi_{1} \in C^\infty_{\rm{c}}  ( \RR^2 ; [ 0, 1 ] ) , \ \ \chi_{1} ( \rho ) = \chi_{1} ( -\rho ) ,  \ \ 
\chi_{1} ( \rho ) = \left\{ \begin{array}{ll} 
1,  &  \left\lVert \rho\right\rVert_\infty < \pi - \frac2{10} ,\\ 
0 , &   \left\lVert \rho \right\rVert_\infty  >  \pi - \frac1{10} ,  \end{array} \right. 
\end{gathered}
\end{equation}
noting that $ Q_{1} ( x, \xi ) - z \in S ( 1 ) $ is now elliptic (in the sense 
that the determinant, $ z^2 - \chi_{1}^2 + \det Q_{0},$ satisfies the conditions of \cite[\S 4.7.1]{ev-zw} for $ z $ in a neighbourhood of $ 0 $). 
From \cite[Theorems 4.29, 8.3]{ev-zw} we conclude that
\begin{equation}
\label{eq:invQ1}
( Q^{\rm{w}}_{1} ( x, hD ) - z )^{-1} = R_{1}^{\rm{w}} (z; x, h D , h ) , \ \
R_{1} \in S ( 1 ) 
, \ \  z \in [ - \delta_{1} , \delta_{1} ] - i [ -1,1]  . 
\end{equation}

Using $ Q_{0}^{\rm{w}} $ and $ Q_{1}^{\rm{w}} $ we define
\begin{equation}
\label{eq:PPP}
\begin{split}
&  p = p ( z ; x, \xi ) := Q ( x, \xi ) - z, \\
&  p_{0}^{\gamma}  = p_{0}^{\gamma}  ( z; x, \xi ) := Q_{0} ( x - \gamma_1 , \xi - \gamma_2 )  - z   \\ 
&  p_{1}^\gamma = p_{1}^{\gamma} ( z ; x , \xi ) := 
Q_{1} ( x - \gamma_1 , \xi - \gamma_2 )  - z   . 
\end{split}
\end{equation}
We denote the Weyl quantizations by $ P = P ( z ) $, $ P^\gamma_{0}  = P^\gamma_{0} ( z ) $ 
and $ P^\gamma_{1} =  P^\gamma_{1} ( z ) $ and note that
\begin{equation}
\label{eq:translph}
  P^\gamma_{0} = r_{\gamma} ( Q^{\rm{w}}_{0} - z ) r_{- \gamma } , \ \ \ 
 P^\gamma_{1} = r_{\gamma } ( Q_{1}^{\rm{w}} - z ) r_{-\gamma } , 
\ \ \ r_\gamma u ( x ) := e^{ \frac i h \gamma_2 x } u ( x - \gamma_1 ) .   \end{equation} 
We always assume that $ z \in [ - \delta_{1} , \delta_{1} ] - i [ -1,1]  $.

We now choose 
$  \chi, \tilde \chi \in C^\infty_{\rm{c}} ( \RR^2 ) $ so that
\begin{equation}
\label{eq:suppchi} 
\begin{gathered}
\tilde \chi|_{{\rm{nbhd}}(\supp \chi)} = 1 , \ \ \ 
\chi_{0} |_{{\rm{nbhd}}(\supp \tilde \chi) } = 1 , \\
 \sum_{\gamma \in \ZZ_*^2 }  \chi_\gamma = 1 , \ \  
\chi_\gamma ( x, \xi )  := \chi ( x - \gamma_1 , \xi - \gamma_2 ) . 
\end{gathered}
\end{equation}
We also define translations $ \tilde \chi_\gamma ( x, \xi) := \tilde \chi (x - \gamma_1, 
\xi - \gamma_2 ) $ and note that for all $ N $ and
with semi-norms independent of $ \gamma $,
\begin{equation}
\label{eq:malph}
\chi_\gamma , \tilde \chi_\gamma \in S ( m_\gamma^{-N} ) , \ \ 
m_\gamma ( x, \xi ) := ( 1 + (x-\gamma_1)^2 + ( \xi - \gamma_2 )^2 )^{\frac12}
\end{equation}

The properties of the cut-off functions guarantee that
\begin{equation}
\label{eq:PPa}  (p - p^\gamma_{0} )|_{{\rm{nbhd}}(\supp \tilde \chi_\gamma ) } = 0 , \ \
( p^\gamma_{0} - p^\gamma_{1})|_{{\rm{nbhd}}(\supp \nabla \tilde \chi_\gamma) } = 0 .
\end{equation}
Combined with \eqref{eq:invQ1} the composition formula for pseudodifferential
operators \cite[Theorem 4.18]{ev-zw} gives 
\begin{equation} 
\label{eq:commalph}
\begin{gathered}
e_{1, \gamma}^{\rm{w}}:=( P- P^\gamma_{0} ) \tilde \chi_\gamma^{\rm{w}}  , \ \ 
 e_{2, \gamma}^{\rm{w}}:=\tilde \chi_\gamma^{\rm{w}}  \chi_\gamma^{\rm{w}}  -
\chi_\gamma^{\rm{w}} , \\ 
e_{3 , \gamma}^{\rm{w}}:=[  P^\gamma_{0}, \tilde \chi_\gamma^{\rm{w}}   ] 
\widetilde P_\gamma^{-1} \chi_\gamma^{\rm{w}}  , \ \
e_{4,\gamma}^{\rm{w}}:=[ P^\gamma_{0} , \tilde \chi_\gamma^{\rm{w}}  ] 
\left(P^\gamma_{0}\right)^{-1} 
(  P^\gamma_{1}  -  P^\gamma_{0} )    ,
\end{gathered}
\end{equation}
where 
$ e_{j, \gamma} \in h^N S( m_\gamma^{-N} ) $, for all $ N $.

If $ d ( z ,\Spec ( Q_{0}^{\rm{w}} ) ) > h^{N_0} $ we define
$ F^0 := \sum_{ \gamma \in \ZZ^2_*} \tilde \chi_\gamma^{\rm{w}} \,
 \left(P^\gamma_{0}\right)^{-1} \, \chi_\gamma^{\rm{w}}  $,
  where the inverse of $ P^\gamma_{0} $ exists in view of \eqref{eq:translph}.
We claim that  
\begin{equation}
\label{eq:apprinv}
F^0 := \sum_{ \gamma \in \ZZ^2_*} \tilde \chi_\gamma^{\rm{w}} \,
\left(P^\gamma_{0}\right)^{-1} \, \chi_\gamma^{\rm{w}} 
  = \mathcal 
 O (  d ( z , \Spec( Q_{0}^{\rm{w}}))^{-1}  )_{ L^2 \to L^2 } .
 \end{equation}
In fact, in view of \eqref{eq:malph}
\begin{equation}
\label{eq:chichi}  
\tilde \chi_\gamma^{\rm{w}} (\tilde \chi_\beta^{\rm{w}})^* = 
(a^1_{\gamma \beta})^{\rm{w}}, \ \ 
 \chi_\gamma^{\rm{w}} (\chi_\beta^{\rm{w}})^* = 
(a^2_{\gamma \beta})^{\rm{w}}, \ \ a^j_{\gamma \beta} \in 
S ( m_\gamma^{-N} m_\beta^{-N} ) .\end{equation}
From \cite[Theorem 4.23]{ev-zw}
\begin{equation}
\label{eq:ajab} \| (a^j_{\gamma \beta})^{\rm{w}} \|_{ L^2 \to L^2 } 
\le C  \sup_{\RR^{2}} m_\gamma^{-N} m_\beta^{-N} \le C_N \langle \gamma - \beta 
\rangle^{-N} , \end{equation}
for all $ N \in \mathbb{N} $. If we put $ A_\gamma :=  \tilde \chi_\gamma^{\rm{w}} \,
 \left(P^\gamma_{0}\right)^{-1} \, \chi_\gamma^{\rm{w}} $, if follows that
 \begin{equation} 
 \label{eq:AAst} A_\gamma^* A_\beta ,  A_\gamma A_\beta^*  = 
 \mathcal O (  d ( z , \Spec( Q_{0}^{\rm{w}}))^{-2}  
 \langle \gamma - \beta 
\rangle^{-N})_{L^2 \to L^2} ,\end{equation}
and \eqref{eq:apprinv} follows from an application of the Cotlar--Stein
Lemma -- see \cite[Theorem C.5]{ev-zw}.

Using the notation of \eqref{eq:commalph} we have 
\begin{equation}
\label{eq:PF0}
\begin{split}
P F^0 & = \sum_{ \gamma \in \ZZ^2_*}P^\gamma_{0} \tilde \chi_\gamma^{\rm{w}} \,
 \left(P^\gamma_{0}\right)^{-1} \, \chi_\gamma^{\rm{w}} +  e_{1, \gamma}  \left(P^\gamma_{0}\right)^{-1} \, \chi_\gamma^{\rm{w}} \\
 & = 
 \sum_{ \gamma \in \ZZ^2_*} \chi_\gamma^{\rm{w}} +  e_{1, \gamma}^{\rm{w}}  \left(P^\gamma_{0}\right)^{-1} \, \chi_\gamma^{\rm{w}}  + e_{2, \gamma }^{\rm{w}} + 
   [ P^\gamma_{0}, \tilde \chi^{\rm{w}}_\gamma ] \left(P^\gamma_{0}\right)^{-1} 
 \chi_\gamma^{\rm{w}}  \\
 & = \operatorname{id} + \sum_{ \gamma \in \ZZ^2_*}  e_{1, \gamma}^{\rm{w}}  \left(P^\gamma_{0}\right)^{-1} \, \chi_\gamma^{\rm{w}} +
 e_{2, \gamma }^{\rm{w}} + [ P^\gamma_{0}, \tilde \chi^{\rm{w}}_\gamma ]\left(P^\gamma_{1}\right)^{-1}  
 \chi_\gamma^{\rm{w}} 
\nonumber\\
 & \qquad  +  \sum_{ \gamma \in \ZZ^2_*}[ P^\gamma_{0}, \tilde \chi^{\rm{w}}_\gamma ] ( \left(P^\gamma_{0}\right)^{-1} - \left(P^\gamma_{1}\right)^{-1} )   \chi_\gamma^{\rm{w}}\\
 & = \operatorname{id} +  \sum_{ \gamma \in \ZZ^2_*} e_{1, \gamma}^{\rm{w}}  \left(P^\gamma_{1}\right)^{-1} \, \chi_\gamma^{\rm{w}} + 
 e_{2, \gamma }^{\rm{w}} + e_{3, \gamma}^{\rm{w}} + 
 [ P^\gamma_{0}, \tilde \chi^{\rm{w}}_\gamma ] \left(P^\gamma_{1}\right)^{-1} (
P^\gamma_{1} -  P^\gamma_{0} )   \left(P^\gamma_{0}\right)^{-1} \chi_\gamma^{\rm{w}} \\
 & = \operatorname{id} +  \sum_{ \gamma \in \ZZ^2_*} e_{1, \gamma}^{\rm{w}}  \left(P^\gamma_{0}\right)^{-1} \, \chi_\gamma^{\rm{w}} +
 e_{2, \gamma }^{\rm{w}} + e_{3, \gamma}^{\rm{w}} +  e_{4, \gamma}^{\rm{w}}
  \left(P^\gamma_{0}\right)^{-1} \chi_\gamma^{\rm{w}} 
  \\ & = 
  \operatorname{id} + r , \ \ \  r  =  
  \mathcal O ( h^\infty d ( z , \Spec( Q_{0}^{\rm{w}}))^{-1}  )_{ L^2 \to L^2 } ,
  \end{split}
 \end{equation}
where the bound on $ r $ follows from \eqref{eq:commalph} and 
\eqref{eq:ajab} and an application of the Cotlar--Stein Lemma as in the proof of \eqref{eq:apprinv}.

 Hence for $ h $ small enough, 
 \begin{gather*} ( Q^{\rm{w}} ( x, h D ) - z )^{-1} = 
 F^0 ( \operatorname{id} + r)^{-1} = \mathcal O (  d ( z , \Spec( Q_{0}^{\rm{w}}))^{-1}  )_{ L^2 \to L^2 } ,
 \end{gather*}
 for 
$ z \in [ - \delta_{1} , \delta_{1} ] - i [ -1,1]  $,  $  d ( z ,
\Spec ( Q_{0}^{\rm{w}} ) ) > h^{N_0} $. 
 \end{proof}

 The proof gives a stronger weighted estimate on the inverse
with similar estimates being crucial later.
Under the assumption of Lemma \ref{l:awayspec} we have, for any $ s \in \RR $
and $ G^{\rm{w}} $ defined in \eqref{eq:defG}
\begin{equation}
\label{eq:weighted1}
e^{-s G^{\rm{w}} } ( Q^{\rm{w}} - z )^{-1} e^{s G^{\rm{w}} } = 
\mathcal O (  d ( z , \Spec( Q_{0}^{\rm{w}} ( x, h D ) ))^{-1}  )_{ L^2 \to L^2 } .
\end{equation}
\begin{proof}[Proof of \eqref{eq:weighted1}] We first check that $ F^0 $ defined in \eqref{eq:apprinv}
satisfies this estimate. (We note that \eqref{eq:weighted1} does not seem to  follow easily from conjugating $ Q^{\rm{w}} - z $ by the weight.) For that
we make the following observations:
\begin{equation}
\label{eq:weighchi}
\begin{gathered}
e^{ s G^{\rm{w}} } \tilde \chi_\gamma^{\rm{w}} = ( \tilde \chi^s_\gamma)^{\rm{w}}  , \ \ \chi_\gamma^{\rm{w}} e^{  s G^{\rm{w}} } = 
( \chi^s_\gamma)^{\rm{w}} , \\ \tilde \chi_\gamma^s , 
\chi_\gamma^s \in \bigcap_{ N } S ( e^{ s G } m_\gamma^{-N} ) = \langle \gamma \rangle^{s} \bigcap_N  S ( m_\gamma^{-N} ) , 
\end{gathered}
\end{equation}
where the equality of symbols spaces follows from the fact that
$  e^{s G ( \rho ) }  = \langle \rho \rangle^{s} $ and 
\begin{equation}
\label{eq:ineq} \langle \rho \rangle^{s}  \langle \rho - \gamma \rangle^{-N} \leq 
\langle \gamma \rangle^{s} \langle \rho - \gamma \rangle^{-N+|s|} \leq 
\langle \rho \rangle^{s} \langle \rho - \gamma \rangle^{-N+2|s|} .
\end{equation}
Proceeding as in \eqref{eq:chichi} and \eqref{eq:ajab} and putting
$ A_\gamma^s := e^{-s G^{\rm{w}} } A_\gamma e^{ s G^{\rm{w}}} $ we see
that estimates \eqref{eq:AAst} hold for $ A_\gamma^s $. That shows that
$ e^{ - s G^{\rm{w}} } F^0 e^{ s G^{\rm{w}} } $ is  bounded on $ L^2$ for 
any $ s \in \RR $. The same argument applies to $ r $ in \eqref{eq:PF0}
and that concludes the proof of \eqref{eq:weighted1}. 
\end{proof}

\medskip

We now use the translates of $ w_n^\pm $ from Proposition  \ref{p:specQ0} to construct a Grushin problem 
for $ Q^{\rm{w}} - z $ for $ z $ near $ \Spec (Q_{0}^{\rm{w}} )) $.
For that we take $ z_1 $ and $ \epsilon_0 $ such that
\begin{equation}  
\label{eq:specrah}     \{ \kappa ( n h, h ) \}_{ n \in \mathbb Z }  \cap [ z_1 - 
2 \epsilon_0 h , z_1 +  2 \epsilon_0 h ]  = \{ \kappa ( n_1 h, h ) \} , \ \ 
n_1 = n_1 ( z_1 , h ) .
\end{equation}
The interval $ [ - \delta_{0} , \delta_{0} ] $ can be covered by intervals of this form and intervals of size $ h $, disjoint from $ \Spec (Q_{0}^{\rm{w}} ) $.

For $ \gamma \in \mathbb Z^2_* $ we use translation \eqref{eq:translph} and
put 
\begin{equation}
\label{eq:wgamma}
w_\gamma = w_\gamma (h) := \left(\begin{matrix} w_\gamma^+ ( h )  , w_\gamma^- ( h ) \end{matrix} \right)
= \left(\begin{matrix}r_\gamma u^+_{n_1} ( h ), r_\gamma u^-_{n_1} ( h ) \end{matrix} \right) \in \CC^2 \otimes \CC^2 , 
\end{equation}
where $ n_1 $ is defined by \eqref{eq:specrah}.

The following lemma will be useful in several places:
\begin{lemm}
\label{l:ortho}
With $ w_\gamma^\pm $ defined by \eqref{eq:wgamma} and $ G $ given in 
\eqref{eq:defG} we have, for every $ s \in \RR $, 
\begin{equation}
\label{eq:almostw}
\begin{split} 
\langle e^{ sG^{\rm{w}} } w_\gamma^\pm, e^{ s G^{\rm{w}} } w_\beta^\pm \rangle & = \mathcal O (  \langle \gamma \rangle^{2s} \delta_{ \gamma \beta} + h^\infty 
\langle \gamma \rangle^{2s} \langle\gamma - \beta \rangle^{-\infty} ), \\ 
 \langle e^{ sG^{\rm{w}} }  w_\gamma^+, e^{ sG^{\rm{w}} }  w_\beta^- \rangle &  = \mathcal O (h^\infty  \langle \gamma \rangle^{2s}  
\langle \gamma - \beta \rangle^{-\infty } ) ,
\\
 \langle e^{ sG^{\rm{w}} } ( 1 - \chi_\gamma^{\rm{w}} )  w_\gamma^\epsilon, e^{ sG^{\rm{w}} }  ( 1 - \chi_\beta^{\rm{w}} )  w_\beta^{\epsilon'} \rangle & = \mathcal O (h^\infty  \langle \gamma \rangle^{2s}  
\langle \gamma - \beta \rangle^{-\infty } ),   \ \ \epsilon, \epsilon' \in \{ +, -\} .
\end{split} 
\end{equation}
\end{lemm}
\begin{proof}
This follows from \eqref{eq:QM}, \eqref{eq:unloc} and arguments presented in 
the remark above. As an example we prove the first estimate in \eqref{eq:almostw} (dropping $ \pm $ in the notation):
\[ \begin{split} 
\langle e^{ sG^{\rm{w}} } w_\gamma, e^{ s G^{\rm{w}} } w_\beta \rangle & = \langle e^{ s G^{\rm{w}} } ( 1 - \chi_\gamma^{\rm{w}} ) w_\gamma , 
e^{ s G^{\rm{w}} } ( 1 - \chi_\beta^{\rm{w}} ) w_\beta \rangle 
+ 
\langle e^{ s G^{\rm{w}} }  \chi_\gamma^{\rm{w}}  w_\gamma , 
e^{ s G^{\rm{w}} }  \chi_\beta^{\rm{w}}  w_\beta \rangle 
\\
& \ \ \ \ + 
\langle e^{ s G^{\rm{w}} }  \chi_\gamma^{\rm{w}}  w_\gamma , 
e^{ s G^{\rm{w}} } ( 1 - \chi_\beta^{\rm{w}} ) w_\beta \rangle +
\langle e^{ s G^{\rm{w}} } ( 1 - \chi_\gamma^{\rm{w}} ) w_\gamma , 
e^{ s G^{\rm{w}} } \chi_\beta^{\rm{w}}  w_\beta \rangle 
\end{split} \]
With $ G_\gamma ( \rho ) := G ( \rho - \gamma ) $, 
\[ \begin{split} 
 e^{ s G^{\rm{w}} - N G_\gamma^{\rm{w}} } e^{ N G_\gamma ^{\rm{w}} } ( 1 - \chi_\gamma^{\rm{w}} ) w_\gamma = b_\gamma^{\rm{w}} ( x, h D )  e^{ N G_\gamma ^{\rm{w}} } ( 1 - \chi_\gamma^{\rm{w}} ) w_\gamma , 
 \end{split}
 \]
 where as in  \eqref{eq:ineq},
 \[ b_\gamma \in  S ( \langle \rho \rangle^s \langle \rho - \gamma \rangle^{-N } )  \subset  
 S ( \langle \gamma \rangle^s \langle \rho - \gamma \rangle^{-N + |s| } ) .  \]
Putting 
\[ \tilde w_\gamma := e^{ N G_\gamma ^{\rm{w}} } ( 1 - \chi_\gamma^{\rm{w}} ) w_\gamma = \mathcal O ( h^\infty )_{L^2}  , \]
and $ M = N - |s| \gg 1 $, we see that
\[ \begin{split}  \langle e^{ s G^{\rm{w}} } ( 1 - \chi_\gamma^{\rm{w}} ) w_\gamma , 
e^{ s G^{\rm{w}} } ( 1 - \chi_\beta^{\rm{w}} ) w_\beta \rangle 
& = \langle ( b_\beta^{\rm{w}})^* b_\gamma^{\rm{w}} \tilde w_\gamma , 
\tilde w_\beta \rangle
 = \| (b_\beta^{\rm{w}})^* b_\gamma^{\rm{w}} \|_{L^2 \to L^2 } \mathcal O 
(h^\infty ) \\ & 
 \leq C \sup_{\rho \in \RR^2}  \langle \gamma \rangle^{s} \langle
\beta \rangle^{s} \langle \rho - \gamma\rangle^{ - M } \langle \rho - \beta\rangle^{ - M } \mathcal O 
(h^\infty )  \\
& \leq \mathcal O ( h^\infty  \langle \gamma \rangle^{2s} \langle \gamma - \beta \rangle^{-M+s} )  .
\end{split} \]
The other terms are treated in the same way. \end{proof}

We then define $ R_+ : L^2 ( \RR , \CC^2 ) \to \ell^2 ( \mathbb Z^2_*; \CC^2 ) $ and
$ R_- = \ell^2 ( \mathbb Z^2_*; \CC^2 ) \to L^2 ( \RR , \CC^2 ) $ as follows
\begin{equation}
\label{eq:defRp}
\begin{split}
\left( R_+ u \right) ( \gamma ) & := \langle 
u , w_\gamma \rangle:= \left(\begin{matrix} \langle u , w_\gamma^+ \rangle \\
\langle u , w_\gamma^- \rangle  \end{matrix}\right) \in \CC^2 ,  
\ \ 
R_- u_- ( x ) 
:= \sum_{ \gamma \in \mathbb Z^2_* }
  w_\gamma ( x )  u_- ( \gamma ) , 
\end{split}
\end{equation}
where $  u_- ( \gamma ) = \left(\begin{matrix}   u_-^{+} ( \gamma) ,  u_-^- (\gamma ) \end{matrix}\right)^t 
\in \CC^2 $ and $ w_\gamma ( x ) = ( w_\gamma^+ , w_\gamma^- ) \in 
\CC^2 \otimes \CC^2 $. 

To see the boundedness of $ R_- $ we use the almost orthogonality of 
$ w_\gamma^\pm $ given in \eqref{eq:almostw} with $ s = 0 $:
Hence, 
\begin{equation}
\label{eq:below} \begin{split}  \left\lVert  \sum_{\gamma \in \mathbb{Z}_*} w_\gamma ( \bullet )  u_- ( \gamma )  \right\rVert_{L^2 }^2 
& \lesssim \sum_{\gamma \in \mathbb{Z}_*} \sum_{\beta \in \mathbb{Z}_*} |u_- ( \gamma )|| u_- ( \gamma + \beta ) | \langle \beta \rangle^{-N}  \\
& \lesssim \| u_-\|_{\ell^2 } 
\left( \sum_{\gamma \in \mathbb{Z}_*} \left( \sum_{\gamma \in \mathbb{Z}_*} | u_- ( \gamma + \gamma ) | \langle \gamma
\rangle^{-N} \right)^2 \right)^{\frac12} \\
& \lesssim  \| u_-\|_{\ell^2 } \left( \sum_{\gamma \in \mathbb{Z}_*} 
\sum_{\gamma \in \mathbb{Z}_*} \sum_{\gamma' \in \mathbb{Z}_*}
| u_-( \gamma + \gamma ) |^2 \langle \gamma \rangle^{-N} 
\langle \gamma'  \rangle^{-N} \right)^{\frac12} 
\\ & 
\lesssim \| u_-\|_{\ell^2 }^2 .
\end{split} \end{equation}
(This is a version of Schur's argument, see for instance \cite[Proof of Theorem 4.21, Step 2,]{ev-zw}; later on we will again need the Cotlar--Stein Lemma 
as in the proof of boundedness of $ F^0 $ in the proof of Lemma \ref{l:awayspec}.) Since $ R_+ = R_-^* $ the boundedness of $R_+ $ also follows.
We note that $ R_+ R_- = \operatorname{id}_{\ell^2 ( \ZZ^2_*; \CC^2 )} $.

\begin{prop} 
\label{l:Gru}
Assume that \eqref{eq:specrah} holds and that $ R_\pm $ are defined by \eqref{eq:defRp}. 
Then 
the Grushin problem 
\begin{equation} 
\label{eq:Gru} \left(\begin{matrix} Q^{\rm{w}} ( x, h D)  - z & R_- \\
R_+ & 0 \end{matrix}\right) : L^2 ( \RR , \CC^2 ) \times \ell^2 ( \mathbb Z^2_*; \CC^2 ) \longrightarrow L^2 ( \RR , \CC^2 ) \times \ell^2 ( \mathbb Z^2_*; \CC^2 ), \end{equation}
is well posed for $ z \in (z_1 - \epsilon_0 h , z_1 + \epsilon_0 h ) 
+ i ( - 1, 1 ) $,  with the 
inverse
\begin{equation}
\label{eq:E} \left(\begin{matrix} E ( z, h ) & E_{+} (z, h ) \\
E_- ( z, h ) & E_{-+} ( z , h ) \end{matrix}\right) =
\left(\begin{matrix} \mathcal O ( 1/h)_{L^2 \to L^2} & \mathcal O ( 1 )_{ \ell^2 \to L^2 } \\
\mathcal {\mathcal O} ( 1 )_{ L^2 \to \ell^2 } & \mathcal O ( h )_{ \ell^2 \to \ell^2 } 
\end{matrix}\right) . \end{equation}
In addition, 
\begin{equation}
\label{eq:propE}
\begin{gathered}
  ( E_{-+} ( z, h ) v_+ ) ( \gamma ) = \sum_{ \beta \in \mathbb Z^2_* } 
E_{-+} ( \gamma- \beta  ) v_+ ( \beta) 
,
\\ 
    E_{-+} (  \gamma  ) =  \delta_{\gamma 0 } (z - \kappa ( n_1 h , h ) ) \operatorname{id}_{\CC^2} + 
\mathcal O (h^\infty  \langle \gamma  \rangle^{-\infty} )
\end{gathered}
\end{equation}
where 
$ \kappa  $ is given by 
\eqref{eq:specQ0} and $n_1 $ by \eqref{eq:specrah}.
\end{prop}

Before proceeding with the proof of Proposition \ref{l:Gru} we explain the 
basic idea in a simple example. Suppose $ P $ is a self-adjoint operator on a
Hilbert space $ H $, say a matrix,
with  $ \Spec (P ) \cap [ - \delta, \delta ] = \{ 0 \} $, where $ 0$ is a simple eigenvalue, $ P w = 0 $, $ \|w \| = 1$. Then for $ z \in \left([ - \delta , \delta ] + i \RR \right) \backslash\{0\} $, 
\[  ( P - z )^{-1} = - \frac{ w \langle \bullet , w \rangle}{ z } + S ( z ), \ \  ( P - z ) S ( z ) = \operatorname{id} - w \langle \bullet, w\rangle ,\]
and $ S ( z ) $ is holomorphic. 

We then define
$ R_- : \CC \to H $, $ R_+: H \to \CC $: $ R_- u_- = u_- w $, 
$ R_+ u = \langle u , w \rangle $. One easily checks $ z \in [ - \delta, \delta ] + i \RR $, 
\begin{equation}
\label{eq:Grum} \left(\begin{matrix} P - z & R_- \\
\ \ \ R_+ & 0 \end{matrix}\right)^{-1} = \left(\begin{matrix} 
S ( z ) & R_-\\ 
\ R_+ & z \end{matrix}\right) =: \left(\begin{matrix} 
E ( z ) & E_+ ( z ) \\
E_- ( z ) & E_{-+} ( z ) \end{matrix}\right) : H \times \CC \to H \times \CC . 
\end{equation}
We now follow a similar procedure for $ P = Q^{\rm{w}} - z $ using 
{\em approximate} eigenfunctions $ w_\gamma $ and a partition of
$ \{\chi_\gamma \}_{\gamma \in \ZZ^2_*} $ as in \eqref{eq:apprinv}. The approximate inverse \eqref{eq:Eapp} 
is then similar to \eqref{eq:Grum}.  To obtain the localization result
in \eqref{eq:propE} we upgrade $ L^2 \times \ell^2 $ estimates
to {\em weighted estimates} \eqref{eq:weighted2} and \eqref{eq:propE0}, as
in the remark after the proof of Lemma \ref{l:awayspec}.

We also record translation symmetries of
our Grushin problem:

\begin{lemm}
\label{l:symG}
Suppose that $ \gamma \in \ZZ^2_* $, $ r_\gamma : L^2(\mathbb{R}^2) \to L^2(\mathbb{R}^2) $ is defined by 
\eqref{eq:translph} and $ s_\gamma : \ell^2(\mathbb{Z}_*^2) \to \ell^2(\mathbb{Z}_*^2) $ by 
$ (s_\gamma f) ( \delta ) := f ( \delta - \gamma ) $. Then in the notation of
\eqref{eq:Gru},
\begin{equation}
\label{eq:symG}
\left(\begin{matrix} 
r_\gamma & 0 \\
0 & s_\gamma \end{matrix}\right)  \left(\begin{matrix} Q^{\rm{w}} ( x, h D)  - z & R_- \\
R_+ & 0 \end{matrix}\right) = 
\left(\begin{matrix} Q^{\rm{w}} ( x, h D)  - z & R_- \\
R_+ & 0 \end{matrix}\right) 
\left(\begin{matrix} 
r_\gamma & 0 \\
0 & s_\gamma \end{matrix}\right)  , \ \ \gamma \in \ZZ^2_* .
\end{equation}
\end{lemm}

\medskip

\begin{proof}[Proof of Proposition \ref{l:Gru}]
We follow the same procedure as in the proof of Lemma \ref{l:awayspec} 
and we use the notation from there. 

To start we note that in our range of $ z$'s with $\kappa(n_1h,n_1)$ excluded,
\begin{equation}
\label{eq:Palph1} P_\gamma^{-1} =  
\frac{w_\gamma  \langle \bullet , w_\gamma \rangle }{ 
 \kappa ( n_1 h , h ) -z} + S_\gamma  , \ \ 
P_\gamma S_\gamma  = \operatorname{id} - w_\gamma \langle \bullet , w_\gamma \rangle, \ \ S_\gamma = \mathcal O ( 1/h)_{L^2 \to L^2 } ,
\end{equation}
where the estimate on $ S_\gamma $ follows from the holomorphy 
of $ S_\gamma $ and the maximum principle: we can 
find $ \epsilon_1 > \epsilon_0 $ such that on the boundary of 
$  (z_1 - \epsilon_1 h , z_1 + \epsilon_1 h ) 
+ i ( - 2, 2 ) $, $ \| P^{-1}_\gamma \| = 1/d(z , \Spec (P_\gamma) ) 
=  \mathcal O ( 1/h ) $ and $ | \kappa ( n_1 h , h ) - z |^{-1} = 
\mathcal O ( 1/h ) $.

For future reference will also note that
\begin{equation}
\label{eq:SalP}  S_\gamma u =  \widetilde P_\gamma^{-1} \left(  u - 
w_\gamma \langle u , w_\gamma \rangle \right) + 
\widetilde P_\gamma^{-1} ( \widetilde P_\gamma -  P_\gamma) 
S_\gamma u . \end{equation}

In the notation of \eqref{eq:apprinv} and 
\eqref{eq:Palph1} we define $ E_\bullet^0 = E_\bullet^0 ( z ) $:
\begin{equation}
\label{eq:Eapp} 
\begin{gathered}
  E^0  := \sum_{\gamma \in \ZZ_*^2 } \widetilde \chi^{\rm{w}}_\gamma  
  S_\gamma  \chi^{\rm{w}}_\gamma   , \\ 
  E_+^0  := R_- , \ \  E_-^0  := R_+  , \ \ 
E_{-+}^0 = ( z - \kappa ( h n_1, n_1 ) ) \operatorname{id}_{\ell^2} . 
\end{gathered}
\end{equation}
Lemma \ref{l:symG} shows that $ r_\gamma E_+^0 = E_+^0 T_\gamma $ and
$ E_-^0 r_\gamma = T_\gamma E_-^0 $. We now check that $ r_{\gamma} E^0 
r_{-\gamma } = E^0 $. In fact, \eqref{eq:translph} shows that
\[ \begin{split}  r_{\gamma} E^0 r_{-\gamma } & = 
\sum_{\gamma \in \mathbb{Z}_*^2}r_\gamma  \widetilde \chi^{\rm{w}}_\gamma  
  S_\gamma \chi^{\rm{w}}_\gamma r_{-\gamma} v 
 = \sum_{\gamma \in \mathbb{Z}_*^2}   \widetilde \chi^{\rm{w}}_{\gamma  + \gamma} 
  S_{\gamma+\gamma}  \chi^{\rm{w}}_{\gamma +\gamma} v = E^0 v. 
\end{split}
\]  
As in the proof of \eqref{eq:weighted1} we also see that 
for $ G $ given by \eqref{eq:defG} and 
\[  ( g u ) ( \gamma ) := 
\log \langle \gamma \rangle u ( \gamma ) ,\]
\begin{equation}
\label{eq:weighted2}
\left(\begin{matrix} e^{ -s G^{\rm{w}} } & 0 
\\ 0 &  e^{ -s g } \end{matrix}\right) \left(\begin{matrix} E^0 & \  E^0_+ \\
\  E^0_- & E^0_{-+}   \end{matrix}\right) \left(\begin{matrix} e^{ s G^{\rm{w}} } & 0 
\\ 0 &  e^{ s g } \end{matrix}\right) = 
\left(\begin{matrix} \mathcal O ( 1/h)_{L^2 \to L^2} & \mathcal O ( 1 )_{ \ell^2 \to L^2 } \\
\mathcal {\mathcal O} ( 1 )_{ L^2 \to \ell^2 } & \mathcal O ( h )_{ \ell^2 \to \ell^2 } 
\end{matrix}\right) . \end{equation}

We claim that
\[
\left(\begin{matrix} Q^{\rm{w}}   - z & R_- \\
\ \ R_+ & 0 \end{matrix}\right) \left(\begin{matrix} E^0 & E_+^0 \\
E_-^0 & \ E_{-+}^0 \end{matrix}\right) = \operatorname{id}_{ L^2 \times \ell^2 } +
\left(\begin{matrix} r & \  r_+ \\
\  r_- & 0  \end{matrix}\right) ,\]
where for all $ s \in \RR $, 
\begin{equation}
\label{eq:propE0}
\left(\begin{matrix} e^{ -s G^{\rm{w}} } & 0 
\\ 0 &  e^{- s g } \end{matrix}\right) \left(\begin{matrix} r & \  r_+ \\
\  r_- & 0  \end{matrix}\right) \left(\begin{matrix} e^{ s G^{\rm{w}} } & 0 
\\ 0 &  e^{ s g } \end{matrix}\right) = 
\mathcal O ( h^\infty )_{ L^2 \times \ell^2 \to L^2 \times \ell^2 } . 
\end{equation}

As in \eqref{eq:PF0} (with \eqref{eq:SalP} used to pass from the third line to the fourth line in \eqref{eq:PF0}) and the proof of \eqref{eq:weighted1} we see that
\[ \begin{split}
P E^0 v + R_-E_-^0 v & = \sum_\gamma P \tilde \chi^{\rm{w}}_\gamma S_\gamma
 \chi^{\rm{w}}_\gamma v + 
w_\gamma  \langle v , w_\gamma \rangle 
\\
& 
=  v + \sum_\gamma w_\gamma  ( \langle v , w_\gamma \rangle - \langle  v , \chi^{\rm{w}}_\gamma w_\gamma \rangle  )  + r_1 v 
\\
& = ( \operatorname{id} + r_1 + r_2 ) v , \ \ r_2 :=  \sum_\gamma w_\gamma \otimes 
( 1 - \chi_\gamma^{\rm{w}} ) \bar w_\gamma . 
\end{split}
\]
where $ e^{ -sG^{\rm{w}} } r_1 e^{  s G^{\rm{w}} } = \mathcal O ( h^\infty )_{L^2 \to L^2 } $.
To show that $ e^{ -sG^{\rm{w}} } r_2 e^{  s G^{\rm{w}} } = \mathcal O ( h^\infty )_{L^2 \to L^2 } $ we use \eqref{eq:almostw} and the bound follows again from the Cotlar--Stein Lemma (or from a direct estimate).

The other estimates in \eqref{eq:propE0} are proved
similarly using the localization properties of $ w_\gamma $. We start with 
$  r_+ v_+  = P E^0_+ v_+ + E_{-+} v_+ = 
 \sum_\gamma  ( P - P_\gamma ) w_\gamma v_+ ( \gamma ) $. 
Hence, 
\[  \begin{split} \left( e^{-s G^{\rm{w}} } r_+ e^{sg} \right) v_+ & = 
\sum_\gamma \langle \gamma \rangle^{s} e^{ -s G^{\rm{w}} } e^{ -N  r_\gamma G^{\rm{w}} r_{-\gamma}  } r_\gamma ( ( P - P_0 ) e^{ N G^{\rm{w}} } w_0 ) 
v_+ ( \gamma ) \\
&  = \sum_\gamma c_{\gamma}^{\rm{w}} u_\gamma 
 v_+ ( \gamma) , \ \ \ c_{\gamma} \in S ( 1 ) , \ \ u_\gamma :=r_\gamma ( ( P - P_0 ) e^{ N G^{\rm{w}} } w_0 )  . 
\end{split} 
\]
As in the proof of Lemma \ref{l:ortho} $ \langle u_\gamma, u_\beta \rangle 
= \mathcal O ( h^\infty \langle \gamma - \beta \rangle^{-\infty } ) $ and from this the bound 
$ e^{-s G^{\rm{w}} } r_+ e^{sg} = \mathcal O ( h^\infty ) $
easily follows (see \eqref{eq:below} for a similar argument).  

For $ r_- $ we write
\begin{gather*}  \langle \gamma \rangle^{-s} (  r_- e^{ s G^{\rm{w}} } v )
 ( \gamma ) =  \langle \gamma \rangle^{-s} (R_+ ( E_-^0 e^{ s G^{\rm{w}} }  v))( \gamma ) = 
\langle v , \tilde u_\gamma \rangle , \\ 
\tilde u_\gamma := \sum_\gamma  
 \chi^{\rm{w}}_\gamma  S_\gamma^*  
\tilde \chi_{\gamma}^{\rm{w}} \langle \gamma \rangle^{-s} w_\gamma .
\end{gather*}
We claim that $ \langle \tilde u_\gamma, \tilde u_\beta \rangle = 
\mathcal O ( h^\infty \langle \beta - \gamma \rangle^{-\infty } ) $.
This follows similarly to previous arguments 
using $ \tilde \chi_\gamma^{\rm{w}} w_\gamma = \mathcal O ( h^\infty | \gamma - \gamma|^{-\infty } ) $, $ \gamma \neq \gamma $, and 
$ S_\gamma^* w_\gamma = 0 $. 

This concludes the proof of \eqref{eq:propE0} and in turn that estimate
shows that 
\[ \left(\begin{matrix} Q^{\rm{w}}   - z & R_- \\
\ \ R_+ & 0 \end{matrix}\right)^{-1} = 
 \left(\begin{matrix} E^0 & E_+^0 \\
E_-^0 & E_{-+}^0 \end{matrix}\right) \left( 
 \operatorname{id}_{ L^2 \times \ell^2 } +
\left(\begin{matrix} \tilde r &   \tilde r_+ \\
\  \tilde r_- & \tilde r_{-+}   \end{matrix}\right) \right) ,\]
where
\begin{equation}
\label{eq:remest}  \left(\begin{matrix} e^{- s G^{\rm{w}} } & 0 
\\ 0 &  e^{ -s g } \end{matrix}\right) \left(\begin{matrix} \tilde r &   \tilde r_+ \\
\  \tilde r_- & \tilde r_{-+}   \end{matrix}\right) \left(\begin{matrix} e^{ s G^{\rm{w}} } & 0 
\\ 0 &  e^{ s g } \end{matrix}\right) = 
\mathcal O ( h^\infty )_{ L^2 \times \ell^2 \to L^2 \times \ell^2 } 
\end{equation}
This and Lemma \ref{l:symG} imply \eqref{eq:E} and \eqref{eq:propE}.
\end{proof}

\section{Density of states}

We now use the analysis of \S \ref{anal} to describe the renormalized trace of 
the resolvent of $ Q^{\rm{w}} (x,hD) $. This will lead us to an explicit semiclassical description of the density of states of the Hamiltonian $H^B$ stated in \eqref{tracediff}.

The Schur complement formula 
and \eqref{eq:E} gives for $ | z - z_1 | \leq \epsilon_0 h $, 
\[ ( Q^{\rm{w}} ( x, h D ) - z )^{-1}  = E (z, h ) - E_{+} ( z, h ) 
E_{-+} ( z, h )^{-1}  E_{-} ( z, h ) .\]
Hence, by \eqref{pdotrace},
\begin{equation}
\label{eq:hattrE}
\begin{split} \widehat \tr ( Q^{\rm{w}} ( x, h D ) - z )^{-1} & = 
 G_{z_1} ( z, h ) + J_{z_1 } ( z , h ) , 
\end{split}
\end{equation}
where 
\[  G_{z_1} ( z, h ) := 
\frac{1}{ 4 \pi^2 } \int_{ \TT^2_*}  \tr_{\CC^2} \sigma ( E (z, h ) ) d x d \xi  \]
is holomorphic in $ (z_1 - \epsilon_0 h , z_1 + \epsilon_0 h ) 
+ i ( - 1, 1 )  $ and
\begin{equation}
\label{eq:Jz1} 
J_{z_1 } ( z , h )  := \frac{1}{ 4 \pi^2 } \int_{ \TT^2_*} \tr_{\CC^2}  \sigma ( E_+ ( z, h ) 
E_{-+} ( z, h )^{-1}  E_{-} ( z, h ) ) dx d\xi.
\end{equation}

Dropping $ ( z, h ) $ and writing 
$ A := E_+ E_{-+}^{-1} E_- $, we are seeking $ \sigma ( A ) $ for the operator with the Schwartz kernel, $ K_A $, given by 
\begin{equation}
\label{eq:KA}
K_A ( x, y ) = \sum_{\gamma ,\beta \in \ZZ_*^2 } E_+ ( x, \gamma ) 
E^{-1}_{\pm} ( \gamma - \beta ) E_- ( \beta , y ) .
\end{equation} 
From \eqref{eq:propE} we see that 
\begin{equation}
\label{eq:Emploc}
\begin{split}  E_{-+} ( \gamma  ) & = \delta_{\gamma 0} E_{-+}^0 ( \gamma ) + 
\mathcal O ( h^\infty \langle \gamma \rangle^{-\infty } ) \\ & 
= 
  \delta_{ \gamma 0 } ( z - \kappa ( n_1 h , h ) ) \operatorname{id}_{\CC^2 } + \mathcal O ( |\Im z|^{-1} 
h^\infty \langle \gamma  \rangle^{-\infty } )  . 
\end{split}
\end{equation}
We recall that $ z \in (z_1 - \epsilon_0 h , z_1 + \epsilon_0 h ) 
+ i ( - 1, 1 )   $, $ n = n_1 ( z_1 , h ) $, 
and that \eqref{eq:specrah} holds. It follows that for 
\begin{equation}
\label{eq:lowerImz}
| \Im z | > h^{M } ,
\end{equation}
where $ M $ is arbitrary and fixed, we have 
\begin{equation}
\label{eq:Einvexp}
\ E_{-+}^{-1} ( \gamma ) = 
 ( z - \kappa ( n_1 h, h ) 
 )^{-1}  \delta_{\gamma, 0 } \operatorname{id}_{ \CC^2 }  
+ \mathcal O ( h^\infty \langle \gamma  \rangle ^{-\infty } )  . 
\end{equation}
We now want to use this expression of $ E_{-+}^{-1} $ to analyse the symbol of $ A $.

\medskip

\noindent
{\bf The leading term.} 
To obtain the leading term in \eqref{eq:Jz1} we 
define
\begin{equation}
\label{eq:J0z1} 
 J^0_{z_1} ( z, h )   := \frac{1}{ 4 \pi^2 } \int_{ \TT^2_*} ( z - \kappa ( n_1 h , h ) )^{-1} 
 \tr_{ \CC^2} \sigma ( E^0_+ ( z, h ) 
  E_{-}^0 ( z, h ) ) dx d\xi,
\end{equation}
where the approximations of $ E_\pm $, $ E_\pm^0 $, are defined in \eqref{eq:Eapp}:
\begin{gather*} E_+^0 ( z, h ) v_+ ( x ) = \sum_{ \gamma} w_\gamma ( x )  v_+ ( \gamma) , \ \ ( E_-^0 (z, h ) v ) ( \gamma ) = \langle v , w_\gamma \rangle = 
\left(\begin{matrix} \langle v , w_\gamma^+ \rangle \\
\langle v , w_\gamma^- \rangle  \end{matrix}\right) , \\ 
w_\gamma = ( w_\gamma^+, w_\gamma^-) = ( r_\gamma  u_{n_1}^+ ( h ) , r_\gamma u_{n_1}^- ) , \ \ 
v_+ \in \ell^2 ( \ZZ^2_*, \CC^2 ) , \ \ v \in L^2 ( \RR, \CC^2 ) .
\end{gather*}
 The inverse $ E_{-+}^{-1} $ was replaced by the first
term on the right hand side of \eqref{eq:Einvexp}. 

To analyse $ J^0_{z_1} $ we use the formula 
for the Weyl symbol in terms of the Schwartz kernel:
\begin{equation}
\label{eq:K2a}
\begin{gathered}
A  u ( x ) = \int_{\RR^2 } K ( x, y ) u ( y ) dy , \ \ 
K ( x, y ) = \frac{1}{ 2 \pi h  } \int_{ \RR }  a ( \tfrac{x+y}2, \xi)
e^{ \frac i h ( x - y )} d \xi , \\
a ( x, \xi ) = \int K ( x - \tfrac w 2 , x + \tfrac w 2 ) e^{ \frac i h  w \xi } dw , 
\end{gathered}
\end{equation}
see \cite[\S 4.1]{ev-zw}. In our case $ A = E_+^0 ( z, h ) E_-^0 ( z, h) $,
where (see \eqref{eq:Eapp}) $ E_+^0 = R_- $, $ E_-^0 = R_+ = R_-^* $, where
$ R_\pm $ are given in \eqref{eq:defRp}. That is, 
\[ E_+^0 f ( x ) = \sum_{ \alpha } E_+^0 ( x, \alpha ) f ( \alpha ) , \ \
E_+^0 ( x , \alpha ) = w_\alpha ( x ) = (w_\alpha^+ ( x ) , w_\alpha^- ( x )  )  \in \CC^2 \otimes \CC^2 ,  \] 
where $  f ( \alpha ) = ( f^+ ( \alpha)  , f^- ( \alpha ) )^t \in \CC^2 $, 
and 
\[ E_-^0 u ( \gamma ) = \int_\RR E_- ( \gamma , x ) u ( x ) dx , \ \ 
E_-^0 ( \gamma, x ) = w_\gamma( x )^* = \left(\begin{matrix}  \bar w_\gamma^+ ( x)  \\
\bar w_\gamma^- ( x)  \end{matrix}\right) , \ \ 
u \in L^2 ( \RR; \CC^2 ) . \] 
This means that 
\[  K ( x, y ) = \sum_{ \alpha } E_+^0 ( x, \alpha ) E_-^0 ( \alpha, h ) = 
\sum_{\alpha } w_\alpha ( x ) w_\alpha ( y )^* ,
\]
which in turn gives,
\[ \begin{split}  \sigma ( E^0_+ ( z, h )   E^0_{-} ( z, h ) ) (x , \xi ) & = 
\sum_\alpha \int_\RR w_\alpha ( x - \tfrac w 2) w_\alpha^* ( x - \tfrac w 2) 
e^{ \frac i h w \xi } w \\
& = \sum_\alpha \int_\RR e^{  \frac i h w ( \xi - \alpha_2) }  w_0 ( x - \tfrac w 2 - \alpha_1 ) { w_0 }( x + \tfrac w2 - \alpha_1 )^* dw .
\end{split}
 \]
Hence, 
\begin{equation}
\label{eq:E0pm} \begin{split} 
&  \int_{ \TT^2_*}  \sigma ( E^0_+ ( z, h ) 
   E^0_{-} ( z, h ) ) \frac{dx d\xi }{ 4 \pi^2 } = \\
& \ \ \ \ \   \sum_\alpha 
   \int_{\TT_*^2}  \int_\RR 
e^{  \frac i h w ( \xi - \alpha_2) }  w_0 ( x - \tfrac w 2 - \alpha_1 ) { w_0 }( x + \tfrac w2 - \alpha_1 )^*
d w \frac{dx d\xi }{ 4 \pi^2 } = 
\\
& \ \ \ \ \     \int_{\RR^2} \int_\RR
e^{  \frac i h w  \xi }  w_0 ( x - \tfrac w 2  ) { w_0 }( x + \tfrac w 2 )^* 
d w \frac{dx d\xi }{ 4 \pi^2 } =  \\
& \ \ \ \ \ 
 \frac{1}{ 2 \pi }  \int_\RR w_0(x) w_0(x)^* dx 
= 
\frac{h}{2\pi} I_{\CC^2} 
\end{split} \end{equation}
Inserting this in \eqref{eq:J0z1} gives 
\[ J_{z_1}^0(z, h ) = \frac{h} {  \pi } ( z - \kappa ( n_1 h , h ) )^{-1} .
\]

To analyze the remaining contribution to \eqref{eq:Jz1} we use 
\eqref{eq:propE} to write 
\begin{equation}
\label{eq:Epm0} 
\begin{gathered} ( E_{-+ } ( z, h )^{-1}  - ( z - \kappa ( n_1 h , h ) )^{-1} I_{\CC^2} )
v_+ ( \gamma ) = \sum_{\alpha } e ( \gamma - \alpha ) v_+ ( \alpha ) , 
\\ e ( \gamma ) = e( z, h , \gamma ) = \mathcal O ( h^\infty \langle \gamma \rangle ^{-\infty }) , \ \ | \Im z | > h^M . \end{gathered}
\end{equation}
Hence $ J_{z_1} ( z , h ) = J^0_{z_1} ( z , h ) + J^1_{z_1} ( z, h ) $ 
where, 
\[  J_{z_1}^1 ( z, h )   = 
 \int_{ \TT^2_*}  \tr_{\CC^2 } \sigma ( E_+ ( z, h ) 
( E_{-+ } ( z, h )^{-1}  - ( z - \kappa ( n_1 h , h ) )^{-1} I_{\CC^2}  )  E_{-} ( z, h ) ) \frac{dx d\xi }{ 4 \pi^2 } .\]
Lemma \ref{l:symG} and \eqref{eq:remest} give 
\begin{equation*}
\begin{gathered}
 E_+ ( z, h ) v_+ ( x ) = \sum_{ \gamma}    r_\gamma  W_+  ( x ) v_+ ( \gamma ) , \ \ W_+ = w_0  + e_0  , \ \ 
e_0 = \mathcal O ( h^\infty)_{\mathscr S } , \\
( E_- ( z , h ) v )( \gamma ) =  \langle v , r_\gamma W_- \rangle, \ \ W_- = w_0 + f_0 , \ \ f_0 \in \mathcal O ( h^\infty )_{\mathscr S } ,
\end{gathered}  
\end{equation*}
so that, using \eqref{eq:K2a} again, 
\[ \begin{split} J_{z_1}^1 ( z, h ) & = \sum_{\gamma, \beta  } 
  \int_{\TT^2_*}  \int_\RR \tr_{ \CC^2} E_+ ( x - \frac w2 , \gamma ) 
e ( \gamma - \beta ) 
E_- ( \beta, x + \frac w2 ) e^{ \frac i h w \xi } dw \frac{dx d\xi }{ 4 \pi^2h } \\
& = \sum_{ \gamma , \beta} 
  \int_{\TT_*^2}  \int_\RR 
 r_{\gamma + \beta}  W_+ ( x - \frac w 2  ) e( \gamma )  { r_\beta  W_-}( x + \frac w 2 )^* e^{ \frac i h w \xi} 
 d w \frac{dx d\xi }{ 4 \pi^2 } .
 \end{split} 
 \]
As in \eqref{eq:E0pm} we now use the sum over $ \beta_2 $ to change integration in $ \xi $ from  $ \TT_*^1 $ to $ \RR $ and then integrate in $
w $ and $ \xi $. This and \eqref{eq:Epm0} give
\[ \begin{split} 
 J_{z_1}^1 ( z, h ) & =
  \frac{h} { 2 \pi} \sum_{\gamma} \sum_{\beta_1 } 
    \int_{\TT^1_*}  e^{ i x \gamma_2} 
W_+ ( x  - \beta_1 - \gamma_1 ) e ( \gamma ) 
{ W_- }( x- \beta_1)  dx 
\\
& = \mathcal O ( h^\infty ) \sum_\gamma \int_\RR
| W_+ ( x - \gamma_1 ) | \langle \gamma \rangle^{-\infty} 
|W_- ( x ) | dx 
=  \mathcal O ( h^\infty \| W_- \|  \| W_+\|) .
\end{split} \end{equation*}

The following proposition summarizes what we have done in this section so far:

\begin{prop}
\label{th:trac1}
Suppose that $ Q $ is given by \eqref{eq:effha} and that $ \widehat \tr $ is defined in \eqref{pdotrace}. Let $ z_1 $ be chosen as in 
\eqref{eq:specrah}. Then 
\begin{equation}
\label{eq:thtrac1}
\begin{gathered}   
\widehat \tr ( Q^{\rm{w}} ( x, hD ) - z )^{-1} := \frac h\pi \sum_{ n \in \mathbb Z }  
({ z - \kappa  ( hn ; h )  })^{-1} + F_{z_1} (z, h) + 
\mathcal O ( h^\infty ) ,  \\ | \Im z|  > h^M ,  \ \ |z - z_0 | \leq \epsilon_0 h ,  
\end{gathered}
\end{equation} 
where $ F_{z_1} ( z , h ) $ is holomorphic in $ | z - z_0 | \leq \epsilon_0 h $, $ M $ is arbitrary and 
$ \kappa ( nh, n ) $ defined by \eqref{eq:QM}. 
\end{prop}

\begin{rem}
Using one variable complex analysis, a crude estimate
$ G_{z_1} ( z, h ) = \mathcal O ( h^{-M_0} ) $ 
and maximum principle similar to \cite[Lemma D.1]{res}
one can show that \eqref{eq:thtrac1} holds
in a fixed neighbourhood of the Dirac point $ \Delta|_{B_k}^{-1} ( 0 ) $
with $ F $ independent of $ z_1 $ and holomorphic. As in \cite{HS2} we opt for a simpler version of piecing together local expressions 
\eqref{eq:thtrac1} using a partition of unity.
\end{rem}

We are now in the position to prove the main theorem describing the semiclassical density of states formula for our model of graphene:

\begin{theo}
\label{theorem3}
Let $z^D := \Delta\vert_{B_k}^{-1}(0)$ be the energy of the {\em Dirac points} located on the $k$-th band. If $ I $ is a sufficiently small 
neighbourhood of $ z^D $, then for $ f \in C^{\alpha}_{\rm{c}} ( I) $\begin{equation}
\label{eq:tracef}
\widetilde \tr f ( H_B ) = \frac{h}{\pi \left\lvert b_1 \wedge b_2 \right\rvert} \sum_{ n \in \ZZ } 
f ( z_n ( h  ) ) + \mathcal O (\| f \|_{ C^\alpha } h^\infty ) , \ \ 
\Delta ( z_n ( h ) ) = \kappa ( n h , h ) , 
\end{equation}
where $ \kappa ( n h , h ) $ is given by \eqref{eq:QM}. 
\end{theo}
\begin{proof}
We cover $ I $ by intervals of type $ I^1_{z_1} := (\Delta|_{B_k} ) ^{-1} ( ( z_1 -  \epsilon_0 h , 
z_1 +  \epsilon_0 h ))  $ where $ z_1 $ is as in \eqref{eq:specrah}, and
intervals $ I^2_{z_2} :=  (\Delta|_{B_k} ) ^{-1} ( ( z_2 - \epsilon_1 h , z_2 + \epsilon_1 h )) $ where $  ( z_2 - 2\epsilon_1 h , z_2 + 2 \epsilon_1 h ) 
\cap \Spec ( Q_0 ( x, h D ) ) = \emptyset $ ($ Q_0 $ is defined in 
\eqref{eq:defQ0}). Lemma \ref{l:awayspec} shows that near intervals $ \Delta ( I^2_{z_2} ) $, $ \widehat \tr ( Q^{\rm{w}} ( x, h D ) -z)^{-1} $ is 
holomorphic. Since we are also away from $ \kappa ( h n ; h )'s $, that means
that \eqref{eq:thtrac1} holds also near $ I^2_{z_2} $.

 Following \cite[\S 10]{HS2} we proceed in two steps. First we recall that for $ f \in C^\infty_{\rm{c}} ( \RR ) $ 
satisfying 
\begin{equation}
\label{eq:condf}  f^{(k)} = \mathcal O ( h^{-N_0} ) \ \text{ for a fixed $ N_0 $ and $ 0 \leq k \leq 4 $.} 
\end{equation}
we can find an extension of $ f $, $\widetilde f  \in C^\infty_{\rm{c}} ( \CC ) $ satisfying
\begin{equation}
\label{eq:aae}
\widetilde f, \widetilde f' = \mathcal O ( h^{-N_0} ) , \ \ 
\partial_{\bar z } \widetilde f = \mathcal O ( h^{-N_0} |\Im z | ) ,
\end{equation}
In fact, Mather's construction of $ \widetilde f $ -- see \eqref{eq:mather} -- shows that 
$$ \partial_{\bar z } \widetilde f = 
| \Im z | \mathcal O ( \| \xi^2 \hat f ( \xi ) \|_{L^1 ( d \xi ) } ) = 
| \Im z | \mathcal O (| \supp f |  \sup_{ k \leq 4 }  | f^{(k)} | ), $$
and \eqref{eq:condf} implies \eqref{eq:aae}.

Using a partition of unity with functions supported in intervals
of type $ I_{z_j}^j$, $ j =1,2$, covering $ I $, we only need to
 consider $ f $ 
supported in $ I_{z_j}^j $ and satisfying \eqref{eq:condf}. 

If $ I $ is a sufficiently small neighbourhood of the Dirac point $ z^D 
=  \Delta\vert_{B_k}^{-1}(0)$,
we obtain no Dirichlet contribution in \eqref{tracediff}. (The Dirichlet spectrum is located at the band edges $\Delta(z)= \pm 1.$) We observe further that $\Delta$ has a non-vanishing derivative inside the $k$-th band and
$
 1/{\left\lvert \Im \Delta(z)  \right\rvert} \sim {1}/{\left\lvert \Im z  \right\rvert}.
$
Inserting \eqref{eq:thtrac1}  into \eqref{tracediff} and using a generalized version of the argument principle, as in the proof of Lemma \ref{Dirichletcalc}, we obtain
\begin{equation}
\label{eq:tracefM} 
\begin{split} \widetilde{\operatorname{tr}}(f(H^B)) 
 & = \frac{ h} { | b_1 \wedge b_2 | \pi^2 } 
 \int_{\CC } \partial_{\bar z } \widetilde f ( z ) \Delta' ( z ) 
\sum_{ n \in \ZZ } ( z - \kappa ( h n , n ) )^{-1 } dm ( z ) \\
& \ \ \ \ \ \  \ \ \ + 
\frac 1 \pi \int_{ |\Im z | < h^M } \partial_{\bar z } \widetilde f ( z )
\mathcal O ( 1/|\Im z | ) dm ( z ) 
\\
& =  
 \frac{h}{\pi \left\lvert b_1 \wedge b_2 \right\rvert}  \sum_{ n \in \ZZ } 
f ( z_n ( h  ) ) + \mathcal O ( h^{M-N_0 } ) , \ \   z_n ( h  ) =\Delta\vert_{B_k}^{-1}(\kappa(nh,h)).
\end{split}
 \end{equation}

We now approximate $ f \in C^{\alpha}_{\rm{c}}  (I) $ by 
\[  f_h ( x ) = h^{ -M_0} \int_\RR f ( y ) \psi ( h^{-M_0} ( x - y ) ) dy , \ \ \psi \in C^\infty_{\rm{c}} ( \RR; [ 0 , 1 ] ) , \ \ \int \psi ( y ) dy = 1. \]
The condition \eqref{eq:condf} is then satisfied with $ N_0 = 4 M_0 $. 
Since $ f \in C^{\alpha} $ we also have
\begin{equation}
\label{eq:appf}  \sup_x | f ( x ) - f_ h ( x ) | \leq \| f \|_{ C^\alpha }  h^{ \alpha M_0 } .
\end{equation}
By using \eqref{eq:tracefM} with $ f $ replaced by $ f_h $ and then using
\eqref{eq:appf} 
\[ \widetilde \tr f ( H_B ) = \frac{h}{\pi \left\lvert b_1 \wedge b_2 \right\rvert} \sum_{ n \in \ZZ } 
f ( z_n ( h  ) ) +  \mathcal O (\| f \|_{ C^\alpha } h^{\alpha M_0 -1 }  ) + 
\mathcal O (\| f \|_{ C^\alpha } h^{ M - 4 M_0 } ) . 
\]
By choosing $ M = 5 M_0 $ and then $ M_0 $ arbitrarily large we obtain 
\eqref{eq:tracef}.
\end{proof} 

Things become much simpler when $ f $ is smooth. For completeness we include
\begin{theo}
\label{t:smooth}
Suppose that $ f \in C^\infty_{\rm{c}} ( I ) $ where $ I $ is a small neighbourhood of a Dirac energy $ z^D $. Then for any  $N $
\begin{equation} 
\label{eq:Ajf}   \widetilde \tr \, f ( H_B ) = \sum_{ j=1}^{N} 
A_j ( f ) h^j + \mathcal O ( h^{N +1 } ) , \ \ 
A_0 ( f ) = \rho_0 ( f ) , \ \ A_1 ( f ) = 0 .  \end{equation}
\end{theo}
\begin{proof}
We use the method of \cite[Chapter 7]{D-S} and consider an almost analytic extension of $ f $ defined by \eqref{eq:mather}. Then, avoiding 
again the Dirichlet eigenvalues by taking $ I $ small enough,
\[  \widetilde \tr \, f ( H_B ) = 
\frac{2}{ ( 2 \pi)^2 3 \sqrt 3 \pi} \int_{ \RR^2 / 2 \pi \ZZ^2 } \left(  \int_\CC 
 \partial_{\bar \lambda } \widetilde f ( \lambda ) \Delta' ( \lambda) 
 \tr_{ \CC^2} \sigma \left( ( Q^{\rm{w}} - \Delta ( \lambda ))^{-1} 
\right) d m ( \lambda ) \right)    dx d \xi ,\]
which follows from Definition \ref{pdotrace} and \eqref{tracediff}. 
From \cite[Proposition 8.6]{D-S} we have, for 
$ z \in D( 0 , C) \setminus \RR $ (and any fixed $ C $),
\begin{gather*}
 ( Q^{\rm{w}} - z )^{-1} = R^{\rm{w}} (z ;x, h D_x , h ) , \\ 
|\partial_x^\alpha \partial_\xi^{\beta} R ( z , x , \xi, h ) |
\leq C_{\alpha \beta} \max ( 1 , h/|\Im z |)^{3} |\Im z|^{ -1 - |\alpha |
 - |\beta| } .
 \end{gather*}
Hence in the formula for $ \widetilde \tr \, f ( H_B ) $ we can replace
$ \sigma ( Q^{\rm{w}} - \Delta ( \lambda ) )^{-1} $ by 
$ R ( \Delta ( \lambda ) , x, \xi , h ) $.  As in \cite[(8.14)]{D-S} 
we see that for $  |\Im \Delta ( \lambda ) | \simeq |\Im \lambda | 
\geq h^\delta $, $ 0 < \delta < \frac12 $, we have an expansion
\[ \begin{split} \tr_{\CC^2}  R ( \Delta ( \lambda ) , x, \xi , h ) & \sim 
 \tr_{\CC^2} ( Q ( x, \xi ) - \Delta ( \lambda ) )^{-1}  + 
h^2 \tr_{ \CC^2}  q_2( \Delta ( \lambda ) , x, \xi )  (Q ( x, \xi ) - \Delta ( \lambda ) )^{-5} \\
& \ \ \ \ \ \ \ + h^3 \tr_{ \CC^2} q_3 ( \Delta ( \lambda ) , x, \xi ) (Q ( x, \xi ) - \Delta ( \lambda ) )^{-7}  + \cdots ,
\end{split}
\]
where $ q_j ( z , x, \xi ) \in \CC^2 \otimes \CC^2 $ are polynomials in $ z $ of degree $ \leq 2j$ and the coefficients are $ (2 \pi \ZZ)^2 $ periodic.

Adapting the calculation in \cite[(8.16)]{D-S}  gives the expansion \eqref{eq:Ajf} with 
\[  A_j ( f ) = \sum_{ \pm } \frac{2}{ ( 2 \pi)^2 3 \sqrt 3 } \frac{1}{(2j)!}
\int_{ \RR^2 / 2 \pi \ZZ^2 } \tr_{\CC^2} 
\partial_z^{2j} ( q_j ( z , x , \xi ) f ( z ) )|_{ z = z_{ \pm } ( x, \xi ) } \, dx d \xi \]
where 
\[ z_{ \pm } ( x, \xi ) :=  \Delta^{-1} ( \pm \tfrac13 | 1 + e^{ix} + e^{i\xi} | )  .\]
In particular,
\[ A_0 ( f ) = \sum_{ \pm } \frac{2}{ ( 2 \pi)^2 3 \sqrt 3 } \frac{1}{(2j)!}
\int_{ \RR^2 / 2 \pi \ZZ^2 }  f ( z_\pm ( x, \xi ) ) d x d \xi, \]
which is $ \rho_0 ( f ) $ for $ f $ supported near $ z_D $.
\end{proof}

\section{Magnetic oscillations}

In this section we show how Theorem \ref{theorem3} can be used to describe low temperature magnetic oscillations in the (smoothed-out) density of states and in magnetization. In the physics literature they are known as 
the \emph{Shubnikov-de Haas} (SdH) and the \emph{de Haas-van Alphen} (dHvA) effects, respectively. 
 
We stress the asymmetry with respect to the Dirac energy levels which comes from semiclassical quantization conditions and the dispersion relations.
 It is not seen when a ``perfect cone" (that is, a harmonic oscillator) approximation is used -- see
\eqref{eq:physde}. We note that an asymmetry is already present in the 
case when there is no magnetic field. An experimental result in the setting
molecular graphene \cite[Figure 4d]{hari} is shown in Figure \ref{Fig:DOS}.
The corrections to the perfect cone approximation are due to the modified linear dispersion relation as energies move away from the Dirac points. The perfectly linear dispersion relation of the $\operatorname{QED}_{2+1}$-model has been a ubiquitous assumption in the physics literature -- see Gusynin--Sharapov \cite{sgb}, \cite{GS05}, \cite{GS} and references therein. The approach presented here leads to modified Landau levels showing the well known $\sqrt{n B}$-scaling only to leading order. 

\begin{figure}
  \centering
  \begin{subfigure}{0.45\textwidth}
 \includegraphics[height=7.5cm]{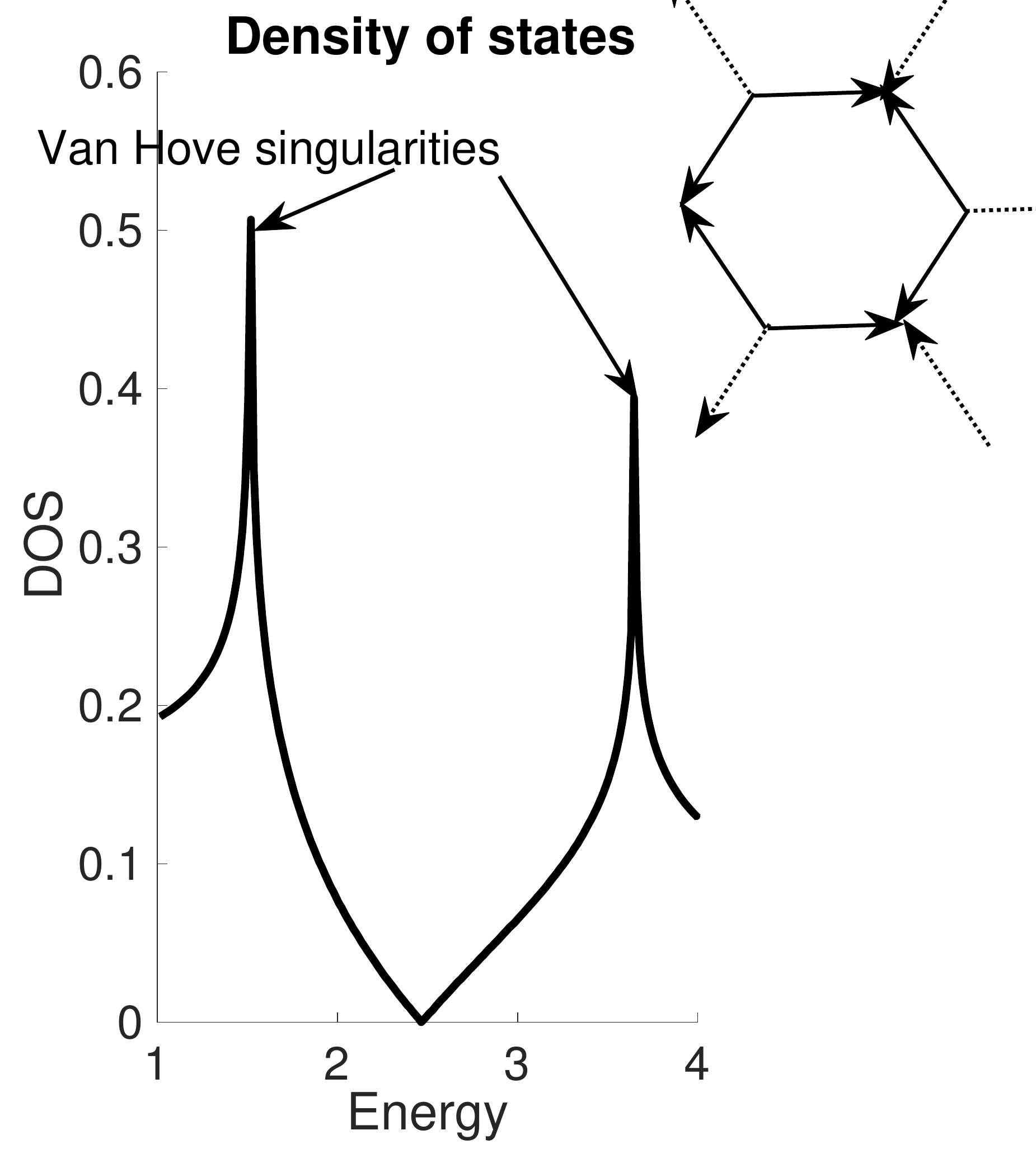}
    \caption{The DOS of the operator $H^{B=0}$ \eqref{magop} per hexagonal cell volume with zero magnetic field potential $(V_e)=0$ on the first Hill band $[0,\pi^2]$ as described in \eqref{Hillbands}.}
  \end{subfigure}
\qquad \
  \begin{subfigure}{0.45\textwidth}
    \includegraphics[height=7.5cm]{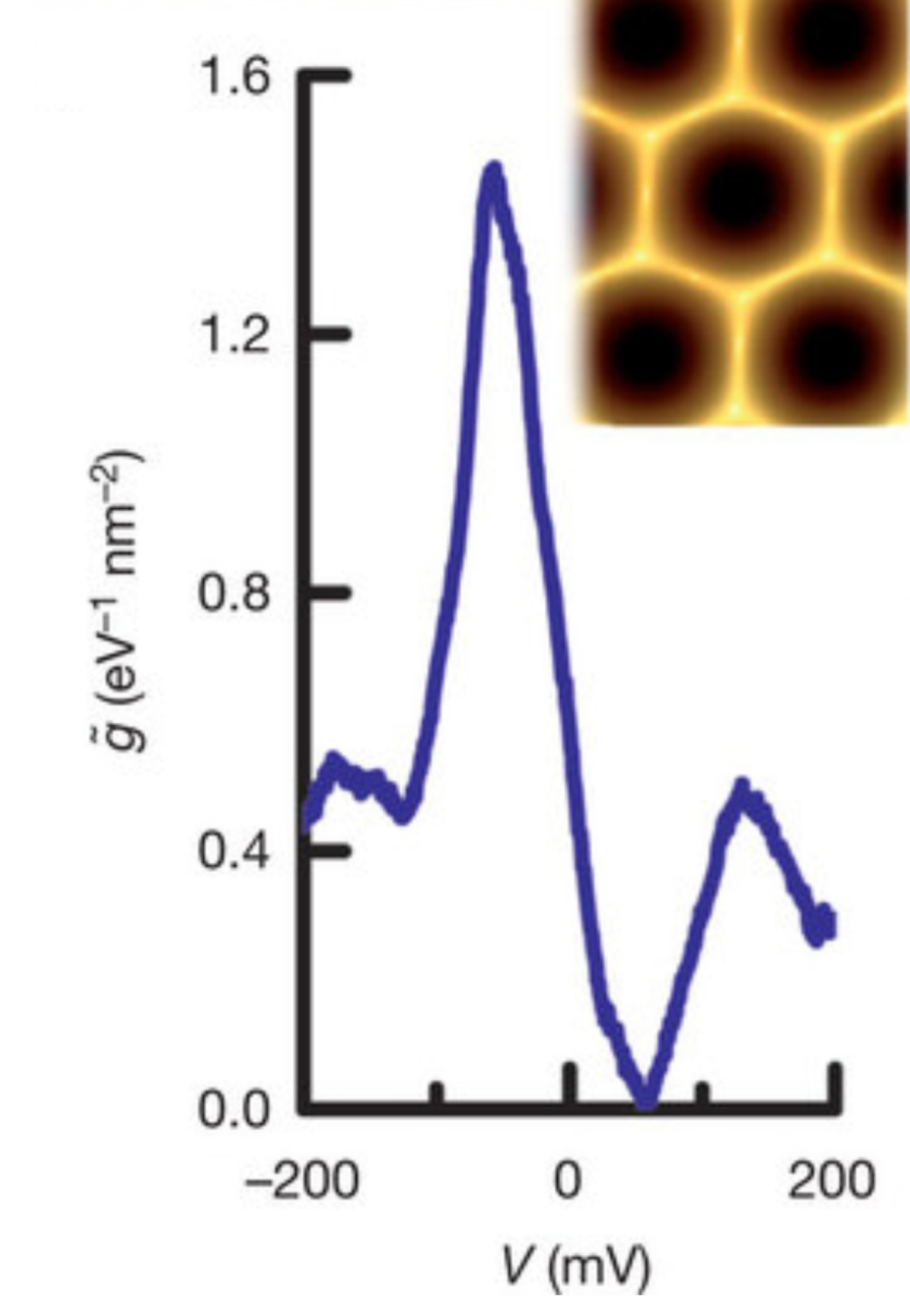}
    \caption{An experimental plot of the density of states for a molecular model of graphene obtained using scanning tunneling microscopy \cite{hari}.}
  \end{subfigure}
  \caption{\label{Fig:DOS}Comparing numerical and experimental no-magnetic field DOS in the 
  quantum graph model and molecular graphene, respectively. }
\end{figure}

\subsection{Shubnikov-de Haas oscillations in DOS}
The Shubnikov-de Haas (SdH) effect is the occurrence of oscillations in the density of states,
with periods proportional to the inverse strength of the magnetic field. These oscillations can be experimentally measured in terms of longitudinal conductivity or resistivity \cite{W11} and \cite{Tan11}. For a theoretical discussion of the relation between oscillations in electric and also thermal conductivities on the one hand and the density of states on the other hand, see also \cite{GS05}. 

We start with an approximation for the semiclassical Landau levels $z_h$ of $H^B$ introduced in Theorem \ref{th:trac1}. For that we consider an 
approxiate Bohr-Sommerfeld condition:
\begin{equation}
\label{BSC}
g(z^{(1)}_n ( h )) =|n|h, \ \ g ( x ) := F_0 \left( \Delta(x)^2\right)\vert_{I_{\delta ,k}} , 
\end{equation}
where $F_0$ is the normalized phase space area of one potential well in the Brouillon zone defined in Proposition \ref{p:specQ0} and $I_{\delta,k}$ as in Theorem \ref{theorem3}. 
Since $ F_0' ( 0 ) \neq 0 $, $ \Delta ( z_D ) = 0 $, 
$ \Delta' ( z_D ) \neq 0 $ (see \eqref{Hillbands}), we have $ g ( z_D ) = g' ( z_D ) 
= 0 $, $ g''( z_D ) > 0 $. This means that we have two branches of the
inverse of $ g $ defined for small $ x \geq 0 $:
 $ \pm ( g^{-1}_\pm ( x ) - z_D ) \geq 0 $. Then
\begin{equation}
\label{eq:appz}
z_{ \pm |n|}^{(1)} ( h ) = g_{\pm}^{-1} ( |n| h ) , \ \  z_0^{(1)} ( h ) = 0 .
\end{equation}

\begin{rem}
\label{rem:asymmetry}
Because of the asymmetry of the cones which are the solutions to $ | Q ( x, \xi ) - \Delta ( z ) |=0$ in a neighbourhood of the Dirac point $\Delta\vert_{B_k}^{-1}(0)$, we observe that although $\kappa(nh,h)=-\kappa(-nh,h)
 $ we have $z^{(1)}_{n}(h)\neq-z^{(1)}_{-n}(h)
+ \mathcal O ( h^\infty ) $ in general. That can already be seen in the simplest case \eqref{eq:Floqf}.
\end{rem}

\begin{figure}
\includegraphics[height=9cm]{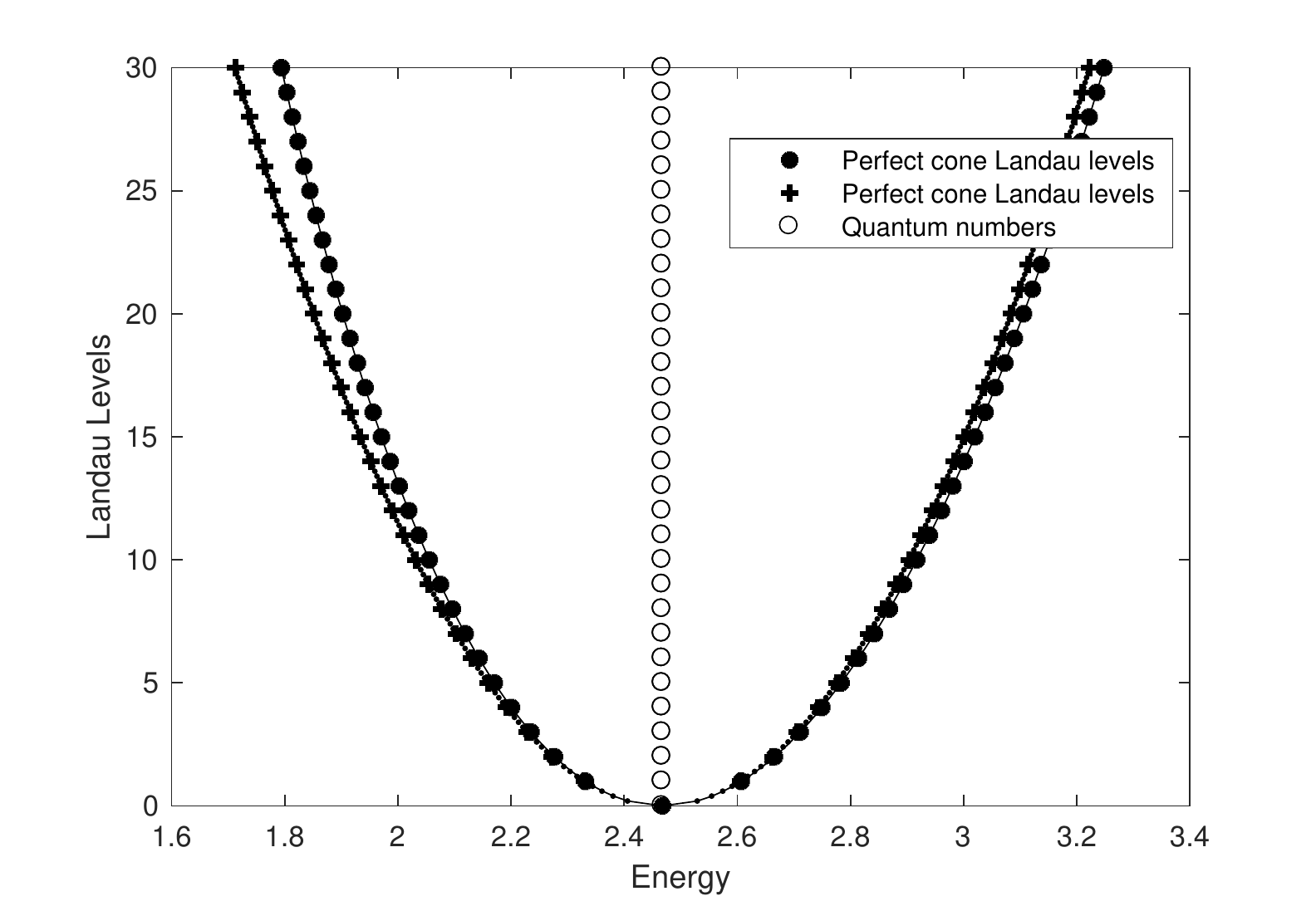} 
\caption{Landau levels \label{Fig:Landau} located on the Dirac cone of the first Hill band with zero potential derived from the Bohr-Sommerfeld condition and normalized phase space area $g$ \eqref{BSC} and its perfect cone approximation \eqref{eq:gc} for magnetic flux $h=0.01.$}
\end{figure}

We recall from  \eqref{eq:g2F} that 
 \[  F ( \Delta ( z_n ( h ) )^2 ) = F_0 ( \Delta ( z_n ( h ) )^2 ) + 
\mathcal O ( h^2 \Delta ( z_n ( h ))^2 ) = | n | h + \mathcal O ( h^\infty ) , \] 
which gives $ \Delta ( z_n ( h ) ) ^2  = \Delta ( z_n^{(1)} ( h))^2 
+ \mathcal O ( |n| h^3) + \mathcal O (h^\infty ) $. 
Hence, 
\begin{equation}
\label{eq:zone}
z_n ( h ) = z_n^{(1) } ( h ) + \mathcal O \left(  {h^{ \frac52} }{ |n|^{\frac12}} \right) , \ \ \ \  n \neq 0 
\end{equation}

For $ f \in C^\alpha ( I ) $, $ 0 < \alpha \leq 1 $, we then have 
\begin{equation}
\label{eq:Diracmeasures}
\rho_{B} ( f ) = \widetilde \rho_B ( f )   + \mathcal O ( \| f \|_{ C^\alpha}  h^{2\alpha } ) , \ \ \
\widetilde \rho_B ( f ) :=  
\frac h{\pi \left\lvert b_1 \wedge b_2 \right\rvert}  \sum_{n \in \mathbb{\ZZ}} f ( z^{(1)}_n ( h) ) .
\end{equation}
The error 
term came from the approximation \eqref{eq:zone} 
 and the fact that the number of terms contributing on the support of 
 $ f $ is bounded by 
$ \mathcal O ( 1/h ) $:
\[  h \sum_{ n \neq 0 } | f ( z_n ( h ) ) - f ( z_n^{(1)} ( h ) ) | 
\leq \| f \|_{ C^\alpha }  h^{1 +  \frac52 \alpha} \sum_{ 0 < n  \leq C/h } n^{  \frac12{\alpha} } = \mathcal O (\| f \|_{ C^\alpha }  h^{ 2\alpha  } )  .\]

The leading term in \eqref{eq:Diracmeasures} provides a refinement 
of \eqref{eq:physde} which is easy to investigate numerically. 
To compare it with \eqref{eq:physde} we calculate $ v_F $ (the value used
here differs by the area factor) by
using the leading term in the Taylor expansion of $ g $ (and \eqref{eq:norfo} to calculate $ F_0' ( 0 ) $):
\[   g ( x ) =  \Delta' ( z_D )^2 F_0' ( 0 ) x^2 + \mathcal O ( x^3 ) , \ \
F_0 ' ( 0 ) = 3^{\frac32} \ \Longrightarrow \ v_F = 3^{-\frac34} \Delta' ( z_D )^{-1}   . \]
In other words, a ``perfect cone" quantization condition reads,
\begin{equation}
\label{eq:gc}
\begin{gathered}
g_{\rm{c}} ( z_n^{\rm{c}} ( h ) ) = |n| h , \ \ g_{\rm{c}} ( x) = 
v_F^{-2} ( x - z_D)^2 ,   \ \  v_F = 3^{-\frac34} \Delta' ( z_D )^{-1}  ,
\\ z_n^{\rm{c}} = z_D + v_F \sgn(n) \sqrt{ |n| h } ,
\end{gathered}
\end{equation}
and the comparison with \eqref{BSC} is shown in Figure \ref{Fig:Landau}.

To plot the density of states we use $ \widetilde \rho_B ( f ) $ 
in \eqref{eq:Diracmeasures} with $ f_\mu ( x ) = e^{ - ( x - \mu )^2 / 2 \sigma^2 } / \sqrt{ 2 \pi} \sigma $ and plot $ \mu \mapsto \widetilde 
\rho_B ( f_\mu ) $. Since $ \| f_\mu \|_{C^1} = \mathcal O ( \sigma^{-2})$
we obtain valid approximation for $ \sigma \gg h $ -- see Figure \ref{fig:DOS4} 

\begin{figure}
\includegraphics[height=11cm]{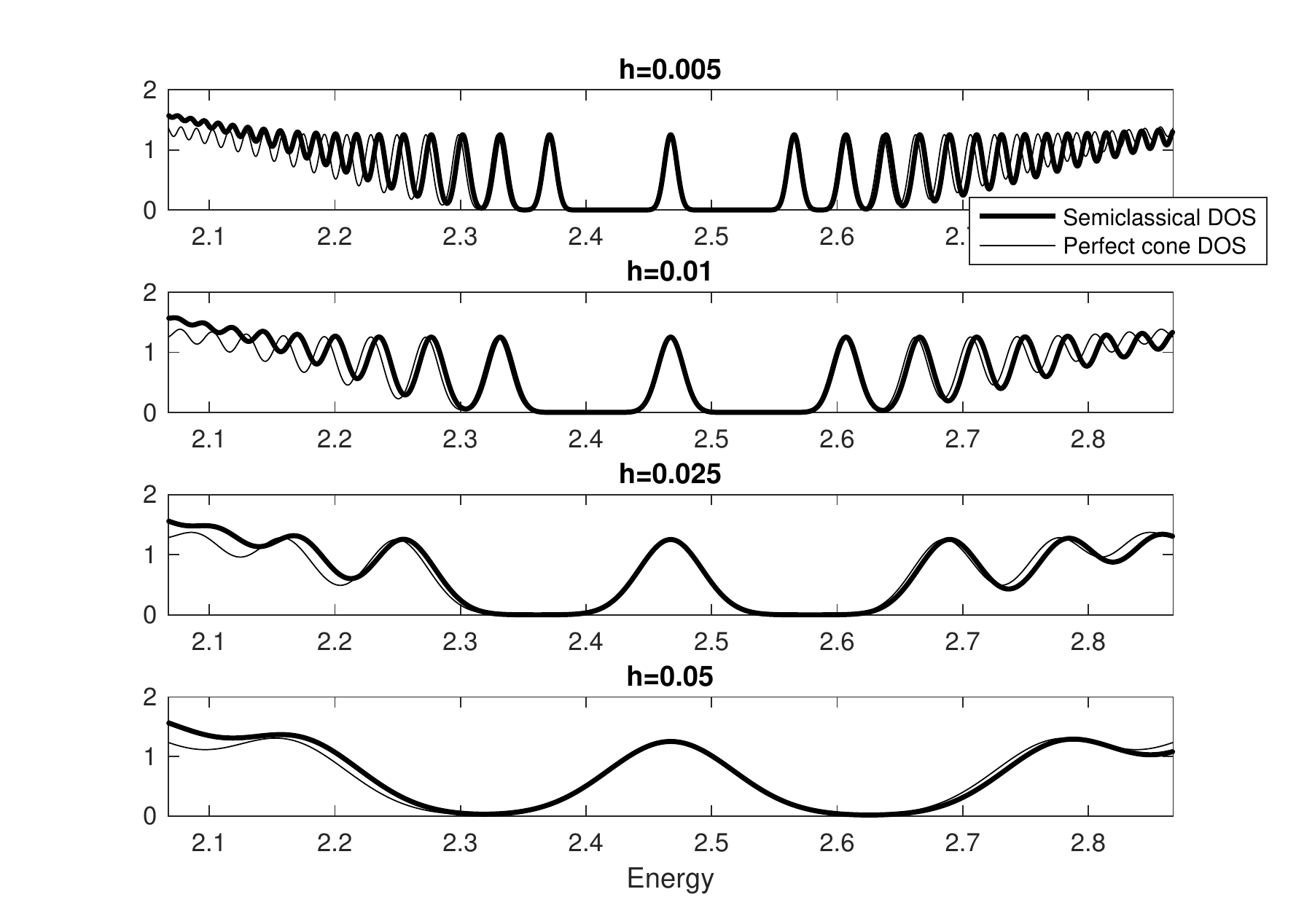} 
\caption{\label{fig:DOS4} The plots of $ \mu \mapsto \widetilde \rho_B ( \exp ( ( \bullet - \mu)^2 /2 \sigma^2 ) / \sqrt{ 2 \pi} \sigma ) $ for different values of $ h $ 
and $ \sigma = h $ (hence pushing the validity of \eqref{eq:Diracmeasures}; see also Figure \ref{fig:DOS1}). We note the asymmetry when compared to the 
density of states obtained using the perfect cone approximation \eqref{eq:gc}.}
\end{figure}

\subsection{De Hass--van Alphen oscillations}\label{dHvA}
As first discovered by de Haas and van Alphen in 1930, magnetization and magnetic susceptibility of
three dimensional metals 
oscillate as functions of $ 1/B $. They were not aware that Landau had just predicted presence of such oscillations. The frequencies are proportional to the areas of the extremal cross sections of the Fermi surface in the direction of the magnetic field. This explanation was provided by Onsager
\cite{O52} and a rigorous mathematical proof was given by Helffer and Sj\"ostrand 
\cite{HS2}. 

In the case of graphene,  the dHvA effect does not seem to be well understood neither experimentally nor theoretically
\cite{L11}. This is partly due to difficulties in accounting for all the parameters of the system: for instance, in the grand-canonical ensemble is frequently used to model the dHvA effect \cite{sgb}, the chemical potentials  are assumed to be independent of the external magnetic field. For a thorough discussion of this assumption, also made in this paper, we refer to \cite{CM01}. (We comment that the assumption of having a constant chemical potential is also assumed in the 3D Lifshitz-Kosevich theory \cite{KF17} for the study of magnetic oscillations in the susceptibility of metals at low temperatures. A 2D analogue of the theory for metals has been developed by Shoenberg \cite{S84} and was discussed in the context of graphene in \cite{L11}.)
 
Compared to previous discussions of magnetic oscillations -- see
for instance \cite{sgb} and \cite{L11} -- where the limit of infinitely many ``perfect cone" Landau levels was considered, we are only going to assume that there are finitely many semiclassically corrected Landau levels.
 
To introduce magnetization, we first define the {\em grand-canonical potential} at temperature $ T = 1/\beta $. Since we are interested in 
chemical potentials (energy) near the Dirac energy, we choose a smooth
function $ \eta \in C_{\rm{c}}^\infty ( I ) $ which is equal to $1 $
in a neighbourhood of the Dirac energy and replace $ \rho_B $ by
$ \eta \rho_B $ we then define
\begin{equation}
\label{eq:Omy}
\Omega_\beta ( \mu , h ) := \rho_B ( \eta ( \bullet ) f_\beta ( \mu - \bullet) ) , \ \  f_\beta ( x )  := - \beta^{-1} \log ( e^{ \beta x } +  1) .
\end{equation}
We note that $ f_\infty ( x ) = - x_+ $ and we define $ \Omega_\infty $ 
using that function. Since $ f_\infty $ is a Lipschitz function, 
Theorem \ref{theorem3} implies that
\begin{equation}
\label{eq:sepot}\begin{gathered}
\Omega_{\beta} ( \mu, h) = 
\frac{h}{\pi \left\lvert b_1 \wedge b_2 \right\rvert}\sum_{n \in \mathbb{Z}} f_{\beta}(\mu-z_n(h)) \eta(z_n(h))+ \mathcal{O}(h^{\infty}), 
\end{gathered}
\end{equation}
which holds true for $ \Omega_\infty $ defined using $ f_\infty = - x_+ $.
The function $ x \mapsto f_\beta ( \mu - x  ) $ is uniformly smooth away from $ x = \mu $. For $ \mu$'s near $ z_D $, changing $ \eta $ gives uniformly smooth 
 contributions (in $ \mu $ and $ h $) -- see Theorem \ref{t:smooth}. 

\begin{rem} 
The grand-canonical potential at non-zero temperatures (finite values of $ \beta$) can be recovered from $ \Omega_\infty $ using 
the Fermi distribution $n_{\beta}$:
\begin{equation}
\label{eq:conv}
\Omega_{\beta} ( \mu, h) = \left(-n_{\beta}'*\Omega_{\infty}(\bullet, h) \right)(\mu), \ \ n_\beta (x ) := ( 1 + e^{\beta x } )^{-1} .
\end{equation}
Indeed, we easily check that 
$ \left(-(\bullet-x)_{+}*n_{\beta}' \right)(\mu) = f_\beta ( x ) $. 
\end{rem}

{\em Magnetization} is defined as 
\begin{equation}
\label{eq:magn} M_{\beta}(\mu, h): =- \left\lvert b_1 \wedge b_2 \right\rvert \frac{\partial }{\partial h} \Omega_{\beta}(\mu, h).
\end{equation}
If we consider the full expansion of the levels $ z_n (h ) $ (obtained from $ F ( \omega, h ) $ in Proposition \ref{p:specQ0}) we could analyse 
$ M_\beta $ for $ \beta < h^{-M_0 } $ for any fixed $ M_0 $ -- see the remarks after \cite[Theorem 10.2]{HS2}. 

To avoid technical complications, 
we will instead, similarly to \cite{HS2}, consider {\em formal} magnetization obtained using leading term DOS, $ \widetilde \rho_B $, from 
\eqref{eq:Diracmeasures}. That already shows the {\em sawtooth} pattern
derived in \cite{sgb} using the ``perfect cone" approximation -- see Theorem \ref{t:sawtooth} and 
Figure \ref{Fig:magnmu}. Remarkably it also agrees with the ``exact" spectral numerical calculation explained in \S \ref{spectral} -- see 
Figure \ref{Fig:cut-off}. 

Let us now consider chemical potentials located on the upper cone of the first Hill band, i.e. $\mu \in \Big[z_D,\Delta\vert_{B_1}^{-1}\left(-\tfrac{1}3\right)\Big].$ {\em Formal} grand-canonical potential and {\em formal} magnetization are obtained from \eqref{eq:sepot} and \eqref{eq:magn} by replacing $(z_n(h))$ with the semiclassical Landau levels $(z_n^{(1)})$ given by the leading order Bohr--Sommerfeld condition \eqref{BSC}, and thus defined as follows
\begin{equation}
\label{eq:formal}
\begin{gathered}
\omega_\beta ( \mu, h ) :=
\frac{h}{\pi \left\lvert b_1 \wedge b_2 \right\rvert}\sum_{n \in \mathbb{Z}} f_{\beta}(\mu-z_n^{(1)} (h)) \eta(z^{1}_n(h)) , \\
m_\beta ( \mu, h ):=  - \left\lvert b_1 \wedge b_2 \right\rvert \frac{\partial }{\partial h} \omega_{\beta}(\mu, h), 
\end{gathered}
\end{equation}
and 
\begin{equation}
\label{eq:Theta}
\eta ( x) = \Theta_{\frac12}(x):= \begin{cases}
0  & x < z_D \\
\tfrac12  & x = z_D \\
1 & z_D < x < \Delta\vert_{B_1}^{-1}\left(-\tfrac{1}3\right)  
\\
0 & x \geq \Delta\vert_{B_1}^{-1}\left(-\tfrac{1}3\right) .
\end{cases}
\end{equation}
(This non-smooth $ \eta $ is convenient for spectral calculations and hence comparing semiclassical and exact numerics. The energy $\Delta\vert_{B_1}^{-1}\left(-\tfrac{1}3\right)$ corresponds to the energetic upper end of the upper cone.)

The construction for chemical potentials on the lower cone of the first Hill band, i.e. $\mu \in \Big[\Delta\vert_{B_1}^{-1}\left(\tfrac{1}3\right),z_D\Big],$ is similar. Using the cut-off function 
\[ \eta=\indic_{\left[\Delta\vert_{B_1}^{-1}\left(\frac{1}{3}\right),z_D \right]}\left(1 - \Theta_{\tfrac{1}{2}}\right), \]
 we obtain the semiclassical approximation from Landau levels located on the lower cone at zero temperature
\begin{equation}
\label{eq:lowercc}
\omega_{\infty}(\mu,h) :=  \tfrac{h}{\pi \left\lvert b_1 \wedge b_2 \right\rvert} \sum_{n \in \mathbb{Z}} (\mu-g^{-1}(nh))_{-} \eta(g^{-1}(nh)).
\end{equation}
We compare the oscillations on the upper \eqref{eq:formal} and lower cone \eqref{eq:lowercc} at zero temperature showing the asymmetry between the two different cones in Figure \ref{Fig:cut-off}.

\begin{figure}
\includegraphics[width=14cm]{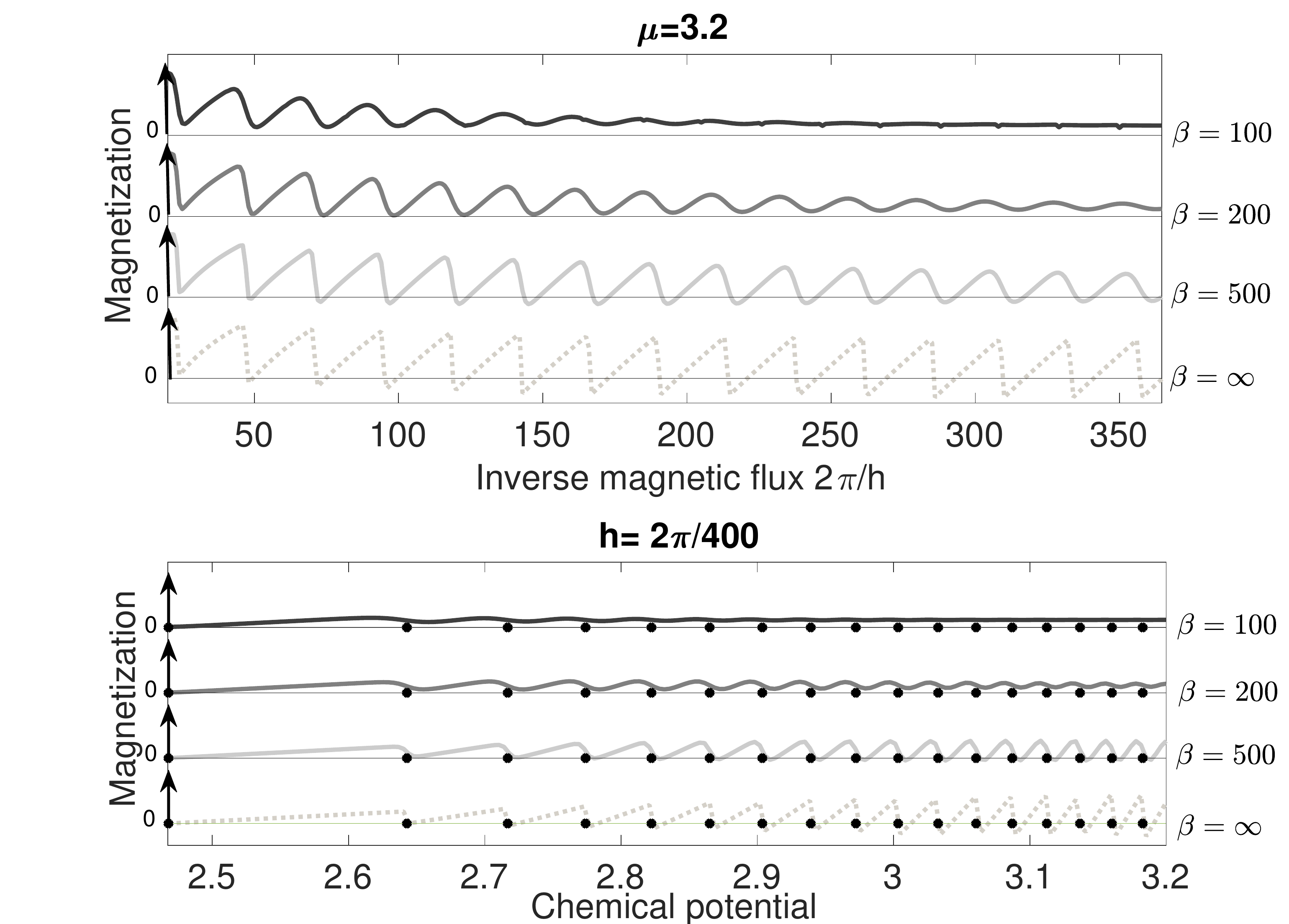}
  \caption{The magnetization \eqref{eq:formal} \label{Fig:magnmu} for different temperatures for a Hamiltonian with zero potential on the first Hill band. The sawtooth profile is clearly visible in the zero temperature limit $\beta= \infty$ and the oscillation period is approximately proportional to the inverse Fermi surface. As temperature increases, the oscillations become more smooth as predicted in \eqref{eq:conv} and the oscillation amplitude decreases. For zero temperature we see that the oscillation period increases linearly in $\mu$. This is no longer true when non-zero temperatures are considered.}
\end{figure}

The following asymptotic result shows the presence of ``sawtooth"
oscillations in magnetization. 

\begin{theo}
\label{t:sawtooth}
The formal magnetization for chemical potentials on the upper cone at zero temperature (defined in \eqref{eq:formal}) satisfies
\begin{equation}
\label{eq:omega3}
\begin{split} 
m_\infty ( \mu, h )  
= \frac{1}{\pi} \sigma \left( \frac{g ( \mu)} h\right) 
\frac{ g ( \mu) }{ g' ( \mu ) } + \mathcal O ( h^{\frac12} ), 
\end{split}  
\end{equation}
where $ g ( x ) = F_0 ( \Delta ( x ) ^2 ) $, with $ F_0 $ given in 
\eqref{eq:g2F},  is the leading term in the 
Bohr--Sommerfeld condition \eqref{BSC} and $ \sigma $ is the 
sawtooth function,
\begin{equation}
\label{eq:sigma} \sigma (y ) :=   y - [y] - \tfrac12 .
\end{equation}
\end{theo}
\begin{proof}  Since in \eqref{eq:formal} $ z_n^{(1)} ( h ) = g_+^{-1} (nh) $ (we drop $ + $ in what follows)
and $ \eta = \Theta_{\frac12} $,
\begin{equation}
\label{eq:omega1}
\omega_{\infty} ( \mu, h ) = - \frac{ h}{ \pi | b_1 \wedge b_2 | }
\left( \tfrac 12 (\mu - z_D) + \sum_{ n \geq 1 } ( \mu - g^{-1} ( nh) )_+ \right). 
\end{equation}
We rewrite the sum appearing in \eqref{eq:omega1} as follows:
\begin{equation}
\label{eq:saw1}
\begin{gathered} \sum_{ n\geq 1 } ( \mu - g^{-1} ( n h ) )_+
=  - \tfrac12 h (\mu - z_D) + \int_{ z_D }^\mu \left( \frac{g ( x ) }{h}
-  \sigma \left( \frac{g ( x ) } h \right) \right), 
\end{gathered} \end{equation}
where $ \sigma $ is defined by \eqref{eq:sigma}.
In fact, both sides are $ 0 $ at $ \mu = \Delta\vert_{B_k}^{-1} ( 0 ) $ and the derivative of the left hand side is
\[  \sum_{ n \geq 0 } ( \mu - g ^{-1} ( n h) )_+^0 =  \left[ \frac{g(\mu) }h \right] = \frac{g ( \mu )}h -  \sigma \left( \frac{g ( x ) } h \right) - \tfrac 1 2  . \]
 This gives the following expression for $ \omega$:
\begin{equation}
\label{eq:omega2}
\begin{split} 
\omega_{\infty} ( \mu, h ) & = -  \frac1 { \pi | b_1 \wedge b_2 | } 
\int_{ \Delta\vert_{B_k}^{-1} ( 0 ) }^\mu \left( g ( x ) 
- h \sigma \left( \frac{g ( x ) } h \right) \right) \\
& = G ( \mu ) + \frac {h^2} { \pi | b_1 \wedge b_2 | } \int_0^{ g( \mu)/h} 
\sigma ( z ) ( g^{-1} )' ( zh ) dz , 
\end{split}
\end{equation}
where $ G ( \mu ) $ is independent of $ h $. Hence, 
\[ \begin{split}  m_{\infty}( \mu , h ) -  \frac{1}{\pi} \sigma \left( \frac{g ( \mu)} h\right) 
\frac{ g ( \mu) }{ g' ( \mu ) } & =  h \int_0^{g(\mu)/h} 
\sigma ( z  ) \left( (g^{-1})'' ( zh ) zh + 2 ( g^{-1})' ( z h ) \right) dz
\\
& =  h^{\frac12}   \int_{0}^{{g(\mu)}/{h}}   {\sigma\left({  z}\right) } { z^{- \frac12} } a ( z h )   dz  , \end{split} \]
where 
$ a ( \xi ) :=  (g^{-1})''(\xi ) \xi^{\frac32}+ 2 (g^{-1})'(\xi ) \xi^{\frac12} $. 
The function $ a $ is smooth since $ g ( x ) = (  G^{-1} ( x - \Delta\vert_{B_k}^{-1} (0)) )^2 $ where
$ G ( 0 ) = 0 $, $ G'(0) \neq 0 $. That means that $ g^{-1} ( \xi ) = 
\Delta\vert_{B_k}^{-1} ( 0 ) + \xi^{\frac12} \varphi( \xi) $, $ \varphi \in C^\infty $ so that
$ a ( \xi ) = \frac34 \varphi ( \xi ) + 3\xi \varphi'( \xi) + \xi^2 \varphi''( \xi) \in C^\infty $.
We then write
\begin{equation}
\label{eq:h12}  \int_{0}^{{g(\mu)}/{h}}   {\sigma\left({  z}\right) } { z^{- \frac12} } a ( z h )   dz = h^{\frac12} \sum_{ n=0}^{ [ g ( \mu)/h ] - 1 } 
\int_{n}^{n+1}   {\sigma\left({  z}\right) } { z^{- \frac12} } a ( z h )   dz + \mathcal O ( h^{\frac12} ) .\end{equation}
For $ 1 \leq n \leq c/h $, 
\[\begin{split} \int_{n}^{n+1}   {\sigma\left({  z}\right) } { z^{- \frac12} } a ( z h )   dz & = 
\int_0^1 \sigma (   z ) ( z + n )^{-\frac12} a ( h ( z + n ) ) d z \\
& = n^{-\frac12} \int_0^1 \sigma (   z ) ( 1 + z/n )^{-\frac12} a ( n h ( 1 + z/ n ) ) d z \\
& = n^{- \frac12} a ( n h )  \int_0^1 \sigma (   z ) dz + 
\mathcal O ( n^{-\frac32} ) = \mathcal O ( n^{ -\frac32} ).
\end{split} \]
Hence the sum on the right hand side of \eqref{eq:h12} is bounded and
that concludes the proof of \eqref{eq:omega3}.
\end{proof}

The leading term in \eqref{eq:omega3} encapsulates the classical features of the dHvA effect: the function $\sigma(x)$ is periodic and its jump discontinuities coincide with the location of the Landau levels visible as the valleys in the lower Figure \ref{Fig:magnmu}.  The sawtooth profile shown in Figures \ref{Fig:fullspec} and \ref{Fig:cut-off}  of the oscillations agrees with the results obtained in \cite{sgb} and \cite{CM01} in which a sawtooth shape for magnetic oscillations in graphene was predicted. The quantity $g(\mu)$ is precisely the area enclosed by the Fermi curve as in the description of dHvA effect given by Onsager \cite{O52}.  In particular, this shows that the dHvA effect can be used as a test to study deviations from the perfect cone shape in graphene.  Finally, the scaling factor ${g(\mu)}/{g'(\mu)}$ implies a (at leading order) linear growth of the magnetic oscillations as a function of the chemical potential shown in Figure \ref{Fig:magnmu}.

\subsection{A Spectral approach to magnetic oscillations}
\label{spectral}
It is well known that when the magnetic flux $h$ satisfies $ h/ 2 \pi \in \mathbb Q $, modified Floquet theory can be used to describe the spectrum of 
$ H^B $ and the density of states. In particular, when 
$ h =\frac{2\pi p}{q}$, $ p, q \in \NN $, then 
the Floquet spectrum as a function of quasi-momentum $ k $ can be 
calculated using $ 2q \times 2q $ matrices -- see \cite{BHJ17}. 

More precisely, for $ k \in \TT^2_* $ we follow \cite{BHJ17} and define 
\begin{equation}
\label{eq:Tq}
\begin{gathered} 
T_q(k):=\tfrac{1}{3}\left( \begin{matrix} 0 & \operatorname{id}_{\mathbb{C}^q} + e^{ik_1} J_{p,q}+ e^{ik_2} K_q \\ \operatorname{id}_{\mathbb{C}^q}  + e^{-ik_1} J_{p,q}^*+ e^{-ik_2} K_q^* & 0 \end{matrix} \right) 
\end{gathered}
\end{equation}
where 
\[ 
 (J_{p,q})_{ j \ell } = e^{\frac{ 2 \pi p} q i(\ell-1)} \delta_{j \ell}, \ \    (K_q)_{j\ell}= \left\{ \begin{array}{ll} 1  & \ell  \equiv j+1 \mod q \\
0 & \text{ \ \ \ otherwise}, \end{array} \right.  
\ \ \ 
 1 \leq j, \ell \leq q .\]
  Then $\lambda \in \operatorname{Spec}(H^B) \backslash \operatorname{Spec}(H^D) $ if and only if $\Delta(\lambda) \in \bigcup_{k \in \mathbb{T}_*^2}\operatorname{Spec}(T_q(k))$. 
Thus, on each Hill band $H^B$ has $2q$ non-overlapping bands that touch at the conical point.  In particular, there are $q$ bands above and below the conical point.

The density of state is given in the following
\begin{lemm}
\label{spectralregtra}
Let $h={2\pi p}/{q}$ then for any $f \in C_c(\mathbb{R} \setminus 
\Spec ( H^D ) )$
\begin{equation}
\widetilde{\operatorname{tr}}(f(H^B)) = \frac{1}{q \left\lvert b_1 \wedge b_2 \right\rvert} \int_{\mathbb{T}_*^2} \sum_{\Delta( \lambda ) \in \Spec ( T_q ( k ) ) } f(\lambda) \frac{dk }{\left\lvert \mathbb{T}_{*}^2 \right\rvert}.
\end{equation}
\end{lemm}
\begin{proof}
Since the flux is of the form $h=\frac{2\pi p}{q}$ there is a fundamental cell $W_{\Lambda}^B$ of measure $q \left\lvert b_1 \wedge b_2 \right\rvert$ with respect to which the operator $H^B$ is translational invariant \cite{BHJ17}. Thus, along the lines of Lemma \ref{existtra} we find that
\begin{equation}
\widetilde{\operatorname{tr}}(f(H^B)) = \tfrac{1}{q \left\lvert b_1 \wedge b_2 \right\rvert} \operatorname{tr}(\indic_{W_{\Lambda}^B}f(H^B)).
\end{equation}
By Floquet theory, $f(H^B)$ is unitary equivalent to the bounded decomposable operator $ \int_{\mathbb{T}_*^2}^{\oplus} f(H^B)(k) \frac{dk}{\left\lvert \mathbb{T}_*^2 \right\rvert}$ such that for any orthonormal basis,  $\{ \varphi_n\} _{n \in \mathbb{N}}$,  of $L^2(W_{\Lambda}^B)$
\begin{equation}
\begin{split}
\widetilde{\operatorname{tr}}(f(H^B)) &= \tfrac{1}{q \left\lvert b_1 \wedge b_2 \right\rvert} \operatorname{tr} \indic_{W_{\Lambda}^B} f(H^B) \\
							&=  \tfrac{1}{q \left\lvert b_1 \wedge b_2 \right\rvert} \sum_{n \in \mathbb{N}}\left\langle \varphi_n, f(H^B)(k) \varphi_n \right\rangle_{L^2\left(\mathbb{T}_*^2, \tfrac{dk}{\left\lvert \mathbb{T}_*^2 \right\rvert}\right) \otimes L^2(W_{\Lambda}^B)}\\
							& = \tfrac{1}{q \left\lvert b_1 \wedge b_2 \right\rvert} \int_{\mathbb{T}_*^2} \operatorname{tr}_{L^2(W_{\Lambda}^B)} f(H^B)(k)  \frac{dk}{\left\lvert \mathbb{T}_*^2 \right\rvert}  \\
							& =\tfrac{1}{q \left\lvert b_1 \wedge b_2 \right\rvert} \int_{\mathbb{T}_*^2} \sum_{\lambda \in \operatorname{Spec}(H^B(k))} f(\lambda) \frac{dk }{\left\lvert \mathbb{T}_{*}^2 \right\rvert}.
\end{split}
\end{equation}
Away from $ \Spec ( H^D ) $ the spectrum of $ H^B ( k ) $ is characterized by $ \Delta ( \lambda ) \in \Spec ( T_q ( k ) )$ and \eqref{spectralregtra} follows.
\end{proof}
In the semiclassical regime $h \to 0$,  the location of the energy bands of  $ \Spec ( H^B) $ coincides with the location of the semiclassical Landau levels close to the conical point. 
By using the actual spectrum of $H^B$, the broadening of the Landau levels, known as \emph{Harper broadening} \cite{KH14}, is already part of the model and does not have to be approximated as in \cite{sgb} or \cite{CM01}.
We should stress that Lemma \ref{l:awayspec} shows that the width of the bands is $ O ( h^\infty ) $ and finer analysis of \cite{HS0} could be used to show that the width is in fact $ O ( e^{ - c/h} ) $.

The advantage of the representation of the density of states in Lemma \ref{spectralregtra} is that we can calculate DOS numerically for larger values of $ h $, that is, for {\em strong magnetic fields.}
This approach is similar to the study of magnetic oscillations in the tight-binding model presented in \cite{KH14}.

Let the magnetic flux be of the form $h=2\pi{ p}/{q}$ with $p \in \mathbb{Z}, q \in \mathbb{N}$, then we study the grand-canonical potential localized to the spectrum on the first Hill band which by Lemma \ref{spectralregtra} satisfies
\begin{equation}
\label{eq:gc1}
\begin{split} 
\Omega_{\beta} ( \mu, h) & := (f_{\beta}* \eta\rho_B)(\mu) \\
& =-\tfrac{1}{q \left\lvert b_1 \wedge b_2 \right\rvert} \tfrac{1}{\beta}\int_{\mathbb{T}_*^2} \sum_{ \Delta\vert_{B_1} ( \lambda ) \in \Spec ( T_q ( k ) )} \log \left( \operatorname{exp} \left( \beta (\mu- \lambda ) \right)+1 \right)  \tfrac{dk}{\left\lvert \mathbb{T}_*^2 \right\rvert}.
\end{split} 
\end{equation}
Here $\eta$ is one on $\mathcal A:=\Delta\vert_{B_1}^{-1}( \bigcup_{k \in \mathbb{T}_*^2}  \Spec ( T_q ( k ) )$ and zero on $\Spec ( H^B )\backslash \mathcal A.$
In the zero temperature limit this reduces to
\begin{equation}
\label{eq:gc2}
\Omega_{\infty} ( \mu, h)= -\tfrac{1}{q \left\lvert b_1 \wedge b_2 \right\rvert} \int_{\mathbb{T}_*^2} \sum_{ \Delta\vert_{B_1} ( \lambda ) \in \Spec ( T_q ( k ) )}  \left(\mu- \lambda \right)_{+}  \tfrac{dk}{\left\lvert \mathbb{T}_*^2 \right\rvert}. 
\end{equation}
Remarkably, $\Omega_{\infty}$ satisfies $\Omega_{\infty}(\mu,h)=(f_{\infty}* \rho_B)(\mu)$ without any cut-off for $\mu < \inf \Spec ( H^D ).$
The definition of the grand-canonical potential used here coincides with the expression in \cite{GA03} up to the regularizing pre-factor $ (q \left\lvert b_1 \wedge b_2 \right\rvert)^{-1} .$

\begin{figure}
   \includegraphics[width=13cm]{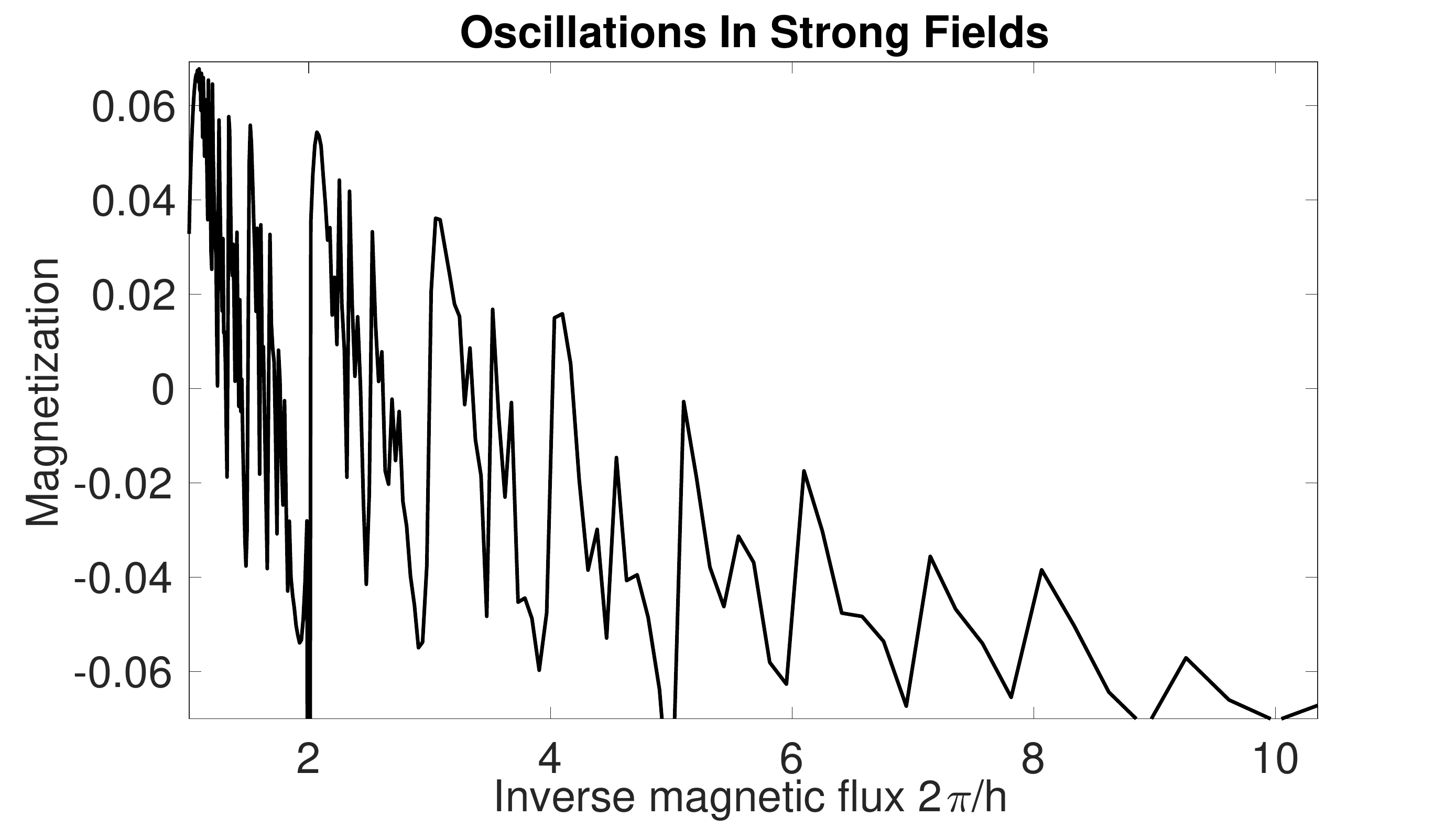} 
  \caption{The magnetization \eqref{eq:magn} \label{Fig:magn} for a Hamiltonian with zero Hill potential at $ \mu=\frac{\pi^2}{4}$ 
  given by the Dirac energy.  The magnetization shows a decaying inverted saw-tooth profile with oscillations in $ 1/h $ and additional high-frequency modulations. As $ 1/h $ increases we move to the semiclassical regime in which no oscillations occur at the Dirac energy -- see Figure \ref{Fig:fullspec}. }
\end{figure}

{\em Magnetization} is defined by \eqref{eq:magn} and we compute it numerically for \eqref{eq:gc2} using finite difference approximation at rational points.
Results for computation using difference quotients for magnetic fluxes $h=2\pi \frac{p}{150}$ and $p \in \left\{1,...,150\right\}$ are shown in
Figure \ref{Fig:magn}.
The results we obtain are in good agreement with the oscillations obtained in \cite{KH14}. The magnetization shows a decaying inverted saw-tooth profile with oscillations in $ 1/h $ and additional high-frequency modulations. These type of magnetic oscillations are an effect of \emph{strong magnetic fields}. Unlike the dHvA oscillations discussed in \S \ref{dHvA}, the magnetization for such strong magnetic fields deviates significantly from the semiclassical approximation. In particular, the characteristic oscillatory profile caused by the strong magnetic field decreases for sufficiently small magnetic fluxes as we see in Figure \ref{Fig:magn}. Moreover, there are no oscillations when the chemical potential agrees with the energy of the Dirac point in the semiclassical limit. 

Figure \ref{Fig:fullspec} shows the magnetization \eqref{eq:magn}
computed using \eqref{eq:gc2} for values of $ h $ in the semiclassical 
regime.

\begin{figure}
  \centering
   \includegraphics[width=14cm]{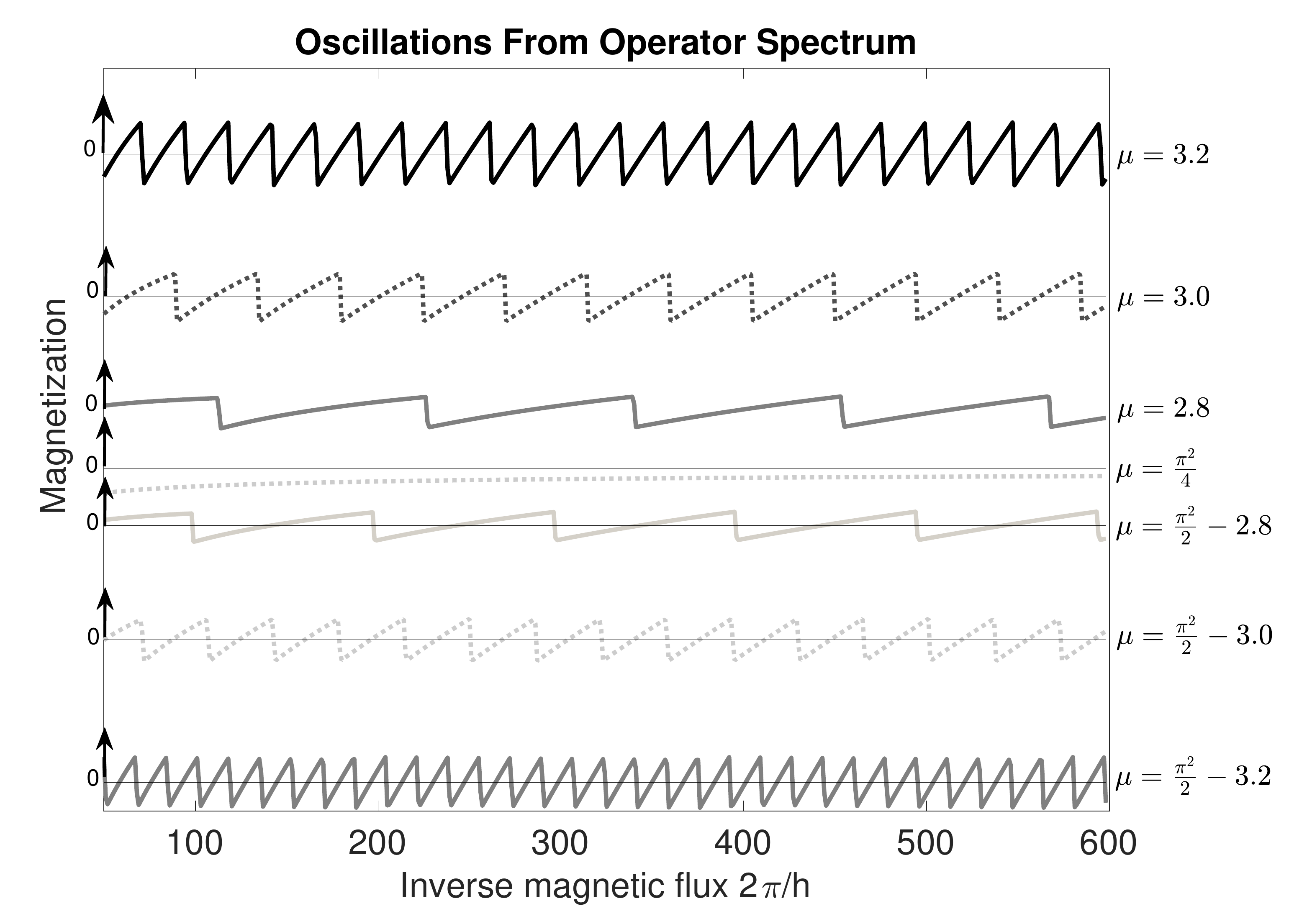} 
\caption{\label{Fig:fullspec}The magnetization \eqref{eq:magn} 
as a function of the inverse flux for specific chemical potentials ($\mu=\frac{\pi^2}{4}$ is the location of the Dirac point) for a magnetic Hamiltonian with zero potential at zero temperature. 
It is computed using the spectral method \eqref{eq:gc2}. 
The magnetization for all chemical potentials is true to scale and calculated at zero temperature from the full operator spectrum (i.e. no cut-off is used). We calculated the magnetization for inverse fluxes $\tfrac{2\pi}{q}$ with $q \in \left\{10,..,600\right\}.$ One clearly sees the antisymmetry between the different magnetic oscillations with respect to the conical point. The figures show (away from the Dirac point) jump discontinuities caused by the crossing of chemical potential and Landau levels.}
\end{figure}

\begin{figure}
  \centering
   \includegraphics[width=13cm]{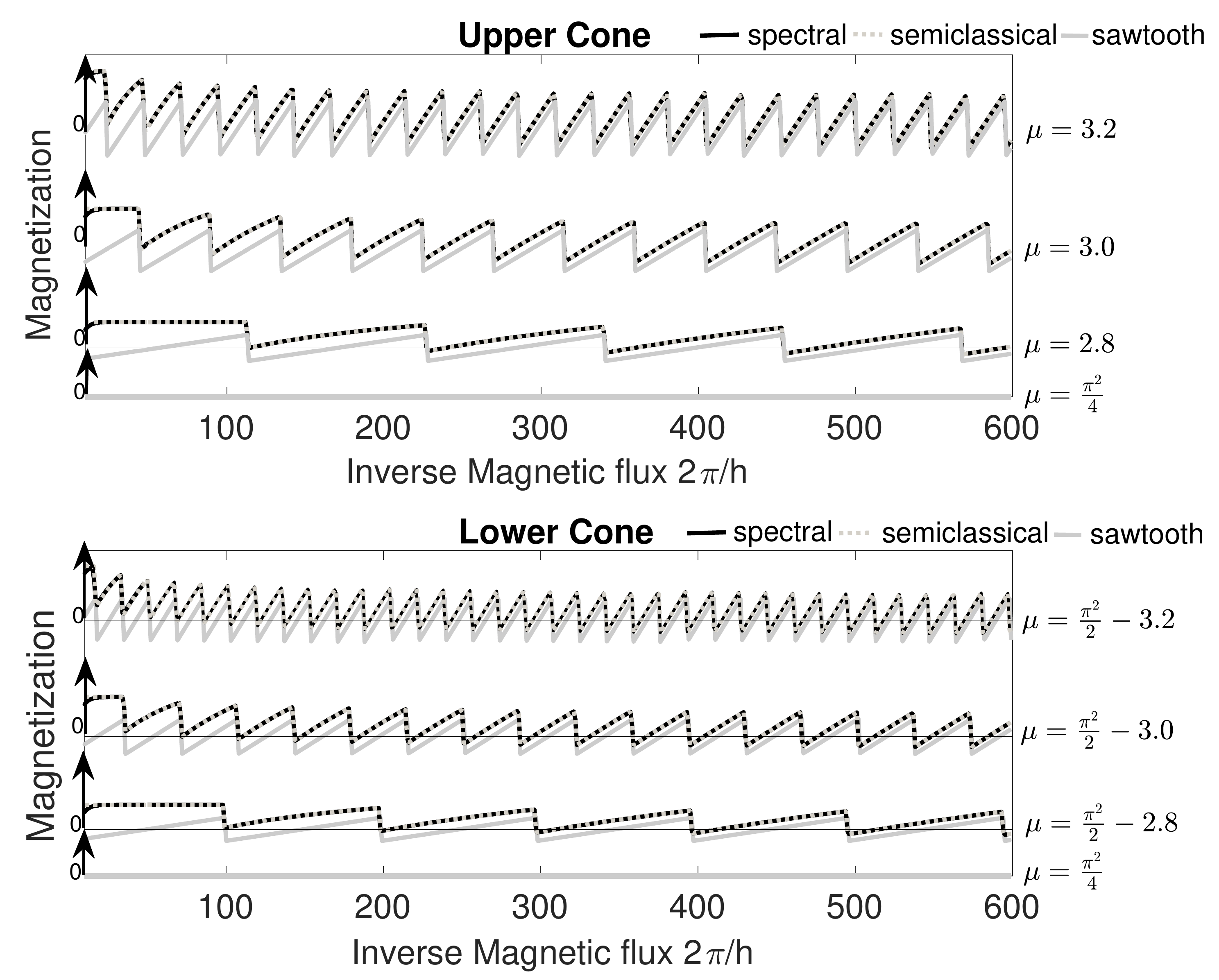} 
\caption{\label{Fig:cut-off} The magnetization  
for four different chemical potentials {\em above/below} the Dirac point located at $\mu=\tfrac{\pi^2}{4}$ on the first Hill band of the magnetic Hamiltonian with zero potential at zero temperature. Continuous lines are computed from the operator spectrum precisely by numerically differentiating \eqref{eq:lowerc} and dotted lines (grey) from numerically differentiating the semiclassical expression  \eqref{eq:lowercc}.  We evaluated those expressions for steps $\tfrac{2\pi}{q}$ with $q \in \left\{10,..,600\right\}.$ The magnetization for all chemical potentials is true to scale. Both the spectral and semiclassical oscillations show (away from the Dirac point) equally spaced jump discontinuities caused by the crossing of chemical potential and Landau levels. We see that both oscillations coincide up to large magnetic fields (small values of 
 $ 1/ h$). }
\end{figure}

\subsection{ Comparing spectral and semiclassical calculations }
\label{compa}

We now compare the exact spectral calculations at 
magnetic fluxes of the form $h= 2\pi {p}/{q}$ \eqref{eq:gc1} with the results obtained from the semiclassical trace formula \eqref{eq:formal} where we approximate $z_h(n)$ in \eqref{eq:sepot} by $z_h^{(1)}(n):=g^{-1}(nh)$. 

As explained in \S \ref{spectral}, the  spectrum of $ H^B $ away from the Dirichlet spectrum of $H^D$ is fully determined by the eigenvalues of $T_q$ as in \eqref{eq:Tq}. This matrix has for every quasi-momentum $k \in \mathbb{T}_{*}^2$ precisely $2q$ eigenvalues $\lambda_1(k)\le...\le \lambda_{2q}(k)$, of which, when pulled back under $\Delta \vert_{B_k}$, precisely half are located below and above the conical point $\Delta \vert_{B_k}^{-1}(0)$. Moreover, it is easy to see that there are always two touching bands at the conical point as discussed in \cite{BHJ17} and \cite{HKL16}. 

For chemical potentials $\mu \in \left[z_D,\Delta\vert_{B_1}^{-1}\left(-\tfrac{1}3\right)\right]$ on the upper cone of the first Hill band we define the grand-canonical potential calculated from DOS of the operator spectrum as in Lemma \ref{spectralregtra}
\begin{equation}
\label{eq:usedgcpot}
\Omega_{\infty} ( \mu, h)= (f_{\infty}*\eta \rho_{B})(\mu)= -\tfrac{1}{q \left\lvert b_1 \wedge b_2 \right\rvert} \int_{\mathbb{T}_*^2} \sum_{ i \in \{ 1, \cdots q \}}  \left(\mu- \Delta|_{B_1}^{-1} (\lambda_i ( k)  ) \right)_{+} \tfrac{dk}{\left\lvert \mathbb{T}_*^2 \right\rvert}. 
\end{equation}
This is the grand-canonical potential calculated from the operator spectrum which corresponds to the semiclassical potential \eqref{eq:Omy}
with $ \beta = \infty $ and $ \eta = \Theta_{\frac12} $ from \eqref{eq:Theta}. 

For $\mu \in \left[\Delta\vert_{B_1}^{-1}\left(\tfrac{1}{3} \right), z_D \right]$ on the lower cone of the first Hill band, we proceed similarly: in this case, the grand-conical potential, which is defined using the spectrum located between the chemical potential and the conical point, reads
\begin{equation}
\label{eq:lowerc} 
\Omega_\infty ( \mu, h) = \tfrac{1}{q \left\lvert b_1 \wedge b_2 \right\rvert} 
\int_{\mathbb{T}_*^2} \sum_{ i \in \{ q+1 , \cdots 2q \}}  \left(\mu- \Delta|_{B_1}^{-1} (\lambda_i ( k) )  \right)_{-} \tfrac{dk}{\left\lvert \mathbb{T}_*^2 \right\rvert}. 
\end{equation}
This potential is the spectral analogue of the semiclassical potential \eqref{eq:lowercc}.

We compare the computation of magnetization \eqref{eq:magn} calculated using
finite difference method from \eqref{eq:usedgcpot} and \eqref{eq:lowerc}  with the formal
semiclassical magnetizations from \eqref{eq:formal} and \eqref{eq:lowercc} on both cones. The results are shown in Figure
\ref{Fig:cut-off} and we see a remarkable agreement of the semiclassical approximation with the spectral computation. The sawtooth approximation given in Theorem \ref{t:sawtooth} is also shown.


\end{document}